\documentclass[10pt,reqno,oneside]{amsart}

\usepackage{cancel,bm, amssymb}
\usepackage{color}

\usepackage{color}
\usepackage{graphicx,framed}
\usepackage[colorlinks=true, pdfstartview=FitV, linkcolor=blue, citecolor=blue, urlcolor=blue]{hyperref}
\definecolor{shadecolor}{rgb}{0.95, 0.95, 0.86}

\renewcommand{\d}{{\mathrm d}}

\newcommand{\im}{\mathrm{i}}
\newcommand{\e}{\mathrm{e}}
\newcommand{\z}{\zeta}

\numberwithin{equation}{section}

\newtheorem{theo}{Theorem}[section]

\newtheorem{lem}[theo]{Lemma}
\newtheorem{rem}[theo]{Remark}
\newtheorem{problem}[theo]{Riemann-Hilbert Problem}
\newtheorem{remark}[theo]{Remark}
\newtheorem{prop}[theo]{Proposition} 
\newtheorem{cor}[theo]{Corollary}

\begin{document}

\title{From gap probabilities in random matrix theory to eigenvalue expansions}

\author{Thomas Bothner}
\address{Department of Mathematics, University of Michigan, 2074 East Hall, 530 Church Street, Ann Arbor, MI 48109-1043, United States}
\email{bothner@umich.edu}



\thanks{The author is grateful to R. Buckingham, P. Deift, A. Its and I. Krasovsky for stimulating discussions about this project. The revised version of this manuscript has benefited from the constructive comments made by anonymous referees for which we are grateful.}
\dedicatory{Dedicated to Percy Deift and Craig Tracy on the occasion of their 70th birthdays}

\begin{abstract}
We present a method to derive asymptotics of eigenvalues for trace-class integral operators $K:L^2(J;\d\lambda)\circlearrowleft$, acting on a single interval $J\subset\mathbb{R}$, which belong to the ring of integrable operators \cite{IIKS}. Our emphasis lies on the behavior of the spectrum $\{\lambda_i(J)\}_{i=0}^{\infty}$ of $K$ as $|J|\rightarrow\infty$ and $i$ is fixed. We show that this behavior is intimately linked to the analysis of the Fredholm determinant $\det(I-\gamma K)|_{L^2(J)}$ as $|J|\rightarrow\infty$ and $\gamma\uparrow 1$ in a Stokes type scaling regime. Concrete asymptotic formul\ae\, are obtained for the eigenvalues of Airy and Bessel kernels in random matrix theory.
\end{abstract}

\date{\today}
\maketitle

\section{Introduction and statement of results}
A well-known result in orthogonal polynomial random matrix ensembles states that the spacing distributions of eigenvalues of $N\times N$ Hermitian matrices are encoded in Fredholm determinants \cite{M,D}. In more detail, for a given positive weight function $w(t)$, the probability $E_N(n;J)$ that a matrix from the ensemble associated with $w(t)$ has $n\in\mathbb{Z}_{\geq 0}$ eigenvalues in $J\subset\mathbb{R}$ equals
\begin{equation}\label{EN:1}
	E_N(n;J)=\frac{(-1)^n}{n!}\frac{\d^n}{\d\gamma^n}\det(I-\gamma K_N)\big|_{\gamma=1}.
\end{equation}
Here, $K_N$ is the finite rank integral operator on $L^2(J;\d\lambda)$ with kernel
\begin{equation*}
	K_N(\lambda,\mu)=\sum_{j=0}^{N-1}p_j(\lambda)p_j(\mu)w^{\frac{1}{2}}(\lambda)w^{\frac{1}{2}}(\mu);\ \ \ \ \ \ \ p_j\in\mathbb{C}[t]:\ \ \int p_j(t)p_k(t)w(t)\d t=\delta_{jk}.
\end{equation*}
The Fredholm representation of eigenvalue spacing statistics is also valid in the limit $N\rightarrow\infty$ and the resulting trace-class integral operator $K:L^2(J;\d\lambda)\circlearrowleft$ depends crucially on the Hermitian model we start with and at which local point in the spectrum the scaling is considered.\smallskip

For instance, two of the most commonly encountered kernels in random matrix theory arise in the Gaussian Unitary Ensemble (GUE) by scaling in the ``bulk" of the spectrum \cite{D,P,WN}, resp. at the ``soft edge" \cite{BB,F,Mo}: 
\begin{eqnarray}
	K_{\sin}(\lambda,\mu)&=&\frac{\sin(\lambda-\mu)}{\pi(\lambda-\mu)},\ \ \ \ \ J_{\sin}=(-s,s),\ \ s>0;\ \ \ \ \ \ \ \ \textnormal{resp.}\label{sinek}\\
	 K_{\textnormal{Ai}}(\lambda,\mu)&=&\frac{\textnormal{Ai}(\lambda)\textnormal{Ai}'(\mu)-\textnormal{Ai}'(\lambda)\textnormal{Ai}(\mu)}{\lambda-\mu},\ \ \ \ \ J_{\textnormal{Ai}}=(s,\infty),\ \ s\in\mathbb{R};\label{Aik}
\end{eqnarray}
i.e. the {\it sine kernel} and the {\it Airy kernel}, where $\textnormal{Ai}(z)$ is the Airy function. A third well-known kernel appears when one scales the Laguerre or Jacobi Unitary Ensembles (LUE or JUE) at the ``hard edge" \cite{F,WN}:
\begin{equation}\label{Bessk}
	K_{\textnormal{Bess}}^{(a)}(\lambda,\mu)=\frac{J_a(\sqrt{\lambda})\sqrt{\mu}J_a'(\sqrt{\mu})-J_a(\sqrt{\mu})\sqrt{\lambda}J_a'(\sqrt{\lambda})}{2(\lambda-\mu)},\ \ \ \ \ J_{\textnormal{Bess}}=(0,s),\ \ s>0;\ \ \ a>-1,
\end{equation}
i.e. the {\it Bessel kernel}, where $J_a(z)$ is the Bessel function of order $a$. From the random matrix point of view one is interested in the analogue of \eqref{EN:1},
\begin{equation*}
	E(n;J)=\frac{(-1)^n}{n!}\frac{\d^n}{\d\gamma^n}D(J;\gamma)\big|_{\gamma=1};\ \ \ \ \ \ \ \ D(J;\gamma)=\det(I-\gamma K)\big|_{L^2(J)}
\end{equation*}
which equals the probability that there are $n\in\mathbb{Z}_{\geq 0}$ bulk scaled ($K=K_{\sin}$), resp. soft-edge scaled ($K=K_{\textnormal{Ai}}$), resp. hard-edge scaled ($K=K_{\textnormal{Bess}}^{(a)}$) eigenvalues in the interval $J_{\sin}$, resp. $J_{\textnormal{Ai}}$, resp. $J_{\textnormal{Bess}}$.\footnote{Equivalently, $E(0;J_{\textnormal{Ai}})$ and $E(0;J_{\textnormal{Bess}})$ are the distribution functions of the (rescaled) largest and smallest eigenvalue drawn from the GUE and LUE.}\bigskip

We shall denote the eigenvalues of the integral operator $K$ by $\{\lambda_i(J)\}_{i=0}^{\infty}$. In case of the three kernels \eqref{sinek}, \eqref{Aik} and \eqref{Bessk} it has been proven \cite{PS,Fu,TW,TW2} that the spectrum is simple and we can order
\begin{equation}\label{cascade}
	1>\lambda_0(J)>\lambda_1(J)>\ldots>0.
\end{equation}
The lower bound follows from positivity and the upper one from the fact that $K:L^2(\mathbb{R};\d\lambda)\circlearrowleft$ with $K=K_{\sin},K_{\textnormal{Ai}}$ and $K:L^2((0,\infty);\d\lambda)\circlearrowleft$ with $K=K_{\textnormal{Bess}}^{(a)}$ are projection operators (cf. \cite{PS,TW,TW2}). In particular, the last statement implies by minimax characterizations that for each fixed $i\in\mathbb{Z}_{\geq 0}$ we have
\begin{equation}\label{minimax}
	\lambda_i(J)\rightarrow 1,\ \ \ \textnormal{as}\ \ \ |J|\rightarrow+\infty\ \ \ \ \Leftrightarrow\ \ \ \ \lambda_i(J)\rightarrow 1,\ \ \ \textnormal{as}\ \ \begin{cases} s\rightarrow+\infty,&K=K_{\sin},K_{\textnormal{Bess}}^{(a)}\\ s\rightarrow-\infty,&K=K_{\textnormal{Ai}}\end{cases}.
\end{equation}
The eigenvalue behavior \eqref{cascade} provides us with a quick qualitative explanation of the phase transition near $\gamma=1$ in the asymptotic behavior of $D(J;\gamma)$ as $|J|\rightarrow\infty$, cf. \cite{W2}. Indeed, use Lidskii's Theorem and write
\begin{equation}\label{Lidskii}
	D(J;\gamma)=\det(I-\gamma K)\big|_{L^2(J)}=\prod_{i=0}^{\infty}\big(1-\gamma\lambda_i(J)\big),
\end{equation}
so that (1) for $\gamma<1$ all factors are bounded away from zero. Hence $D(J;\gamma)$ approaches zero exponentially fast. On the other hand (2) for $\gamma=1$ the $i$-th factor will approach zero for each $i$ and $D(J;\gamma)$ tends to zero faster than for $\gamma<1$. In the remaining case (3) for $\gamma>1$ the determinant will vanish for a discrete set $\{J_i\}_{i\in\mathbb{Z}}$, the solutions of $\lambda_i(J)=\gamma^{-1}$.\smallskip

As another application, the eigenvalue asymptotics \eqref{minimax} can be used to derive quantitative asymptotic information for the normalized eigenvalue probabilities
\begin{equation*}
	r(n;J)\equiv\frac{E(n;J)}{E(0;J)}=\sum_{i_1<\ldots<i_n}\frac{\lambda_{i_1}\cdot\ldots\cdot \lambda_{i_n}}{(1-\lambda_{i_1})\cdot\ldots\cdot(1-\lambda_{i_n})},\ \ n\in\mathbb{Z}_{\geq 1},
\end{equation*}
i.e. \eqref{minimax} serves in the exact evaluation of ``large gap probabilities", see \cite{TW,TW2} for concrete formul\ae. For this to work one requires subleading terms in \eqref{minimax} and a standard way to obtain these is the method of commuting differential operators: by a remarkable coincidence the three integral operators on $L^2(J;\d\lambda)$ with kernels \eqref{sinek} and \eqref{Aik}, \eqref{Bessk} commute with the operators given by
\begin{equation*}
	\mathcal{L}_{\sin}[f]=\left(\frac{\d}{\d x}\big(x^2-s^2\big)\frac{\d}{\d x}+x^2\right)f,
\end{equation*}
and
\begin{equation}\label{TWcomm}
	\mathcal{L}_{\textnormal{Ai}}[f]=\left(\frac{\d}{\d x}(x-s)\frac{\d}{\d x}-x(x-s)\right)f,\ \ \ \ \ \ \ \ 
\mathcal{L}_{\textnormal{Bess}}^{(a)}[f]=\left(\frac{\d}{\d x}x(s-x)\frac{\d}{\d x}-\left(\frac{a^2s}{4x}+\frac{x}{4}\right)\right)f;
\end{equation}
defined on appropriate function spaces. Since $(K,\mathcal{L})$ share the same eigenfunctions the desired eigenvalue expansion is then obtained from WKB arguments applied to the differential equation. This is exactly the approach employed by Pollak, Slepian and Fuchs \cite{PS,Fu} for $K_{\sin}$: the eigenfunctions of $\mathcal{L}_{\sin}$ are prolate spheroidal wave functions, cf. \cite{NIST}, and the following asymptotics was derived,
\begin{equation}\label{sinres}
	1-\lambda_i(J_{\sin})=\frac{\sqrt{\pi}}{i!}2^{3i+2}s^{i+\frac{1}{2}}\e^{-2s}\left(1+\mathcal{O}\left(s^{-1}\right)\right),\ \ s\rightarrow+\infty,
\end{equation}
valid for fixed $i\in\mathbb{Z}_{\geq 0}$.\smallskip

In case of \eqref{TWcomm} no explicit formul\ae\, for eigenfunctions are known, still Tracy and Widom \cite{TW,TW2} used the operators \eqref{TWcomm} in combination with a trick and derived analogues of \eqref{sinres}, see Corollaries \ref{Aispecex} and \ref{Bessspecex} below. This derivation was somewhat heuristic and it seems to work well only for Sturm-Liouville operators $\mathcal{L}$ with rather simple potentials. For more general kernels in random matrix theory, compare Section \ref{singker} for a short discussion, the potentials would be transcendental functions of Painlev\'e type and the method fails.\smallskip

For this reason we choose a different approach to the asymptotics of $\lambda_i(J)$ as $|J|\rightarrow\infty$ which will only rely on the integrable structure \cite{IIKS} of the kernels, i.e. them being of the form
\begin{equation}\label{IIKSker}
	K(\lambda,\mu)=\frac{\phi(\lambda)\psi(\mu)-\psi(\lambda)\phi(\mu)}{\lambda-\mu}.
\end{equation}
The idea is to obtain ``large gap asymptotics" for $D(J;\gamma)$ 
as $|J|\rightarrow\infty$ and {\it simultaneously} $\gamma\uparrow 1$ in a specific scaling regime. This regime has to be chosen in such a way that we can observe individual factors of the infinite product \eqref{Lidskii}, i.e. preferably we would like to produce an asymptotic expansion of $D(J;\gamma)$ which involves a finite product with distinct factors. 
\begin{rem}
It is not clear at all how to read off individual factors from the expansions of $D(J;\gamma)$ as $|J|\rightarrow\infty$ and $\gamma\leq 1$ is kept fixed: For instance, for the sine-kernel determinant \cite{BW,BuBu,K,E0,DIKZ}, as $s\rightarrow+\infty$,
\begin{eqnarray}
	D(J_{\sin};\gamma<1)&=&\exp\left[-\frac{2v}{\pi}s\right](4s)^{\frac{v^2}{2\pi^2}}G^2\left(1+\frac{\im v}{2\pi}\right)G^2\left(1-\frac{\im v}{2\pi}\right)\big(1+\mathcal{O}\left(s^{-1}\right)\big),\ \ \ \ \ v=-\ln(1-\gamma);\nonumber\\
	D(J_{\sin};1)&=&\exp\left[-\frac{s^2}{2}\right]s^{-\frac{1}{4}}c_0\big(1+\mathcal{O}\left(s^{-1}\right)\big),\ \ \ \ \ \ \ \ \ c_0=\exp\left[\frac{1}{12}\ln 2+3\z'(-1)\right];\label{DyWidgap}
\end{eqnarray}
involving Barnes $G$-function $G(\cdot)$ and the Riemann zeta-function $\z(\cdot)$. 
\end{rem}
In order to motivate the usefulness of the phase transition point $\gamma=1$ we recall the following result from \cite{BDIK}. This result was obtained with the help of the known behavior \eqref{sinres}, i.e. the line of reasoning is reversed in loc. cit. -- nevertheless it will provide us with valuable insight for the general case \eqref{IIKSker} where little information on $\lambda_i(J)$ is available. The result we are referring to is the Theorem below.\footnote{We have changed the definition of $v=-\frac{1}{2}\ln(1-\gamma)$ in \cite{BDIK} to $v=-\ln(1-\gamma)$. This allows us to have a uniform $v$ for all kernels considered in this paper. Also, \cite{BDIK} uses $K_{\sin}(\lambda,\mu)=\frac{\sin s(\lambda-\mu)}{\pi(\lambda-\mu)}$ acting on $(-1,1)$.}
\begin{theo}[\cite{BDIK}, Theorem $1.12$] Given $\chi\in\mathbb{R}$ let $p=p(\chi)\in\mathbb{Z}_{\geq 1}$ such that $p=1$ for $\chi<\frac{1}{2}$ and $\chi+\frac{1}{2}<p\leq \chi+\frac{3}{2}$ for $\chi\geq\frac{1}{2}$. There exist constants $s_0=s_0(\chi)>0$ and $v_0=v_0(\chi)>0$ such that
\begin{equation}\label{percy}
	D(J_{\sin};\gamma)=\exp\left[-\frac{s^2}{2}\right]s^{-\frac{1}{4}}c_0\prod_{i=0}^{p-1}\left(1+\e^{-v}\frac{\lambda_i(J_{\sin})}{1-\lambda_i(J_{\sin})}\right)\left(1+\mathcal{O}\left(s^{-\min\{p-\chi-\frac{1}{2},1\}}\right)\right)
\end{equation}
uniformly for $s\geq s_0,v=-\ln(1-\gamma)\geq v_0$ and $v\geq 2s-\chi\ln s$. The universal constant $c_0$ appeared in \eqref{DyWidgap}.
\end{theo}
 Note that from \eqref{sinres}, we have explicit information on the factors,
\begin{equation*}
	1+\e^{-v}\frac{\lambda_i(J_{\sin})}{1-\lambda_i(J_{\sin})}=1+\frac{i!}{\sqrt{\pi}}2^{-3i-2}s^{-i-\frac{1}{2}}\e^{2s-v}\left(1+\mathcal{O}\left(s^{-1}\right)\right),\ \ \ 0\leq i\leq p-1
\end{equation*}
and these contribute in general to the leading behavior of $D(J_{\sin};\gamma)$ as $|J_{\sin}|\rightarrow\infty$. In fact by construction of $p\in\mathbb{Z}_{\geq 1}$ we have $0<p-\chi-\frac{1}{2}\leq 1$ and thus for $v=2s-\chi\ln s$,
\begin{eqnarray*}
	i=p-1:\ \ &&\e^{-v}\big(1-\lambda_{p-1}(J_{\sin})\big)^{-1}\sim C_{p-1}s^{-(p-\chi-\frac{1}{2})}=o(1);\\
	i=0,\ldots,p-2:\ \ && \e^{-v}\big(1-\lambda_{i}(J_{\sin})\big)^{-1}\sim C_is^{-(i-\chi+\frac{1}{2})}\geq s^{-(p-\chi-\frac{3}{2})}\neq o(1)
\end{eqnarray*}
so that the asymptotics of $D(J_{\sin};\gamma)$ changes each time we cross one of the infinitely many Stokes curves
\begin{equation*}
	v=2s-\chi\ln s,\ \ \ \chi=p-\frac{1}{2},\ \ p\in\mathbb{Z}_{\geq 1}.
\end{equation*}
In short, the asymptotic Stokes phenomenon of the Fredholm determinant encodes the information of the spectrum we are after. The natural idea is now to reverse the procedure which lead to \eqref{percy}, i.e.
\begin{enumerate}
	\item[(A)] Identify the asymptotic Stokes region for a given determinant $D(J;\gamma)$ and derive transition asymptotics of type \eqref{percy}, i.e. derive an expansion of the form
	\begin{equation*}
		D(J;\gamma)\sim D(J;1)\prod_{i=0}^{p-1}\left(1+\e^{-v}f_i(J)\right),\ \ \ |J|\rightarrow+\infty,\ \ \gamma\uparrow 1,
	\end{equation*}
	with distinct factors $f_i(J)$ and where $p\in\mathbb{Z}_{\geq 1}$ is tailored to the double scaling regime.
	\item[(B)] Extract from $f_i(J)$ the behavior of $\lambda_i(J)$ as $|J|\rightarrow\infty$ and $i\in\mathbb{Z}_{\geq 0}$ is fixed.
\end{enumerate}
\begin{remark}\label{philophase} The existence of a phase transition for $D(J;\gamma)$ near $\gamma=1$ is a priori ensured for many, if not all, kernels in orthogonal polynomial random matrix theory. The operators are positive definite and since $D(J;1)$ has a probabilistic interpretation we have automatically 
\begin{equation*}
	1>\lambda_0(J)\geq\lambda_1(J)\geq\ldots>0.
\end{equation*}
\end{remark}

In this paper we carry out both steps for $K_{\textnormal{Ai}},K_{\textnormal{Bess}}^{(a)}$ and discuss applicability of the results to a more general Painlev\'e type kernel to be discussed below.
\subsection{Results for the Airy and Bessel kernel}
The following two Theorems are the major results of the manuscript.
\begin{theo}\label{Airymain} Given $\chi\in\mathbb{R}$ determine $p=p(\chi)\in\mathbb{Z}_{\geq 0}$ such that $p=0$ for $\chi<-\frac{1}{2}$ and $\chi+\frac{1}{2}<p\leq\chi+\frac{3}{2}$ for $\chi\geq-\frac{1}{2}$. There exist positive $t_0=t_0(\chi)$ and $v=v_0(\chi)$ such that
\begin{equation}\label{the1res}
	D(J_{\textnormal{Ai}};\gamma)=\exp\left[\frac{s^3}{12}\right]|s|^{-\frac{1}{8}}\eta_0\prod_{i=0}^{p-1}\left(1+\frac{i!}{\sqrt{\pi}}2^{-\frac{7}{2}i-\frac{9}{4}}t^{-i-\frac{1}{2}}\e^{\frac{2}{3}\sqrt{2}t-v}\right)\left(1+\mathcal{O}\left(t^{-\min\{p-\chi-\frac{1}{2},\frac{1}{2}\}}\right)\right)
\end{equation}
uniformly for $t=(-s)^{\frac{3}{2}}\geq t_0,v=-\ln(1-\gamma)\geq v_0$ and $v\geq\frac{2}{3}\sqrt{2}t-\chi\ln t$. Here
\begin{equation*}
	\eta_0=\exp\left[\frac{1}{24}\ln 2+\z'(-1)\right],
\end{equation*}
involving the Riemann zeta function $\z(\cdot)$ and in case $p=0$ we take $\prod_{i=0}^{p-1}(\ldots)\equiv 1$.
\end{theo}
Expansion \eqref{the1res} is the direct analogue of \eqref{percy} for the Airy kernel determinant and completes part (A) of the aformentioned procedure for the same object.
\begin{remark}
Note that in the limit $\varkappa_{_{\textnormal{Ai}}}\equiv\frac{v}{t}\rightarrow+\infty$, all factors in the product \eqref{the1res} contribute to the error term and we correctly restore the well-known expansion for $D(J_{\textnormal{Ai}};1)$ to leading order, cf. \cite{BBD,DIK}. See Appendix \ref{Airyapp} for further details on this matter.
\end{remark}
\begin{remark} Expansion \eqref{the1res} completely captures the behavior of $D(J_{\textnormal{Ai}};\gamma)$ in the underlying Stokes region, the Stokes curves in the $(v,t)$-plane are given by
\begin{equation*}
	v=\frac{2}{3}\sqrt{2}\,t-\chi \ln t,\ \ \ \ \ \ \ \chi=q-\frac{1}{2},\ \ q\in\mathbb{Z}_{\geq 1}.
\end{equation*}
In previous work \cite{B} the very first Stokes curve $q=1$ was identified and here we extend the results (with different techniques) to the full Stokes region. 
\end{remark}
The most important consequence for us is contained in the following Corollary which is part (B) of the procedure for the Airy kernel.
\begin{cor}\label{Aispecex} For any fixed $i\in\mathbb{Z}_{\geq 0}$, we have, as $s\rightarrow-\infty$,
\begin{equation*}
	1-\lambda_i(J_{\textnormal{Ai}})=\frac{\sqrt{\pi}}{i!}2^{\frac{7}{2}i+\frac{9}{4}}t^{i+\frac{1}{2}}\e^{-\frac{2}{3}\sqrt{2}\,t}\big(1+o(1)\big),\ \ \ t=(-s)^{\frac{3}{2}}.
\end{equation*}
\end{cor}
This expansion matches exactly the formal result of Tracy and Widom in \cite{TW}, $(1.23)$. Next we turn our attention towards \eqref{Bessk}.
\begin{theo}\label{Besselmain} Given $\chi\in\mathbb{R}$ and $a>-1$, determine $p=p(\chi)\in\mathbb{Z}_{\geq 0}$ such that $p=0$ for $\chi<-\frac{1}{2}$ and $\chi+\frac{1}{2}<p\leq\chi+\frac{3}{2}$ for $\chi\geq-\frac{1}{2}$. There exist positive constants $t_0=t_0(\chi,a)$ and $v_0=v_0(\chi,a)$ such that
\begin{equation}\label{the2res}
	D(J_{\textnormal{Bess}};\gamma)=\exp\left[-\frac{s}{4}+a\sqrt{s}\right]s^{-\frac{1}{4}a^2}\tau_a\prod_{i=0}^{p-1}\left(1+d_i(a)t^{-2i-1-a}\e^{2t-v}\right)
	\left(1+\mathcal{O}\left(\max\left\{t^{-2(p-\chi-\frac{1}{2})},\frac{\ln t}{t}\right\}\right)\right)
\end{equation}
uniformly for $t=s^{\frac{1}{2}}\geq t_0,v=-\ln(1-\gamma)\geq v_0$ and $v\geq 2t-2(\chi+\frac{a}{2})\ln t$. Here
\begin{equation*}
	d_i(a)=i!\,\Gamma(1+a+i)\pi^{-1}2^{-4i-2a-3},\ \ \ \ \ \tau_a=\frac{G(1+a)}{(2\pi)^{\frac{a}{2}}},
\end{equation*}
involving Barnes G-function $G(\cdot)$ and $\Gamma(\cdot)$ is the Euler gamma function. We take $\prod_{i=0}^{p-1}(\ldots)\equiv 1$ for $p=0$.
\end{theo}
\begin{rem} As $\varkappa_{_\textnormal{Bess}}\equiv\frac{v}{t}\rightarrow+\infty$, we rediscover the leading terms of $D(J_{\textnormal{Bess}};1)$ since all factors in the product of \eqref{the2res} move to the error term, compare \cite{DKV,E} as well as Appendix \ref{Bessapp}. On the other hand the Stokes curves of $D(J_{\textnormal{Bess}};\gamma)$ are given by
\begin{equation*}
	v=2t-2\left(\chi+\frac{a}{2}\right)\ln t,\ \ \ \ \ \ \ \ \ \chi=q-\frac{1}{2},\ \ q\in\mathbb{Z}_{\geq 1}
\end{equation*}
and these have not been analyzed previously.
\end{rem}
The Corollary below summarizes the anticipated eigenvalue asymptotics for $K_{\textnormal{Bess}}^{(a)}$. We have again matching with the formal result of \cite{TW2},$(1.27)$.
\begin{cor}\label{Bessspecex} For any fixed $i\in\mathbb{Z}_{\geq 0}$ and $a>-1$, we have, as $s\rightarrow+\infty$,
\begin{equation*}
	1-\lambda_i(J_{\textnormal{Bess}})=\frac{\pi}{i!}\frac{2^{4i+2a+3}}{\Gamma(1+a+i)}\,t^{2i+1+a}\e^{-2t}\big(1+o(1)\big),\ \ \ \ t=s^{\frac{1}{2}}.
\end{equation*}
\end{cor}
\begin{remark} With \cite{TW,TW2}, Theorem \ref{Airymain} and \ref{Besselmain} can be used to derive transition asymptotics for the underlying Painlev\'e II and III transcendents.
\end{remark}
\subsection{The ring of integrable integral operators}
The technical part of the procedure is contained in its part (A), i.e. the derivation of transition asymptotics of, say, type \eqref{the1res} and \eqref{the2res}. These asymptotics are obtained through an application of the Deift-Zhou nonlinear steepest descent method \cite{DZ} to the following master Riemann-Hilbert problem (RHP). This RHP is associated with integral operators of type \eqref{IIKSker} and first appeared in \cite{IIKS}.
\begin{problem}[Master RHP]\label{masterIIKS} Determine $Y(z)=Y(z;J,\gamma)\in\mathbb{C}^{2\times 2}$, a matrix-valued piecewise analytic function which is uniquely characterized by the following four properties.
\begin{enumerate}
	\item $Y=Y(z)$ is analytic for $z\in\mathbb{C}\backslash\overline{J}$ and $J\subset\mathbb{C}$ is assumed to be a simple oriented curve.
	\item The limiting values $Y_{\pm}(z)$ from either side of the contour $J$ are square integrable and related via the jump condition
	\begin{equation*}
		Y_+(z)=Y_-(z)\begin{pmatrix}1-2\pi\im\gamma\phi(z)\psi(z)&2\pi\im\gamma\phi^2(z)\\ -2\pi\im\gamma\psi^2(z) & 1+2\pi\im\gamma\phi(z)\psi(z)\end{pmatrix}, \ \ z\in J.
	\end{equation*}
	\item At possible finite endpoints of $J$, the function $Y(z)$ is assumed to be square integrable.
	\item As $z\rightarrow\infty$, in a full neighborhood of infinity,
	\begin{equation*}
		Y(z)=I+Y_{\infty}z^{-1}+\mathcal{O}\left(z^{-2}\right),\ \ \ \ Y_{\infty}=\big(Y_{\infty}^{jk}\big)_{j,k=1}^2.
	\end{equation*}
\end{enumerate}
\end{problem}
For a given kernel \eqref{IIKSker} the derivation of an asymptotic solution for RHP \ref{masterIIKS} is kernel specific, nevertheless the general philosophy is always to obtain first a ``local identity" for the logarithmic derivative
\begin{equation}\label{philo}
	\partial\ln D(J;\gamma),
\end{equation}
taken with respect to endpoints of $J$ and $\gamma$ fixed. This means an identity for $\partial\ln D(J;\gamma)$ involving local characteristica of $Y(z)$ such as its residue at $z=\infty$ or behavior near specific points of $J$, see \eqref{diff:1} and \eqref{diff:2} below. Second, through an application of the nonlinear steepest descent method \cite{DZ}, one derives asymptotic expansions for the local characteristica and, third, integrates \eqref{philo} to obtain an expansion for $D(J;\gamma)$.
\begin{remark} The RHP \ref{masterIIKS} has been asymptotically analyzed for a variety of kernels over the past 20 years. In particular in the context of gap probabilities $D(J;1)$ we mention, for instance, \cite{DIZ,KKMST,CIK,BI1}. Away from $\gamma=1$ the method has been successfully used in \cite{BI2,BDIK}.
\end{remark}
\subsection{Non-generic kernels in Hermitian matrix models}\label{singker} The Airy kernel \eqref{Aik} 
is obtained from scaling matrices in the GUE at the edge of the support of the eigenvalue densities $\rho$. Near a, say, right edge point $x^{\ast}$ of an interval in the density support we have
\begin{equation}\label{nongen}
	\rho(x)\sim c(x^{\ast}-x)^{2k+\frac{1}{2}},\ \ x\uparrow x^{\ast},\ \ c>0,\ k\in\mathbb{Z}_{\geq 0}
\end{equation}
and \eqref{Aik} simply corresponds to the generic situation $k=0$. In more fine tuned cases, i.e. for $k\in\mathbb{Z}_{\geq 1}$, the limiting kernels are described by a Lax pair solution associated with a distinguished solution of the $(2k)$-th member of the Painlev\'e I hierarchy. We will briefly discuss the first non-trivial case $k=1$ and refer to \cite{CV} for a rigorous derivation of the underlying critical kernel:  
\begin{equation}\label{PI2ker}
	K_{P_I^2}(\lambda,\mu;x,\tau)=\frac{\im}{2\pi}\frac{\Phi_{11}(\lambda;x,\tau)\Phi_{21}(\mu;x,\tau)-\Phi_{11}(\mu;x,\tau)\Phi_{21}(\lambda;x,\tau)}{\lambda-\mu};\ \ \ J_{P_I^2}=(s,\infty),\ \ s\in\mathbb{R}
\end{equation}
with $x,\tau\in\mathbb{R}$ and the matrix entries $\Phi_{jk}(\z)$ (viewed as analytic extensions from $\textnormal{arg}\,\z\in(0,\frac{6\pi}{7})$ to the entire complex plane) are characterized in terms of the following RHP.
\begin{problem}\label{PI2RHP} Determine $\Phi(\z)=\Phi(\z;x,\tau)\in\mathbb{C}^{2\times 2}$ with $x,\tau\in\mathbb{R}$ such that
\begin{enumerate}
	\item $\Phi(\z)$ is analytic for $\z\in\mathbb{C}\backslash\bigcup_{j=1}^4\Gamma_j$ with
	\begin{equation*}
		\Gamma_1=[0,\infty),\ \ \ \Gamma_3=(-\infty,0],\ \ \ \Gamma_2=\e^{-\im\frac{\pi}{7}}(-\infty,0],\ \ \ \Gamma_4=\e^{\im\frac{\pi}{7}}(-\infty,0]
	\end{equation*}
	and the rays are oriented as shown in Figure \ref{figure1}.\footnote{Except for the opening angles of $\Gamma_2\cup\Gamma_4$ the jump contours in RHP \ref{PI2RHP} and RHP \ref{Airypara} below are identical.}
	\item The limiting values $\Phi_{\pm}(\z)$ from either side of the jump contour satisfy the jump relations
	\begin{equation*}
		\Phi_+(\z)=\Phi_-(\z)\bigl(\begin{smallmatrix}0 &1\\ -1&0\end{smallmatrix}\bigr),\ \z\in\Gamma_3;\ \ \ \ \ \ \Phi_+(\z)=\Phi_-(\z)\bigl(\begin{smallmatrix}1&1\\ 0&1\end{smallmatrix}\bigr),\ \z\in\Gamma_1
	\end{equation*}
	and
	\begin{equation*}
		\Phi_+(\z)=\Phi_-(\z)\bigl(\begin{smallmatrix}1 & 0\\ 1 & 1\end{smallmatrix}\bigr),\ \  \z\in\Gamma_2\cup\Gamma_4.
	\end{equation*}
	\item $\Phi(\z)$ is bounded at $\z=0$.
	\item As $\z\rightarrow\infty$, away from the jump contours
	\begin{equation*}
		\Phi(\z)=\z^{-\frac{1}{4}\sigma_3}\frac{1}{\sqrt{2}}\begin{pmatrix}1 & 1\\ -1 & 1\end{pmatrix}\e^{-\im\frac{\pi}{4}\sigma_3}\left\{I-v\sigma_3\z^{-\frac{1}{2}}+\frac{1}{2}\begin{pmatrix}v^2&-\im u\\ \im u & v^2 \end{pmatrix}\z^{-1}+\mathcal{O}\left(\z^{-\frac{3}{2}}\right)\right\}\e^{-\theta(\z;x,\tau)\sigma_3}
	\end{equation*}
	where $u,v$ do not depend on $\z$, we define $\z^{\alpha}:\mathbb{C}\backslash(-\infty,0]\rightarrow\mathbb{C}$ such that $\z^{\alpha}>0$ for $\z>0$ and
	\begin{equation}\label{PI2:1}
		\theta(\z;x,\tau)=\frac{2}{7}\z^{\frac{7}{2}}+\frac{2}{3}\tau\z^{\frac{3}{2}}-2x\z^{\frac{1}{2}}.
	\end{equation}
\end{enumerate}
\end{problem}
\begin{remark} The particular scaling chosen in \eqref{PI2:1} differs slightly from \cite{CV} and also \cite{C1}. For instance $\Psi(\z;s,t_1)$ in \cite{C1}, Section $2$, connects to $\Phi(\z;x,\tau)$ in RHP \ref{PI2RHP} via
\begin{equation*}
	\Phi(\z;x,\tau)=2^{-\frac{\sigma_3}{14}}\Psi\left(2^{-\frac{2}{7}}\z;-2^{\frac{8}{7}}x,2^{-\frac{4}{7}}\tau\right),\ \ \z\in\mathbb{C}\Big\backslash\bigcup_{j=1}^4\Gamma_j,\ \ x,\tau\in\mathbb{R}.
\end{equation*}
On the other hand \eqref{PI2:1} matches exactly \cite{CIK},$(1.21)$.
\end{remark}
The above RHP characterizes a solution of the second member of the Painlev\'e I hierarchy, the $P_I^2$ equation: indeed it is straightforward to show that $\Phi(\z;x,\tau)$ solves a Lax system
\begin{equation*}
	\frac{\partial\Phi}{\partial x}(\z;x,\tau)=A(\z;x,\tau)\Phi(\z;x,\tau),\ \ \ \ \frac{\partial\Phi}{\partial \z}(\z;x,\tau)=B(\z;x,\tau)\Phi(\z;x,\tau)
\end{equation*}
and its compatibility yields for $u=u(x;\tau)$ and $v=v(x;\tau)$,
\begin{equation}\label{PI2:2}
	(P_I^2):\ \ \ x=-\tau u-\left(\frac{5}{2}u^3-\frac{5}{32}\big(u_x^2+2uu_{xx}\big)+\frac{1}{256}u_{xxxx}\right);\ \ \ \ v_x=-2u.
\end{equation}
The precise expressions for $A(\z)$ and $B(\z)$ can be found in, say, \cite{C1} but they will not be relevant for us, we only use the existence of a real-valued, pole-free solution of \eqref{PI2:2} for $x,\tau\in\mathbb{R}$, see again \cite{C1}. Equivalently we only use solvability of RHP \ref{PI2RHP} for $x,\tau\in\mathbb{R}$.
\begin{remark} The solution to RHP \ref{PI2RHP} is unique within the family of solutions of the form $\bigl(\begin{smallmatrix} 1 & 0\\ \omega & 1\end{smallmatrix}\bigr)\Phi(\z;x,\tau)$ with $\omega$ independent of $\z$. However, the critical kernel $K_{P_I^2}(\lambda,\mu;x,\tau)$ is invariant under this gauge.
\end{remark}
Using the Lax equation $\frac{\partial\Phi}{\partial x}=A\Phi$ as well as symmetry constraints of RHP \ref{PI2RHP} we have for $x,\tau\in\mathbb{R}$,
\begin{equation*}
	K_{P_I^2}(\lambda,\mu;x,\tau)=\frac{\im}{2\pi}\int_0^{\infty}\Phi_{11}(\lambda;x+t,\tau)\Phi_{11}(\mu;x+t,\tau)\d t=\frac{1}{2\pi}\int_0^{\infty}\Phi_{11}(\lambda;x+t,\tau)\overline{\Phi_{11}(\mu;x+t,\tau)}\,\d t,
\end{equation*}
which should be viewed as the generalization of \eqref{convu}. The last identity implies in particular that $K_{P_I^2}:L^2\big((s,\infty);\d\lambda\big)\circlearrowleft$ is positive-definite, compare Remark \ref{philophase}. For the analysis of the spectrum we would have to solve RHP \ref{masterIIKS} tailored to the choice
\begin{equation*}
	\phi(z)=\Phi_{11}(z;x,\tau),\ \ \ \ \psi(z)=\frac{\im}{2\pi}\Phi_{21}(z;x,\tau),\ \ \ \ J_{P_I^2}=(s,\infty)
\end{equation*}
as $s\rightarrow-\infty,\gamma\uparrow 1$ (in an appropriate Stokes regime) and $x,\tau\in\mathbb{R}$ are kept fixed. The derivation of this solution is in fact very similar\footnote{This can already be seen from the analysis of the gap probability $D(J_{P_I^2};1)$ presented in \cite{CIK}.} to the analysis of $K_{\textnormal{Ai}}$ presented in Sections \ref{Airyp1} and \ref{Airyp2} and it would provide us with an expansion for $D(J_{P_I^2};\gamma)$ of type \eqref{the1res}. From it we would then obtain the desired information on the spectrum $\{\lambda_i(J_{P_I^2})\}_{i\geq 0}$. We shall devote a separate publication to the spectrum analysis associated with non-generic edge behavior \eqref{nongen} for general $k\in\mathbb{Z}_{\geq 1}$.
\begin{remark} The purpose of this small section is to emphasize that the method of commuting differential operators will become very involved, if not impossible, when applied to, say, \eqref{PI2ker}. As can be seen from the Lax system a commuting differential operator would necessarily involve a Painlev\'e transcendent. And besides \eqref{PI2ker} several other kernels in the theory of random determinantal point processes (for instance Pearcey \cite{TWP,BK} or Tacnode \cite{DKZ} kernels) fit naturally into the Riemann-Hilbert based scheme rather than the method of commuting differential operators.
\end{remark}
\subsection{Outline of paper} The manuscript is split into the following two major parts.
\begin{enumerate}
	\item[(i)] First, in Sections \ref{Airyp1} and \ref{Airyp2}, we derive an asymptotic solution to RHP \eqref{masterIIKS} subject to \eqref{Airytailor} and \eqref{scale}. As can be seen in particular from Section \ref{Airyp1} the steps we carry out in that section are closely related to the ones chosen in the analysis of the gap probability $D(J_{\textnormal{Ai}};1)$, see \cite{CIK}, Section $3$. The effect of $\gamma\neq 1$, in contrast to the gap probability, becomes fully visible only in Section \ref{Airyp2}: we have a new $g$-function and we require different model functions than the ones chosen in \cite{CIK}, some differ only by a rank one perturbation, others are entirely new. In particular the new ones encode the discretized effect of the Stokes region and we use classical Hermite polynomials for our construction.\smallskip

Once all local contributions are in place we derive our first small norm estimates, however additional steps are required to resolve a certain singular structure, see Subsection \ref{Airysing} for all details. The necessity of solving RHP \ref{masterIIKS} is related to the existence of differential identities for $\ln D(J_{\textnormal{Ai}};\gamma)$, compare Subsection \ref{diffid}. Using these identities we first obtain an asymptotic expansion for $\frac{\partial}{\partial s}\ln D(J_{\textnormal{Ai}};\gamma)$ and then perform a definite integration using the known expansion for $D(J_{\textnormal{Ai}};1)$. In fact the known expansion for $D(J_{\textnormal{Ai}};1)$ has to be improved to fit our purposes and we carry out the necessary steps in Appendix \ref{Airyapp}. Finally, the information on $\{\lambda_i(J_{\textnormal{Ai}})\}$ is derived in Section \ref{Airyp4} using an inductive argument.
	\item[(ii)] Second, in Sections \ref{Bessp1} and \ref{Bessp2}, we focus on RHP \ref{masterIIKS} with \eqref{BIIKS} and \eqref{Bscale} in the background. To the author's knowledge the Bessel determinant $D(J_{\textnormal{Bess}};\gamma)$ has not been analyzed previously with the help of RHP \ref{masterIIKS}, although the analysis displays a few overall similarities with part (i). In fact we use again rank perturbations of more standard model functions and also here orthogonal polynomials, this time Laguerre polynomials. The asymptotic expansion for $\frac{\partial}{\partial s}\ln D(J_{\textnormal{Bess}};\gamma)$ is derived in Section \ref{Bessinteg} and subsequently integrated using a refinement for $D(J_{\textnormal{Bess}};1)$ derived in Appendix \ref{Bessapp}. The eigenvalue expansions are obtained via the same argument as in part (i), compare Section \ref{Besseigexp}.
\end{enumerate}
\begin{remark} In recent years several model functions in nonlinear steepest descent analysis with discretized parameter values have appeared, we mention in chronological order \cite{C,BL,BT,BM}. These works have all in common that orthogonal polynomials (with the corresponding degrees as the discrete parameter) were used in the relevant constructions.
\end{remark}

\section{Nonlinear steepest descent analysis associated with $K_{\textnormal{Ai}}$ -- part 1}\label{Airyp1}
The Airy kernel \eqref{Aik} is of type \eqref{IIKSker} with
\begin{equation}\label{Airytailor}
	\phi(z)=\textnormal{Ai}(z),\ \ \ \psi(z)=\textnormal{Ai}'(z);\ \ \ \ \ J_{\textnormal{Ai}}=(s,\infty),\ \ s\in\mathbb{R}
\end{equation}
and we shall derive an asymptotic solution of RHP \ref{masterIIKS} for sufficiently large (negative) $s$ and $\gamma$ close to $1$ such that
\begin{equation}\label{scale}
	\varkappa_{_\textnormal{Ai}}\equiv \frac{v}{t}=\frac{2}{3}\sqrt{2}-\chi\frac{\ln t}{t},\ \ \ \ \ v=-\ln(1-\gamma)>0,\ \ \ \ t=(-s)^{\frac{3}{2}};\ \ \ \ \ \chi\in\mathbb{R}_{\geq 0}.
\end{equation}
\subsection{Preliminary transformations}\label{prelim}
Before we display the connection between $Y(z)$ and $D(J_{\textnormal{Ai}};\gamma)$ we first simplify RHP \ref{masterIIKS} in case of \eqref{Airytailor}: Introduce the entire, unimodular function
\begin{equation}\label{e:2}
	\Phi_0(\z)=\sqrt{2\pi}\,\e^{-\im\frac{\pi}{4}\sigma_3}\bigl(\begin{smallmatrix}
	1 & 0\\
	0 & -\im
	\end{smallmatrix}\bigr)\begin{pmatrix}
	\textnormal{Ai}(\z) & \e^{\im\frac{\pi}{3}}\textnormal{Ai}\big(\e^{-\im\frac{2\pi}{3}}\z\big)\smallskip\\
	\textnormal{Ai}'(\z) & \e^{-\im\frac{\pi}{3}}\textnormal{Ai}'\big(\e^{-\im\frac{2\pi}{3}}\z\big)
	\end{pmatrix},\ \ \ \ \ \z\in\mathbb{C}
\end{equation}
and assemble an {\it Airy-type parametrix},
\begin{equation}\label{e:3}
	\Phi(\z)=\Phi_0(\z)\begin{cases}
	I,&\textnormal{arg}\,\z\in(0,\frac{2\pi}{3})\\
	\bigl(\begin{smallmatrix}
	1 & 0\\
	-1 & 1
	\end{smallmatrix}\bigr),&\textnormal{arg}\,\z\in(\frac{2\pi}{3},\pi)\smallskip\\
	\bigl(\begin{smallmatrix}
	1 & -1\\
	0 & 1
	\end{smallmatrix}\bigr),&\textnormal{arg}\,\z\in(-\frac{2\pi}{3},0)\smallskip\\
	\bigl(\begin{smallmatrix}
	1 & -1\\
	0 & 1
	\end{smallmatrix}\bigr)\bigl(\begin{smallmatrix}
	1 & 0\\
	1 & 1
	\end{smallmatrix}\bigr),&\textnormal{arg}\,\z\in(-\pi,-\frac{2\pi}{3}).
	\end{cases}
\end{equation}
\begin{figure}[tbh]
\begin{center}
\resizebox{0.35\textwidth}{!}{\includegraphics{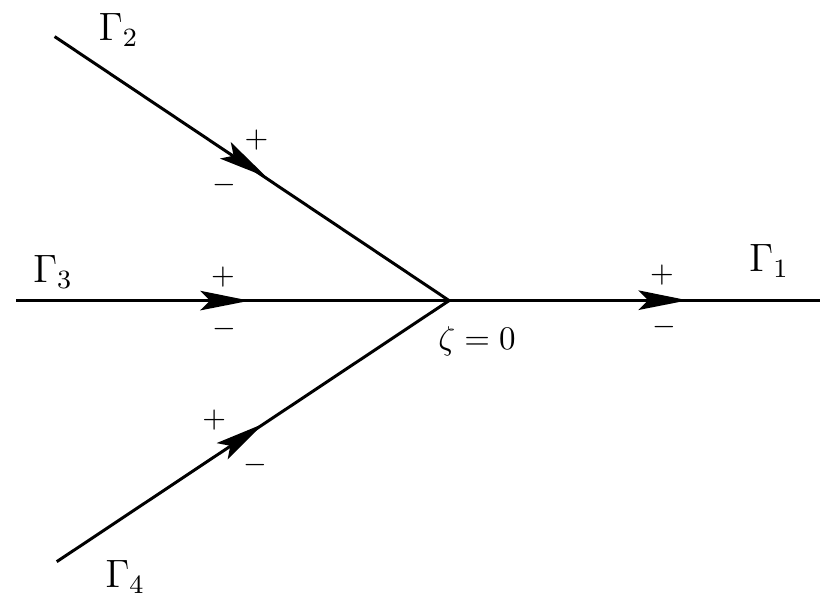}}
\caption{The oriented jump contours for the Airy parametrix $\Phi(\z)$ in the complex $\z$-plane.}
\label{figure1}
\end{center}
\end{figure}
This matrix valued function solves the following standard model problem which, in one form or another, has appeared numerous times in nonlinear steepest descent literature.
\begin{problem}\label{Airypara} The Airy parametrix $\Phi(\z)$ satisfies the properties below
\begin{enumerate}
	\item $\Phi(\z)$ is analytic for $\z\in\mathbb{C}\backslash\bigcup_{j=1}^4\Gamma_j$ with 
	\begin{equation*}
		\Gamma_1=[0,\infty),\ \ \ \Gamma_3=(-\infty,0],\ \ \ \Gamma_2=\e^{-\im\frac{\pi}{3}}(-\infty,0],\ \ \ \Gamma_4=\e^{\im\frac{\pi}{3}}(-\infty,0]
	\end{equation*}
	and all rays are oriented ``from left to right", compare Figure \ref{figure1}.
	\item The limiting values from either side of the jump contours are related via
	\begin{eqnarray*}
		\Phi_+(\z)&=&\Phi_-(\z)\bigl(\begin{smallmatrix}
		0 & 1\\
		-1 & 0
		\end{smallmatrix}\bigr),\ \ \z\in\Gamma_3;\ \ \ \ \ \ \ \ 
		 \Phi_+(\z)=\Phi_-(\z)\bigl(\begin{smallmatrix}
		1 & 1\\
		0 & 1
		\end{smallmatrix}\bigr),\ \ \ \ \z\in\Gamma_1;\\
		\Phi_+(\z)&=&\Phi_-(\z)\bigl(\begin{smallmatrix}
		1 & 0\\
		1 & 1
		\end{smallmatrix}\bigr),\ \ \ \ \z\in\Gamma_2\cup\Gamma_4.
	\end{eqnarray*}
	\item $\Phi(\z)$ is bounded at $\z=0$.
	\item As $\z\rightarrow\infty$, valid in a full neighborhood of infinity off the jump contours,
	\begin{equation*}
		\Phi(\z)=\z^{-\frac{1}{4}\sigma_3}\frac{1}{\sqrt{2}}\begin{pmatrix}
		1 & 1\\
		-1 & 1
		\end{pmatrix}\e^{-\im\frac{\pi}{4}\sigma_3}\left\{I+\frac{1}{48\z^{\frac{3}{2}}}\begin{pmatrix}
		1 & 6\im\\
		6\im & -1
		\end{pmatrix}+\mathcal{O}\left(\z^{-\frac{6}{2}}\right)\right\}\e^{-\frac{2}{3}\z^{\frac{3}{2}}\sigma_3}.
	\end{equation*}
	where $\z^{\alpha}$ is defined and analytic for $\z\in\mathbb{C}\backslash(-\infty,0]$ such that $\z^{\alpha}>0$ for $\z>0$.
\end{enumerate}
\end{problem}

The model function \eqref{e:3} is useful as it allows us to reduce RHP \ref{masterIIKS} to a problem with constant, that is $z$-independent, jumps. This step appeared in \cite{CIK} and amounts to the following transformation: We set for $s<0$, see Figure \ref{figure2},
\begin{equation}\label{center:1}
	X(z)=Y(z)\Phi(z)\begin{cases}
	I,&z\in\Omega_1\cup\Omega_2\cup\Omega_3\\
	\bigl(\begin{smallmatrix}
	1 & 0\\
	1 & 1
	\end{smallmatrix}\bigr),&z\in\Omega_4\smallskip\\
	\bigl(\begin{smallmatrix}
	1 & 0\\
	-1 & 1
	\end{smallmatrix}\bigr),&z\in\Omega_5
	\end{cases}
\end{equation}
and in case $s>0$, see Figure \ref{figure3},
\begin{equation}\label{center:2}
	X(z)=Y(z)\Phi(z)\begin{cases}
	I,&z\in\Omega_1\cup\Omega_2\cup\Omega_3\\
	\bigl(\begin{smallmatrix}
	1 & 0\\
	-1 & 1
	\end{smallmatrix}\bigr),&z\in\Omega_4\smallskip\\
	\bigl(\begin{smallmatrix}
	1 & 0\\
	1 & 1
	\end{smallmatrix}\bigr),&z\in\Omega_5.
	\end{cases}
\end{equation}
\begin{figure}[tbh]
\begin{minipage}{0.4\textwidth} 
\begin{center}
\resizebox{0.9\textwidth}{!}{\includegraphics{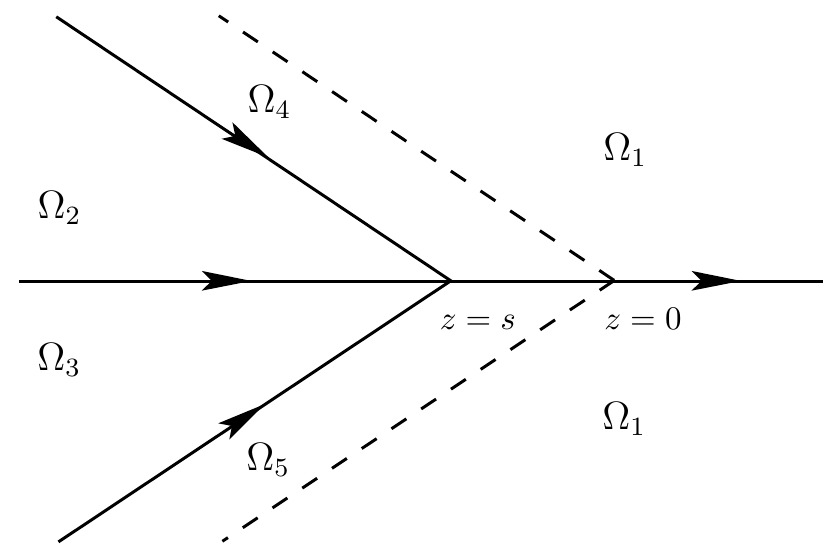}}
\caption{``Undressing" of RHP \ref{masterIIKS} in case $s<0$. Jump contours of $X(z)$ as solid lines.}
\label{figure2}
\end{center}
\end{minipage}
\begin{minipage}{0.4\textwidth}
\begin{center}
\resizebox{0.87\textwidth}{!}{\includegraphics{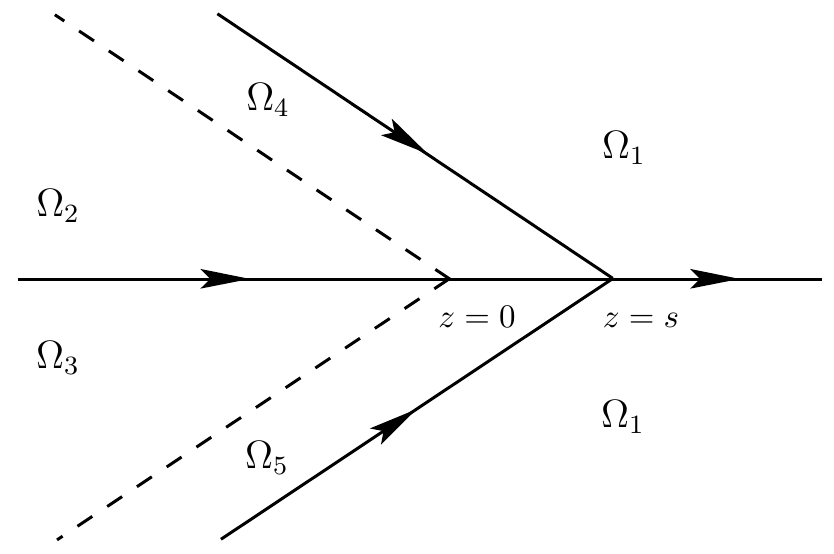}}
\caption{``Undressing" of RHP \ref{masterIIKS} in case $s>0$. Jump contours of $X(z)$ as solid lines.}
\label{figure3}
\end{center}
\end{minipage}
\end{figure}

Hence we obtain for $X(z)$ the RHP below.
\begin{problem}\label{unifRHP} Determine $X(z)\in\mathbb{C}^{2\times 2}$ such that
\begin{enumerate}
	\item $X(z)$ is analytic for $z\in\mathbb{C}\backslash\bigcup_{j=1}^4\Gamma_j^{(s)}$ with
	\begin{equation*}
		\Gamma_1^{(s)}=(s,\infty),\ \ \ \Gamma_3^{(s)}=(-\infty,s),\ \ \ \ \Gamma_2^{(s)}=s+\e^{-\im\frac{\pi}{3}}(-\infty,0),\ \ \ \ \Gamma_4^{(s)}=s+\e^{\im\frac{\pi}{3}}(-\infty,0)
	\end{equation*}
	\item The following jump conditions relate the limiting values $X_{\pm}(z)$, the jump contours are shown in Figures \ref{figure2} and \ref{figure3} as solid black lines:
	\begin{equation*}
		X_+(z)=X_-(z)\bigl(\begin{smallmatrix}
		0 & 1\\
		-1 & 0
		\end{smallmatrix}\bigr),\  z\in\Gamma_3^{(s)};\ \ \ \ \ \ \ X_+(z)=X_-(z)\bigl(\begin{smallmatrix}
		1 & 0\\
		1 & 1
		\end{smallmatrix}\bigr),\ z\in\Gamma_2^{(s)}\cup\Gamma_4^{(s)}
	\end{equation*}
	and
	\begin{equation*}
		X_+(z)=X_-(z)\bigl(\begin{smallmatrix}
		1 & 1-\gamma\\
		0 & 1
		\end{smallmatrix}\bigr),\ \ z\in\Gamma_1^{(s)}.
	\end{equation*}
	\item In a full vicinity of $z=s$,
	\begin{equation}\label{Xsing}
		X(z)=\widehat{X}(z)\begin{pmatrix}
		1 & \frac{\gamma}{2\pi\im}\ln(z-s)\\
		0 & 1
		\end{pmatrix}\begin{cases}
		I,&\textnormal{arg}(z-s)\in(0,\frac{2\pi}{3})\\
		\bigl(\begin{smallmatrix}
		1 & 0\\
		-1 & 1
		\end{smallmatrix}\bigr),&\textnormal{arg}(z-s)\in(\frac{2\pi}{3},\pi)\\
		\bigl(\begin{smallmatrix}
		1 & -1\\
		0 & 1
		\end{smallmatrix}\bigr)\bigl(\begin{smallmatrix}
		1 & 0\\
		1 & 1
		\end{smallmatrix}\bigr),&\textnormal{arg}(z-s)\in(\pi,\frac{4\pi}{3})\\
		\bigl(\begin{smallmatrix}
		1 & -1\\
		0 & 1
		\end{smallmatrix}\bigr),&\textnormal{arg}(z-s)\in(\frac{4\pi}{3},2\pi)
		\end{cases}
	\end{equation}
	with $\widehat{X}(z)$ analytic at $z=s$ and we fix the branch of the logarithm with $\textnormal{arg}(z-s)\in(0,2\pi)$.
	\item As $z\rightarrow\infty$,
	\begin{equation*}
		X(z)=z^{-\frac{1}{4}\sigma_3}\frac{1}{\sqrt{2}}\begin{pmatrix}
		1 & 1\\
		-1 & 1\\
		\end{pmatrix}\e^{-\im\frac{\pi}{4}\sigma_3}\left\{I+X_{\infty}z^{-\frac{1}{2}}+\widehat{X}_{\infty}z^{-1}+\mathcal{O}\left(z^{-\frac{3}{2}}\right)\right\}\e^{-\frac{2}{3}z^{\frac{3}{2}}\sigma_3}
	\end{equation*}
	where we choose principal branches for all fractional exponents. The matrices $X_{\infty},\widehat{X}_{\infty}$ are $z$-independent,
\begin{equation*}
  X_{\infty}=\frac{1}{2}\begin{pmatrix}
                                 -1 & \im \\ \im  & 1
                                \end{pmatrix}Y_{\infty}^{12};\ \ \ \ \widehat{X}_{\infty}=\frac{1}{2}\begin{pmatrix}
                                0 & \im\\ -\im & 0
                                \end{pmatrix}\big(Y_{\infty}^{11}-Y_{\infty}^{22}\big).
\end{equation*}
\end{enumerate}
\end{problem}
\subsection{Differential identity}\label{diffid} The connection of $D(J_{\textnormal{Ai}};\gamma)$ to the solution of RHP \ref{unifRHP} has been established for $\gamma=1$ in \cite{CIK}, $(2.17)$. For $\gamma\neq 1$ the derivation in loc. cit. can be copied almost verbatim and we simply state the following result.
\begin{prop} For fixed $\gamma\leq 1$, we have
\begin{equation}\label{diff:1}
	\frac{\partial}{\partial s}\ln D(J_{\textnormal{Ai}};\gamma)=\frac{\gamma}{2\pi\im}\big(X^{-1}(z)X'(z)\big)_{21}\Big|_{z\rightarrow s},\ \ \ \ (')=\frac{\d}{\d z}
\end{equation}
in terms of the solution $X(z)$ of RHP \ref{unifRHP} and the limit is carried out for $\textnormal{arg}(z-s)\in(0,\frac{2\pi}{3})$.
\end{prop}
\begin{remark} The requirement $\gamma\leq 1$ is imposed for technical purposes only, it is in this case that RHP \ref{unifRHP} is solvable for sufficiently large $v$ and $t$. For $\gamma>1$ the problem is only solvable for large $(v,t)$ away from a discrete set in the $(v,t)$-plane.
\end{remark}
\begin{remark} Another differential identity can be derived for $\frac{\partial}{\partial\gamma}\ln D(J_{\textnormal{Ai}};\gamma)$, see Section \ref{gderiv} for further details.
\end{remark}
On the upcoming pages we will derive an asymptotic solution of RHP \ref{unifRHP} and then compute the asymptotics of $\frac{\partial}{\partial s}\ln D(J_{\textnormal{Ai}};\gamma)$ through \eqref{diff:1}.

\section{Nonlinear steepest descent analysis associated with $K_{\textnormal{Ai}}$ -- part 2}\label{Airyp2}
\subsection{Initial transformation} 
From now on we assume that $s<0$ is sufficiently large negative. Define
\begin{equation*}
	T(z)=X(|s|z+s),\ \ z\in\mathbb{C}\backslash(\Sigma_T\cup\{0\}).
\end{equation*}
which ``centers" the problem at the origin $z=0$ so that we have jumps on the contour
\begin{equation*}
	\Sigma_T=\bigcup_{j=1}^4\Gamma_j
\end{equation*}
shown in Figure \ref{figure1}. More precisely we obtain
\begin{problem}\label{centerRHP} Determine a function $T(z)=T(z;s,\gamma)\in\mathbb{C}^{2\times 2}$ which is uniquely characterized by the following properties:
\begin{enumerate}
	\item $T(z)$ is analytic for $z\in\mathbb{C}\backslash(\Sigma_T\cup\{0\})$
	\item We have the jump conditions
\begin{equation*}
	T_+(z)=T_-(z)\bigl(\begin{smallmatrix}
	1 & 1-\gamma\\
	0 & 1
	\end{smallmatrix}\bigr),\ z\in\Gamma_1\backslash\{0\};\hspace{0.8cm} T_+(z)=T_-(z)\bigl(\begin{smallmatrix}
	0 & 1\\
	-1 & 0
	\end{smallmatrix}\bigr),\ z\in\Gamma_3\backslash\{0\}
\end{equation*}
and
\begin{equation*}
	T_+(z)=T_-(z)\bigl(\begin{smallmatrix}
	1 & 0\\
	1 & 1
	\end{smallmatrix}\bigr),\ \ \ z\in(\Gamma_2\cup\Gamma_4)\backslash\{0\}.
\end{equation*}
	\item In a neighborhood of $z=0$,
	\begin{equation*}
		T(z)=\widehat{T}(z)\begin{pmatrix}
		1 & \frac{\gamma}{2\pi\im}\ln z\\
		0 & 1
		\end{pmatrix}\begin{cases}
		I,&\textnormal{arg}\,z\in(0,\frac{2\pi}{3})\\
		\bigl(\begin{smallmatrix}
		1 & 0\\
		-1 & 1
		\end{smallmatrix}\bigr),&\textnormal{arg}\,z\in(\frac{2\pi}{3},\pi)\\
		\bigl(\begin{smallmatrix}
		1 & -1\\
		0 & 1
		\end{smallmatrix}\bigr)\bigl(\begin{smallmatrix}
		1 & 0\\
		1 & 1
		\end{smallmatrix}\bigr),&\textnormal{arg}\,z\in(\pi,\frac{4\pi}{3})\\
		\bigl(\begin{smallmatrix}
		1 & -1\\
		0 & 1
		\end{smallmatrix}\bigr),&\textnormal{arg}\,z\in(\frac{4\pi}{3},2\pi)\\
		\end{cases}
	\end{equation*}
	where $\widehat{T}(z)$ is analytic at $z=0$ and the branch of the logarithm is specified by the requirement $\textnormal{arg}\,z\in(0,2\pi)$.
	\item The normalization at $z=\infty$ reads as
\begin{equation*}
	T(z)=\big(|s|z\big)^{-\frac{1}{4}\sigma_3}\frac{1}{\sqrt{2}}\begin{pmatrix}
	1 & 1\\
	-1 & 1
	\end{pmatrix}\e^{-\im\frac{\pi}{4}\sigma_3}\left\{I+X_{\infty}(|s|z)^{-\frac{1}{2}}+\mathcal{O}\left(z^{-1}\right)\right\}\e^{-\frac{2}{3}(|s|z+s)^{\frac{3}{2}}\sigma_3}.
\end{equation*}
\end{enumerate}
\end{problem}
\begin{remark} Compared to \cite{CIK}, the RHPs \ref{unifRHP}, resp. \ref{centerRHP} also have a jump on $(s,\infty)$, resp. $(0,\infty)$ since we are interested in the analysis of $D(J_{\textnormal{Ai}};\gamma)$ with $\gamma\neq 1$ in general. This will have a crucial impact on the steps below.
\end{remark}
\subsection{Normalization transformation} We define for $z\in\mathbb{C}\backslash\mathbb{R}$,
\begin{equation}\label{g:1}
	g(z)=\frac{2}{3}z^{\frac{1}{2}}\left(z-\frac{3}{2}\right)+V\ln\left(\frac{1+(2z)^{\frac{1}{2}}}{1-(2z)^{\frac{1}{2}}}\right),\ \ \ \ \ \ V=\frac{\chi}{t},
\end{equation}
with principal branches for all fractional exponents and logarithms. In particular we choose
\begin{equation*}
	-\pi<\textnormal{arg}\left(\frac{1+(2z)^{\frac{1}{2}}}{1-(2z)^{\frac{1}{2}}}\right)\leq\pi.
\end{equation*}
Further steps require the following analytical properties of the $g$-function:
\begin{prop}\label{gimp} The function $g(z)$ introduced in \eqref{g:1} is analytic for $z\in\mathbb{C}\backslash((-\infty,0)\cup(\frac{1}{2},+\infty))$. In more detail along the real axis with orientation as shown in Figure \ref{figure1},
\begin{equation*}
	g_{\pm}(z)\equiv\lim_{\varepsilon\downarrow 0}g(z\pm\im\varepsilon)=\pm\frac{2}{3}\im\sqrt{|z|}\left(z-\frac{3}{2}\right)+\im V\textnormal{arg}\left(\frac{1\pm\im\sqrt{2|z|}}{1\mp\im\sqrt{2|z|}}\right),\ \ z\in(-\infty,0);
\end{equation*}
as well as
\begin{equation*}
	g_{\pm}(z)=\frac{2}{3}\sqrt{z}\left(z-\frac{3}{2}\right)+V\ln\left(\frac{1+\sqrt{2z}}{1-\sqrt{2z}}\right),\ \ z\in\left(0,\frac{1}{2}\right);
\end{equation*}
and
\begin{equation*}
	g_{\pm}(z)=\frac{2}{3}\sqrt{z}\left(z-\frac{3}{2}\right)+V\ln\left(\frac{\sqrt{2z}+1}{\sqrt{2z}-1}\right)\pm\im\pi V,\ \ z\in\left(\frac{1}{2},+\infty\right).
\end{equation*}
Also, near $z=\infty$,
\begin{equation*}
	-\frac{2}{3}(z-1)^{\frac{3}{2}}+g(z)=\im\pi V+\left(V\sqrt{2}-\frac{1}{4}\right)z^{-\frac{1}{2}}+\mathcal{O}\left(z^{-1}\right),\ \ z\rightarrow\infty,\ z\notin\mathbb{R}.
\end{equation*}
\end{prop}
At this point we introduce
\begin{equation*}
	S(z)=\e^{-\im\pi tV\sigma_3}T(z)\e^{tg(z)\sigma_3},\ \ \ z\in\mathbb{C}\backslash(\Sigma_T\cup\{0\}),
\end{equation*}
which leads us to the problem below.
\begin{problem} \label{gRHP}The normalized function $S(z)=S(z;s,\gamma)\in\mathbb{C}^{2\times 2}$ is characterized by the following properties
\begin{enumerate}
	\item $S(z)$ is analytic for $z\in\mathbb{C}\backslash(\Sigma_T\cup\{0\})$.
	\item The limiting values $S_{\pm}(z),z\in\Sigma_T$ are related by the equations
	\begin{equation*}
		S_+(z)=S_-(z)\begin{pmatrix} 1 & 0\\ \e^{2tg(z)} & 1 \end{pmatrix},\ \ z\in(\Gamma_2\cup\Gamma_4)\backslash\{0\};\hspace{0.75cm}S_+(z)=S_-(z)\begin{pmatrix} 0 & 1\\
		-1 & 0 \end{pmatrix},\ \ z\in\Gamma_3\backslash\{0\};
	\end{equation*}
	and
	\begin{eqnarray*}
		S_+(z)&=&S_-(z)\begin{pmatrix} 1 & \e^{-t(\varkappa_{\textnormal{Ai}}+2g(z))}\\ 0 & 1 \end{pmatrix},\ \ z\in\left(0,\frac{1}{2}\right);\\
		S_+(z)&=&S_-(z)\begin{pmatrix} \e^{t(g_+(z)-g_-(z))} & \e^{-t(\varkappa_{\textnormal{Ai}}+g_+(z)+g_-(z))}\\ 0 & \e^{-t(g_+(z)-g_-(z))} \end{pmatrix},\ \ z\in\left(\frac{1}{2},+\infty\right).
	\end{eqnarray*}
	\item Near $z=0$, with $\textnormal{arg}\,z\in(0,2\pi)$,
	\begin{equation*}
		\e^{\im\pi tV\sigma_3}S(z)\e^{-tg(z)\sigma_3} = \widehat{T}(z)\begin{pmatrix}
		1 & \frac{\gamma}{2\pi\im}\ln z\\
		0 & 1
		\end{pmatrix}\begin{cases}
		I,&\textnormal{arg}\,z\in(0,\frac{2\pi}{3})\\
		\bigl(\begin{smallmatrix}
		1 & 0\\
		-1 & 1
		\end{smallmatrix}\bigr),&\textnormal{arg}\,z\in(\frac{2\pi}{3},\pi)\\
		\bigl(\begin{smallmatrix}
		1 & -1\\
		0 & 1
		\end{smallmatrix}\bigr)\bigl(\begin{smallmatrix}
		1 & 0\\
		1 & 1
		\end{smallmatrix}\bigr),&\textnormal{arg}\,z\in(\pi,\frac{4\pi}{3})\\
		\bigl(\begin{smallmatrix}
		1 & -1\\
		0 & 1
		\end{smallmatrix}\bigr),&\textnormal{arg}\,z\in(\frac{4\pi}{3},2\pi).\\
		\end{cases}
	\end{equation*}
	\item We have the normalized behavior at $z=\infty\notin\Sigma_T$, 
	\begin{equation*}
		S(z)=(|s|z)^{-\frac{1}{4}\sigma_3}\e^{-\im\pi tV\sigma_3}\frac{1}{\sqrt{2}}\begin{pmatrix} 1 & 1\\ -1 & 1 \end{pmatrix}\e^{-\im\frac{\pi}{4}\sigma_3}\e^{\im\pi tV\sigma_3}\left\{I+S_{\infty}z^{-\frac{1}{2}}+\mathcal{O}\left(z^{-1}\right)\right\},
	\end{equation*}
	with
	\begin{equation*}
		S_{\infty}=\e^{-\im\pi tV\sigma_3}X_{\infty}\e^{\im\pi tV\sigma_3}|s|^{-\frac{1}{2}}+t\left(V\sqrt{2}-\frac{1}{4}\right)\sigma_3.
	\end{equation*}
\end{enumerate}
\end{problem}
We make a few important observations: With $0<r<\frac{1}{8}$ fixed,
\begin{equation}\label{crux:1}
	\Re\big(g(z)\big)<0,\ \ \ z\in\big(\Gamma_2\cup\Gamma_4\big)\backslash D(0,r);\ \ \ \ \ D(z_0,r)=\big\{z\in\mathbb{C}:\ |z-z_0|<r\big\};
\end{equation}
next for $z\in(0,\frac{1}{2})\backslash(D(0,r)\cup D(\frac{1}{2},r))$ and sufficiently large $t\geq t_0,v\geq v_0$ (see \eqref{scale} for the definition of $\varkappa_{_\textnormal{Ai}}$),
\begin{equation}\label{crux:2}
	\varkappa_{_\textnormal{Ai}}+2g(z)=\varkappa_{_\textnormal{Ai}}+\frac{4}{3}\sqrt{z}\left(z-\frac{3}{2}\right)+2V\ln\left(\frac{1+\sqrt{2z}}{1-\sqrt{2z}}\right)\geq \delta>0.
\end{equation}
Finally,
\begin{eqnarray}
	g_+(z)-g_-(z)&=&2\pi\im V\equiv\textnormal{const.}\in\im\mathbb{R},\ \ \ \ z\in\left(\frac{1}{2},+\infty\right)\label{crux:3};\\
	\varkappa_{_\textnormal{Ai}}+g_+(z)+g_-(z)&=&\varkappa_{_\textnormal{Ai}}+\frac{4}{3}\sqrt{z}\left(z-\frac{3}{2}\right)+2V\ln\left(\frac{\sqrt{2z}+1}{\sqrt{2z}-1}\right)\geq\delta>0,\ \ z\in\left(\frac{1}{2},+\infty\right)\Big\backslash D\left(\frac{1}{2},r\right).\nonumber
\end{eqnarray}
These estimates lead us to the expectation that the major contribution to the asymptotic solution of the $S$-RHP arises from the line segments $(-\infty,0)\cup(\frac{1}{2},+\infty)$ as well as two small vicinities of $z=0$ and $z=\frac{1}{2}$. 
\begin{remark}\label{approx} Observe that through \eqref{g:1}, 
\begin{equation*}
	tV=\chi,
\end{equation*}
and if we represent
\begin{equation*}
	\chi=k+\alpha;\hspace{0.7cm} k\in\mathbb{Z}_{\geq 0},\ \ -\frac{1}{2}\leq\alpha<\frac{1}{2},
\end{equation*}
thus
\begin{equation*}
	\e^{2\pi\im tV}=\e^{2\pi\im\alpha}.
\end{equation*}
\end{remark}
We now continue with the relevant local analysis.
\subsection{Analysis of model Riemann-Hilbert problems} The {\it outer model function}, 
\begin{equation}\label{out}
	P^{(\infty)}(z)=(|s|z)^{-\frac{1}{4}\sigma_3}\e^{-\im\pi\alpha\sigma_3}\frac{1}{\sqrt{2}}\begin{pmatrix} 1 & 1\\ -1 & 1 \end{pmatrix}\e^{-\im\frac{\pi}{4}\sigma_3}\big(\mathcal{D}(z)\big)^{\sigma_3},\ \ \ z\in\mathbb{C}\Big\backslash\left(\big(-\infty,0\big]\cup\left[\frac{1}{2},+\infty\right)\right),
\end{equation}
with the scalar Szeg\H{o} function
\begin{equation*}
	\mathcal{D}(z)=\left(\frac{1+(2z)^{\frac{1}{2}}}{1-(2z)^{\frac{1}{2}}}\right)^{\alpha},\ \ z\in\mathbb{C}\backslash\mathbb{R},\hspace{1cm} \mathcal{D}(z)=\e^{\im\pi\alpha}\left(1+\frac{2\alpha}{(2z)^{\frac{1}{2}}}+\mathcal{O}\left(z^{-1}\right)\right),\ z\rightarrow\infty;
\end{equation*}
satisfies the properties listed below.
\begin{problem}\label{outer} The parametrix $P^{(\infty)}(z)$ has the following analytical properties
\begin{enumerate}
	\item $P^{(\infty)}(z)$ is analytic for $z\in\mathbb{C}\backslash((-\infty,0]\cup[\frac{1}{2},+\infty))$ with orientation of the real axis as indicated in Figure \ref{figure1}.
	\item The function $P^{(\infty)}(z)$ assumes square integrable limiting values on $(-\infty,0]\cup[\frac{1}{2},+\infty)$ which are related by the jump conditions
	\begin{eqnarray*}
		P_+^{(\infty)}(z)&=&P_-^{(\infty)}(z)\bigl(\begin{smallmatrix} 0 & 1\\ -1 & 0 \end{smallmatrix}\bigr),\ \ z\in(-\infty,0)\\
		P_+^{(\infty)}(z)&=&P_-^{(\infty)}(z)\bigl(\begin{smallmatrix} \e^{2\pi\im\alpha} & 0 \\ 0 & \e^{-2\pi\im\alpha} \end{smallmatrix}\bigr),\ \ z\in\left(\frac{1}{2},+\infty\right).
	\end{eqnarray*}
	\item As $z\rightarrow\infty,z\notin\mathbb{R}$, 
	\begin{equation*}
		P^{(\infty)}(z)=(|s|z)^{-\frac{1}{4}\sigma_3}\e^{-\im\pi\alpha\sigma_3}\frac{1}{\sqrt{2}}\begin{pmatrix} 1 & 1\\ -1 & 1 \end{pmatrix}\e^{-\im\frac{\pi}{4}\sigma_3}\e^{\im\pi\alpha\sigma_3}\Big\{I+\mathcal{O}\left(z^{-\frac{1}{2}}\right)\Big\}.
	\end{equation*}
\end{enumerate}
\end{problem}
The local parametrix near $z=0$ differs from the one used in \cite{CIK}, Section $3.4$, by a rank one perturbation, compare \eqref{e:7} below. In more detail we first define
\begin{equation}\label{bessbare}
	J(\z)=\e^{-\im\frac{\pi}{4}\sigma_3}\pi^{\frac{1}{2}\sigma_3}\begin{pmatrix}
	I_0(\z^{\frac{1}{2}}) & \frac{\im}{\pi}K_0(\z^{\frac{1}{2}})\\
	\im\pi \z^{\frac{1}{2}}I_0'(\z^{\frac{1}{2}}) & -\z^{\frac{1}{2}}K_0'(\z^{\frac{1}{2}})
	\end{pmatrix},\ \ \ \ \z\in\mathbb{C}\backslash(-\infty,0]
\end{equation}
in terms of the modified Bessel functions $I_0$ and $K_0$ and with principal branches for $\z^{\frac{1}{2}}:\,\textnormal{arg}\,\z\in(-\pi,\pi]$. The standard properties of Bessel functions \cite{NIST} in mind we obtain a bare {\it Bessel parametrix}:
\begin{problem} The function $J(\z)\in\mathbb{C}^{2\times 2}$ defined in \eqref{bessbare} has the following properties
\begin{enumerate}
	\item $J(\z)$ is analytic for $\z\in\mathbb{C}\backslash(-\infty,0]$.
	\item On the negative half ray, oriented from $-\infty$ to the origin, we have
	\begin{equation}\label{e:4}
		J_+(\z)=J_-(\z)\begin{pmatrix}
		1 & 1\\
		0 & 1
		\end{pmatrix},\ \ \ \ \z\in(-\infty,0).
	\end{equation}
	\item As $\z\rightarrow 0$,
	\begin{equation}\label{e:5}
		J(\z)=\widehat{J}(\z)\begin{pmatrix}
		1 & \frac{1}{2\pi\im}\ln\z\\
		0 & 1
		\end{pmatrix},\ \ \ \textnormal{arg}\,\z\in(-\pi,\pi),
	\end{equation}
	and $\widehat{J}(\z)$ is analytic at $\z=0$. In more detail,
	\begin{equation*}
	  \widehat{J}(\z)=\e^{-\im\frac{\pi}{4}\sigma_3}\pi^{\frac{1}{2}\sigma_3}
	\begin{pmatrix}
	I_0(\z^{\frac{1}{2}}) & \frac{\im}{\pi}\left(I_0(\z^{\frac{1}{2}})\ln 2+\sum_{k=0}^{\infty}\psi(k+1)\frac{(\frac{1}{4}\z)^k}{(k!)^2}\right)\\
	\im\pi\z^{\frac{1}{2}}I_0'(\z^{\frac{1}{2}}) & -\left(\z^{\frac{1}{2}}I_0'(\z^{\frac{1}{2}})\ln 2-I_0(\z^{\frac{1}{2}})+
	\sum_{k=1}^{\infty}\psi(k+1)\frac{(\frac{1}{4}\z)^k2k}{(k!)^2}\right)
	                                                                        \end{pmatrix},\ |\z|<r
	\end{equation*}
	with
	\begin{equation*}
	  I_0(\z^{\frac{1}{2}})=\sum_{k=0}^{\infty}\frac{(\frac{1}{4}\z)^k}{(k!)^2};\ \ \ 
	  \z^{\frac{1}{2}}I_0'(\z^{\frac{1}{2}})=\sum_{k=1}^{\infty}2k\frac{(\frac{1}{4}\z)^k}{(k!)^2},\,\,\z\in\mathbb{C};\ \ \ \psi(\z)=\frac{\Gamma'(\z)}{\Gamma(\z)},\ \ \z\in\mathbb{C}\backslash\{0,-1,-2,\ldots\}.
	\end{equation*}
	\item The function $J(\z)$ is normalized so that
	\begin{equation}\label{e:6}
		J(\z) = \z^{-\frac{1}{4}\sigma_3}\frac{1}{\sqrt{2}}\begin{pmatrix}
		1 & 1\\
		-1 & 1
		\end{pmatrix}\e^{-\im\frac{\pi}{4}\sigma_3}\left\{I+\frac{1}{8\z^{\frac{1}{2}}}\begin{pmatrix}
		-1 & -2\im\\
		-2\im & 1\\
		\end{pmatrix}+\mathcal{O}\left(\z^{-1}\right)\right\}\e^{\z^{\frac{1}{2}}\sigma_3},
	\end{equation}	
	as $\z\rightarrow\infty$ with $-\pi+\delta\leq\textnormal{arg}\,\z\leq\pi-\delta$ and $\delta>0$ fixed.
\end{enumerate}
\end{problem}
We can now define the parametrix $P^{(0)}(z)$ near the origin in terms of the model function $J(\z)$: For $0<|z|<\frac{1}{8}$,
\begin{equation}\label{e:7}
	\e^{\im\pi tV\sigma_3}P^{(0)}(z)\e^{-tg(z)\sigma_3}=E^{(0)}(z)J\big(\z(z)\big)\begin{pmatrix}
	1 & \frac{\gamma-1}{2\pi\im}\ln z\\
	0 & 1
	\end{pmatrix}\begin{cases}
		I,&\textnormal{arg}\,z\in(-\frac{2\pi}{3},\frac{2\pi}{3})\\
		\bigl(\begin{smallmatrix}
		1&0\\
		-1 & 1
		\end{smallmatrix}\bigr),&\textnormal{arg}\,z\in(\frac{2\pi}{3},\pi)\\
		\bigl(\begin{smallmatrix}
		1 & 0\\
		1 & 1
		\end{smallmatrix}\bigr),&\textnormal{arg}\,z\in(-\pi,-\frac{2\pi}{3})
		\end{cases}
\end{equation}
with the locally analytic left multiplier
\begin{equation*}
	E^{(0)}(z)=\e^{\im\pi k\sigma_3}t^{-\frac{1}{6}\sigma_3}z^{-\frac{1}{4}\sigma_3}\big(\z(z)\big)^{\frac{1}{4}\sigma_3},\ \ \ \ \ |z|<\frac{1}{8},
\end{equation*}
the locally conformal change of coordinates, for $|z|<\frac{1}{8}$,
\begin{equation*}
	\z(z)=t^2\left\{g(z)-\frac{\alpha}{t}\ln\left(\frac{1+(2z)^{\frac{1}{2}}}{1-(2z)^{\frac{1}{2}}}\right)\right\}^2=t^2z\left(1-2^{\frac{3}{2}}\frac{k}{t}\right)^2\left(1-\frac{4}{3}\frac{1+2^{\frac{3}{2}}\frac{k}{t}}{1-2^{\frac{3}{2}}\frac{k}{t}}z+\mathcal{O}\left(z^2\right)\right),
\end{equation*}
and the branch of the logarithm in \eqref{e:7} such that $\textnormal{arg}\,z\in(0,2\pi)$. It is straightforward to verify the analytical properties of $P^{(0)}(z)$:
\begin{problem}\label{origin} The parametrix $P^{(0)}(z)$ has the following analytical properties
\begin{enumerate}
	\item $P^{(0)}(z)$ is analytic for $z\in D(0,\frac{1}{8})\backslash(\Sigma_T\cup\{0\})$ with $D(z_0,r)=\{z\in\mathbb{C}:\,|z-z_0|<r\}$.
	\item Since $\z=\z(z)$ locally conformal near $z=0$, we obtain directly the jump behavior
	\begin{eqnarray*}
		P_+^{(0)}(z)&=&P_-^{(0)}(z)\begin{pmatrix}
		1 & 0\\
		\e^{2tg(z)} & 1
		\end{pmatrix},\ \ \ z\in \big((\Gamma_2\cup\Gamma_4)\backslash\{0\}\big)\cap D\left(0,\frac{1}{8}\right);\\
		 P_+^{(0)}(z)&=&P_-^{(0)}(z)\begin{pmatrix}
		 0 & 1\\
		 -1 & 0
		 \end{pmatrix},\ \ \ \ \ \ \ z\in\big(\Gamma_3\backslash\{0\}\big)\cap D\left(0,\frac{1}{8}\right)
	\end{eqnarray*}
	compare Figure \ref{figure1} for orientation, and where we used \eqref{e:4} in the last identity. Also, by the choice of branches in \eqref{e:6},
	\begin{equation*}
		P_+^{(0)}(z)=P_-^{(0)}(z)\begin{pmatrix}
		1 & \e^{-t(\varkappa_{\textnormal{Ai}}+2g(z))}\\
		0 & 1
		\end{pmatrix},\ \ \ \ z\in\big(\Gamma_1\backslash\{0\}\big)\cap D\left(0,\frac{1}{8}\right).
	\end{equation*}
	All together, $P^{(0)}(z)$ models precisely the jump behavior of $S(z)$ for $z\in D(0,\frac{1}{8})\backslash\{0\}$, see RHP \ref{gRHP}.
	\item Near $z=0$ with $\textnormal{arg}\,z\in(0,2\pi)$, we deduce from \eqref{e:5} and \eqref{e:7},
	\begin{equation*}
		\e^{\im\pi tV\sigma_3}P^{(0)}(z)\e^{-tg(z)\sigma_3}=\widehat{P}^{(0)}(z)\begin{pmatrix}
		1 & \frac{\gamma}{2\pi\im}\ln z\\
		0 & 1
		\end{pmatrix}\begin{cases}
		I,&\textnormal{arg}\,z\in(0,\frac{2\pi}{3})\\
		\bigl(\begin{smallmatrix}
		1 & 0\\
		-1 & 1
		\end{smallmatrix}\bigr),&\textnormal{arg}\,z\in(\frac{2\pi}{3},\pi)\\
		\bigl(\begin{smallmatrix}
		1 & -1\\
		0 & 1
		\end{smallmatrix}\bigr)\bigl(\begin{smallmatrix}
		1 & 0\\
		1 & 1
		\end{smallmatrix}\bigr),&\textnormal{arg}\,z\in(\pi,\frac{4\pi}{3})\\
		\bigl(\begin{smallmatrix}
		1 & -1\\
		0 & 1
		\end{smallmatrix}\bigr),&\textnormal{arg}\,z\in(\frac{4\pi}{3},2\pi)\\
		\end{cases}
	\end{equation*}
	which matches exactly the singular behavior of $S(z)$ near $z=0$. 
	\item As $t\rightarrow+\infty,\gamma\uparrow 1$ subject to \eqref{scale}, we derive from \eqref{e:6},
	\begin{equation}\label{e:8}
		P^{(0)}(z)=P^{(\infty)}(z)\left\{I+\frac{1}{8\z^{\frac{1}{2}}(z)}\big(\mathcal{D}(z)\big)^{-\sigma_3}\begin{pmatrix}
		-1 & -2\im\\
		-2\im & 1
		\end{pmatrix}\big(\mathcal{D}(z)\big)^{\sigma_3}+\mathcal{O}\left(t^{-2}\right)\right\}
	\end{equation}
	uniformly for $0<r_1\leq|z|\leq r_2<\frac{1}{8}$. Here, we used in particular that on the latter annulus,
	\begin{equation*}
		\e^{-t(\varkappa_{\textnormal{Ai}}+g_+(z)+g_-(z))}=\mathcal{O}\left(t^{-\infty}\right), \ \ t\rightarrow+\infty,\gamma\uparrow 1:\ \ \varkappa_{_\textnormal{Ai}}=\frac{2}{3}\sqrt{2}-\chi\frac{\ln t}{t}.
	\end{equation*} 
\end{enumerate}
\end{problem}
\begin{remark}\label{cem:1} Note that from properties $(2)$ and $(3)$ in RHP \ref{origin} we obtain
\begin{equation*}
	S(z)=N_0(z)P^{(0)}(z),\ \ \ 0\leq|z|<\frac{1}{4}
\end{equation*}
where $N_0(z)$ is analytic at $z=0$.
\end{remark}
The outstanding parametrix near $z=\frac{1}{2}$ is in some sense more elementary than the Bessel-type parametrix \eqref{e:7}. We draw inspiration from \cite{C,BL,BT,BM} and define first
\begin{equation}\label{hermitebare}
	H(\z)=\begin{pmatrix} p_k(\z) & \frac{1}{2\pi\im}\int_{\mathbb{R}}p_k(t)\e^{-t^2}\frac{\d t}{t-\z}\smallskip\\ 
	\gamma_{k-1} p_{k-1}(\z) & \frac{\gamma_{k-1}}{2\pi\im}\int_{\mathbb{R}}p_{k-1}(t)\e^{-t^2}\frac{\d t}{t-\z} \end{pmatrix}\e^{-\frac{1}{2}\z^2\sigma_3},\ \ \ \z\in\mathbb{C}\backslash\mathbb{R},\ \ k\in\mathbb{Z}_{\geq 0}
\end{equation}
with the help of monic Hermite polynomials $\{p_k(\z)\}_{k\in\mathbb{Z}_{\geq 0}},p_{-1}(\z)\equiv 0=\gamma_{-1}$, cf. \cite{NIST}:
\begin{equation*}
	p_k(\z)=\z^k+a_{k,k-2}z^{k-2}+\mathcal{O}\left(\z^{k-4}\right),\ \ \z\rightarrow\infty;\hspace{1cm} \int_{\mathbb{R}}p_j(t)p_k(t)\e^{-t^2}\d t = h_k\delta_{jk},
\end{equation*}
where
\begin{equation*}
	h_k=k!\frac{\sqrt{\pi}}{2^k},\ \ k\in\mathbb{Z}_{\geq 0};\ \ \ \ \ a_{k,k-2}=-\frac{1}{4}k(k-1),\ \ k\in\mathbb{Z}_{\geq 0}.
\end{equation*}	
\begin{remark} The normalization in \cite{NIST} of the Hermite polynomials $\{H_k(\z)\}_{k\in\mathbb{Z}_{\geq 0}}$ is different from the one chosen here: we have to use the relation $p_k(\z)=2^{-k}H_k(\z),\z\in\mathbb{C}$.
\end{remark}
These properties lead at once to a bare {\it Hermite parametrix}:
\begin{problem} For any $k\in\mathbb{Z}_{\geq 0}$, the function $H(\z)\in\mathbb{C}^{2\times 2}$ defined in \eqref{hermitebare} has the following properties
\begin{enumerate}
	\item $H(\z)$ is analytic for $\z\in\mathbb{C}\backslash\mathbb{R}$ and we orient the real axis from $-\infty$ to $+\infty$.
	\item Along the real line we have
	\begin{equation*}
		H_+(\z)=H_-(\z)\begin{pmatrix} 1 & 1\\ 0 & 1 \end{pmatrix},\ \ \ \z\in\mathbb{R}.
	\end{equation*}
	\item As $\z\rightarrow\infty,\z\notin\mathbb{R}$, with $\gamma_k=-\frac{2\pi\im}{h_k}$,
	\begin{eqnarray}\label{e:9}
		H(\z)&=&\bigg\{I+\frac{1}{\z}\begin{pmatrix} 0&\gamma_k^{-1} \\ \gamma_{k-1} & 0 \end{pmatrix}+\frac{1}{\z^2}\begin{pmatrix} -\frac{\gamma_{k-2}}{\gamma_k} & 0\\ 0 & \frac{\gamma_{k-1}}{\gamma_{k+1}}\end{pmatrix}+\frac{1}{\z^3}\begin{pmatrix} 0 & \gamma_{k+2}^{-1} \\ -\gamma_{k-3} & 0\end{pmatrix}\nonumber\\
		&&+\mathcal{O}\left(\z^{-4}\right)\bigg\}\z^{k\sigma_3}\e^{-\frac{1}{2}\z^2\sigma_3}.
	\end{eqnarray}
\end{enumerate}
\end{problem}
For the actual model function we then take with $|z-\frac{1}{2}|<\frac{1}{8}$,
\begin{equation}\label{local}
	P^{(\frac{1}{2})}(z)=E^{(\frac{1}{2})}(z)H\big(\z(z)\big)\e^{\frac{t}{2}\varkappa_{\textnormal{Ai}}\sigma_3}\e^{tg(z)\sigma_3},
\end{equation}
where
\begin{equation*}
	E^{(\frac{1}{2})}(z)=(|s|z)^{-\frac{1}{4}\sigma_3}\e^{-\im\pi\alpha\sigma_3}\frac{1}{\sqrt{2}}\begin{pmatrix} 1 & 1\\ -1 & 1 \end{pmatrix}\e^{-\im\frac{\pi}{4}\sigma_3}\big(\beta(z)\big)^{-\sigma_3}
	,\ \ \ 0\leq\left|z-\frac{1}{2}\right|<\frac{1}{8},
\end{equation*} 
with the choice
\begin{equation*}
	\beta(z)=\left(\z(z)\frac{1+(2z)^{\frac{1}{2}}}{1-(2z)^{\frac{1}{2}}}\right)^kt^{-\frac{\chi}{2}}=(-1)^k2^{\frac{5k}{4}}t^{-\frac{\alpha}{2}}\left\{1+\frac{2k}{3}\left(z-\frac{1}{2}\right)+\mathcal{O}\left(\left(z-\frac{1}{2}\right)^2\right)\right\},
\end{equation*}
is analytic at $z=\frac{1}{2}$. Throughout, $\z=\z(z),|z-\frac{1}{2}|<\frac{1}{8}$ denotes the locally conformal change of variables
\begin{equation*}
	\z(z)=\sqrt{2t}\left(g(z)-\frac{\chi}{t}\ln\left(\frac{1+(2z)^{\frac{1}{2}}}{1-(2z)^{\frac{1}{2}}}\right)+\frac{\sqrt{2}}{3}\right)^{\frac{1}{2}}=2^{\frac{1}{4}}\sqrt{t}\left(z-\frac{1}{2}\right)\left\{1-\frac{1}{3}\left(z-\frac{1}{2}\right)+\mathcal{O}\left(\left(z-\frac{1}{2}\right)^2\right)\right\},
\end{equation*}
and we use again the representation
\begin{equation}\label{e:10}
	\mathbb{R}_{\geq 0}\ni\chi=k+\alpha,\ \ \ k\in\mathbb{Z}_{\geq 0},\ \ -\frac{1}{2}\leq\alpha<\frac{1}{2}.
\end{equation}
The important properties of $P^{(\frac{1}{2})}(z)$ are summarized below.
\begin{problem}\label{statRHP} The parametrix $P^{(\frac{1}{2})}(z)$ has the following analytical properties
\begin{enumerate}
	\item $P^{(\frac{1}{2})}(z)$ is analytic for $z\in D(\frac{1}{2},\frac{1}{8})\backslash\Sigma_T$.
	\item By local analyticity of $\z=\z(z)$, we have
	\begin{eqnarray*}
		P_+^{(\frac{1}{2})}(z)&=&P_-^{(\frac{1}{2})}(z)\begin{pmatrix}1 & \e^{-t(\varkappa_{\textnormal{Ai}}+2g(z))} \\ 0 & 1 \end{pmatrix},\ \ z\in\left(0,\frac{1}{2}\right)\cap D\left(\frac{1}{2},\frac{1}{8}\right);\\
		P_+^{(\frac{1}{2})}(z)&=&P_-^{(\frac{1}{2})}(z)\begin{pmatrix}\e^{2\pi\im tV} & \e^{-t(\varkappa_{\textnormal{Ai}}+g_+(z)+g_-(z))}\\ 0 & \e^{-2\pi\im tV} \end{pmatrix},\ \ z\in\left(\frac{1}{2},+\infty\right)\cap D\left(\frac{1}{2},\frac{1}{8}\right).
	\end{eqnarray*}
	This matches exactly the jump behavior of $S(z)$ for $z\in D(\frac{1}{2},\frac{1}{8})$, compare RHP \ref{gRHP}.
	\item As $t\rightarrow+\infty,\gamma\uparrow 1$ subject to \eqref{scale}, we derive from \eqref{e:9},
	\begin{align}\label{e:11}
		P^{(\frac{1}{2})}(z)&=P^{(\infty)}(z)\bigg\{I+\frac{1}{\z(z)}\big(\mathcal{D}(z)\big)^{-\sigma_3}\begin{pmatrix} 0 & \gamma_k^{-1}\beta^{-2}(z)\\ \gamma_{k-1}\beta^2(z)& 0 \end{pmatrix}\big(\mathcal{D}(z)\big)^{\sigma_3}+\frac{1}{\z^2(z)}\begin{pmatrix} -\frac{\gamma_{k-2}}{\gamma_k} & 0 \\ 0 & \frac{\gamma_{k-1}}{\gamma_{k+1}}\end{pmatrix}\nonumber\\
		&+\frac{1}{\z^3(z)}\big(\mathcal{D}(z)\big)^{-\sigma_3}\begin{pmatrix} 0 & \gamma_{k+2}^{-1}\,\beta^{-2}(z)\\ -\gamma_{k-3}\,\beta^2(z) & 0 \end{pmatrix}\big(\mathcal{D}(z)\big)^{\sigma_3}+\mathcal{O}\left(t^{-2+|\alpha|}\right)\bigg\}
	\end{align}
	uniformly for $0<r_1\leq|z-\frac{1}{2}|\leq r_2<\frac{1}{8}$.
\end{enumerate}
\end{problem}
\begin{remark}\label{cem:2} Note that from property $(2)$ in RHP \ref{statRHP} we obtain
\begin{equation*}
	S(z)=N_{\frac{1}{2}}(z)P^{(\frac{1}{2})}(z),\ \ \ 0\leq\left|z-\frac{1}{2}\right|<\frac{1}{8}
\end{equation*}
where $N_{\frac{1}{2}}(z)$ is analytic at $z=\frac{1}{2}$.
\end{remark}
This completes the construction of local model functions, we now use the explicit functions $P^{(\infty)}(z), P^{(0)}(z)$ and $P^{(\frac{1}{2})}(z)$ and compare them to the unknown $S(z)$ in RHP \ref{gRHP}.
\subsection{Ratio transformation and first small norm estimate} With \eqref{out}, \eqref{e:7} and \eqref{local} this steps amounts to the transformation
\begin{equation*}
	R(z)=\begin{pmatrix} 1 & 0\\ \omega & 1 \end{pmatrix}S(z)\begin{cases} \big(P^{(0)}(z)\big)^{-1},&|z|<r\\ \big(P^{(\frac{1}{2})}(z)\big)^{-1},&|z-\frac{1}{2}|<r\\
	\big(P^{(\infty)}(z)\big)^{-1},&|z|>r,|z-\frac{1}{2}|>r\end{cases};\ \ \omega=-|s|^{\frac{1}{2}}\big(N(S_{\infty}-\alpha\sqrt{2}\,\sigma_3)N^{-1}\big)_{21}
\end{equation*}
in which $0<r<\frac{1}{8}$ is kept fixed. Here we have made use of the abbreviation
\begin{equation*}
	N=\e^{-\im\pi\alpha\sigma_3}\frac{1}{\sqrt{2}}\begin{pmatrix} 1 & 1 \\ -1 & 1 \end{pmatrix}\e^{-\im\frac{\pi}{4}\sigma_3}\e^{\im\pi\alpha\sigma_3},
\end{equation*}
and $S_{\infty}$ occurred in RHP \ref{gRHP}. Recalling RHP \ref{outer}, \ref{origin} and \ref{statRHP} we are lead to the following problem.
\begin{figure}[tbh]
\begin{center}
\resizebox{0.4\textwidth}{!}{\includegraphics{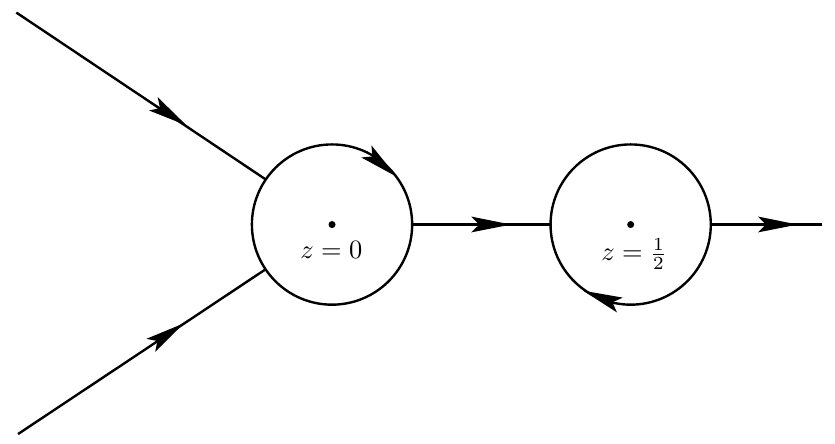}}
\caption{The oriented jump contours for the ratio function $R(z)$ in the complex $z$-plane.}
\label{figure4}
\end{center}
\end{figure}
\begin{problem} Determine $R(z)=R(z;s,\gamma)\in\mathbb{C}^{2\times 2}$ such that
\begin{enumerate}
	\item $R(z)$ is analytic for $z\in\mathbb{C}\backslash\Sigma_R$ with square integrable boundary values on the contour 
	\begin{equation*}
		\Sigma_R=\partial D(0,r)\cup\partial D\left(\frac{1}{2},r\right)\cup\left(r,\frac{1}{2}-r\right)\cup\left(\frac{1}{2}+r,\infty\right)\cup\Big(\big(\Gamma_2\cup\Gamma_4\big)\cap\{z\in\mathbb{C}:\,|z|>r\}\Big)
	\end{equation*}
	which is shown in Figure \ref{figure4}.
	\item On the contour $\Sigma_R$ we have jumps $R_+(z)=R_-(z)G_R(z;s,\gamma),z\in\Sigma_R$ with
	\begin{align*}
		G_R(z;s,\gamma)&=P^{(0)}(z)\big(P^{(\infty)}(z)\big)^{-1},\ \ z\in\partial D(0,r);\ \ \ \ 
		G_R(z;s,\gamma)=P^{(\frac{1}{2})}(z)\big(P^{(\infty)}(z)\big)^{-1},\ \ z\in\partial D\left(\frac{1}{2},r\right);\\
		G_R(z;s,\gamma)&=P^{(\infty)}(z)\begin{pmatrix}
		1 & \e^{-t(\varkappa_{\textnormal{Ai}}+2g(z))}\\
		0 & 1
		\end{pmatrix}\big(P^{(\infty)}(z)\big)^{-1},\ \ z\in\left(r,\frac{1}{2}-r\right);\\
		G_R(z;s,\gamma)&=P_-^{(\infty)}(z)\begin{pmatrix} 1 & \e^{2\pi\im\alpha}\e^{-t(\varkappa_{\textnormal{Ai}}+g_+(z)+g_-(z))} \\ 0 & 1 \end{pmatrix}\big(P_-^{(\infty)}(z)\big)^{-1},\ \ z\in\left(\frac{1}{2}+r,+\infty\right);\\
		G_R(z;s,\gamma)&=P^{(\infty)}(z)\begin{pmatrix}
		1 & 0\\
		\e^{tg(z)} & 1
		\end{pmatrix}\big(P^{(\infty)}(z)\big)^{-1},\ \ z\in\big(\Gamma_2\cup\Gamma_4\big)\cap\{z\in\mathbb{C}:\,|z|>r\}.
	\end{align*}
	By construction, see Remarks \ref{cem:1} and \ref{cem:2}, there are no jumps inside $D(0,r)\cup D(\frac{1}{2},r)$ and on $(-\infty,-r)$. Moreover $R(z)$ is bounded at $z=0$.
	\item As $z\rightarrow\infty$,
	\begin{align*}
		R(z)&=\begin{pmatrix}
	1 & 0\\
	\omega & 1
	\end{pmatrix}\e^{-\im\pi k\sigma_3}P^{(\infty)}(z)\big(\mathcal{D}(z)\big)^{-\sigma_3}\e^{\im\pi tV\sigma_3}\left\{I+S_{\infty}z^{-\frac{1}{2}}+\mathcal{O}\left(z^{-1}\right)\right\}\big(P^{(\infty)}(z)\big)^{-1}\\
	&=I+\mathcal{O}\left(z^{-\frac{1}{2}}\right).
	\end{align*}
\end{enumerate}
\end{problem}
Through standard small norm estimations, compare \eqref{crux:1}, \eqref{crux:2}, \eqref{crux:3} and \eqref{e:8}, we obtain at once
\begin{prop}\label{DZ:1} Given $\chi\in\mathbb{R}_{\geq 0}$ there exist positive $t_0=t_0(\chi),v_0=v_0(\chi)$ and $c=c(\chi)$ such that
\begin{equation*}
	\|G_R(\cdot;s,\gamma)-I\|_{L^2\cap L^{\infty}(\Sigma_R\backslash\partial D(\frac{1}{2},r))}\leq\frac{c}{t^{\frac{2}{3}}},\ \ \ \forall\,t\geq t_0,v=-\ln(1-\gamma)\geq v_0:\ \ v=\frac{2}{3}\sqrt{2}\,t-\chi\ln t.
\end{equation*}
\end{prop}
The circle boundary $\partial D(\frac{1}{2},r)$ requires further analysis: From \eqref{e:11}, uniformly for $z\in\partial D(\frac{1}{2},r)$,
\begin{eqnarray}\label{fac:0}
	G_R(z;s,\gamma)-I&=&(|s|z)^{-\frac{1}{4}\sigma_3}\Big\{R_1^+(z)t^{-\frac{1}{2}+\alpha}+R_1^-(z)t^{-\frac{1}{2}-\alpha}+R_2(z)t^{-1}+R_3^+(z)t^{-\frac{3}{2}+\alpha}\nonumber\\
	&&+\,R_3^-(z)t^{-\frac{3}{2}-\alpha}+\mathcal{O}\left(t^{-2+|\alpha|}\right)\Big\}(|s|z)^{\frac{1}{4}\sigma_3},
\end{eqnarray}	
with the $t$-independent matrices
\begin{equation*}
	R_1^+(z)=-\im\gamma_k^{-1}\frac{\hat{\beta}^{-2}(z)}{2\hat{\z}(z)}\e^{-\im\pi\alpha\sigma_3}\begin{pmatrix} 1 & 1\\ -1 & -1 \end{pmatrix} \e^{\im\pi\alpha\sigma_3};\ \ \ 
	R_1^-(z)=\im\gamma_{k-1}\frac{\hat{\beta}^2(z)}{2\hat{\zeta}(z)}\e^{-\im\pi\alpha\sigma_3}\begin{pmatrix} 1 & -1 \\ 1 & -1 \end{pmatrix}\e^{\im\pi\alpha\sigma_3}
\end{equation*}
which are both of rank one, and
\begin{equation*}
	R_2(z)=\frac{k}{4\hat{\z}^2(z)}\e^{-\im\pi\alpha\sigma_3}\begin{pmatrix} 1 & k \\ k & 1 \end{pmatrix}\e^{\im\pi\alpha\sigma_3},\ \ \ R_3^+(z)=-\im\gamma_{k+2}^{-1}\frac{\hat{\beta}^{-2}(z)}{2\hat{\z}^3(z)}\e^{-\im\pi\alpha\sigma_3}\begin{pmatrix} 1 & 1\\ -1 & -1 \end{pmatrix}\e^{\im\pi\alpha\sigma_3}
\end{equation*}
as well as
\begin{equation*}
	R_3^-(z)=-\im\gamma_{k-3}\frac{\hat{\beta}^2(z)}{2\hat{\z}^3(z)}\e^{-\im\pi\alpha\sigma_3}\begin{pmatrix} 1 & -1 \\ 1 & -1 \end{pmatrix} \e^{\im\pi\alpha\sigma_3}.
\end{equation*}
We have introduced $\beta(z)=t^{-\frac{\alpha}{2}}\hat{\beta}(z)$ and $\z(z)=t^{\frac{1}{2}}\hat{\z}(z)$ where $\hat{\beta}(z)$ and $\hat{\z}(z)$ are $t$-independent. Note that for $\alpha\in[-\frac{1}{2},\frac{1}{2})$ the first two leading terms in \eqref{fac:0} are in general not close to zero. In order to overcome this feature we use matrix factorizations,
\begin{eqnarray}
	G_R(z;s,\gamma)&=&E^+(z)(|s|z)^{-\frac{1}{4}\sigma_3}\Big\{I+R_1^-(z)t^{-\frac{1}{2}-\alpha}+\big(R_2(z)-R_1^+(z)R_1^-(z)\big)t^{-1}\label{fac:1}\\
	&&+\,\big(R_3^+(z)-R_1^+(z)R_2(z)\big)t^{-\frac{3}{2}+\alpha}+R_3^-(z)t^{-\frac{3}{2}-\alpha}+\mathcal{O}\left(t^{-2+|\alpha|}\right)\Big\}(|s|z)^{\frac{1}{4}\sigma_3},\ \ 0\leq\alpha<\frac{1}{2};\nonumber\\
	G_R(z;s,\gamma)&=&E^-(z)(|s|z)^{-\frac{1}{4}\sigma_3}\Big\{I+R_1^+(z)t^{-\frac{1}{2}+\alpha}+\big(R_2(z)-R_1^-(z)R_1^+(z)\big)t^{-1}\label{fac:2}\\
	&&+\,\big(R_3^-(z)-R_1^-(z)R_2(z)\big)t^{-\frac{3}{2}-\alpha}+R_3^+(z)t^{-\frac{3}{2}+\alpha}+\mathcal{O}\left(t^{-2+|\alpha|}\right)\Big\}(|s|z)^{\frac{1}{4}\sigma_3},\ \ -\frac{1}{2}\leq\alpha\leq 0,\nonumber
\end{eqnarray}
where
\begin{equation*}
	E^+(z)=(|s|z)^{-\frac{1}{4}\sigma_3}\big(I+R_1^+(z)t^{\alpha-\frac{1}{2}}\big)(|s|z)^{\frac{1}{4}\sigma_3},\ \ \ E^-(z)=(|s|z)^{-\frac{1}{4}\sigma_3}\big(I+R_1^-(z)t^{-\alpha-\frac{1}{2}}\big)(|s|z)^{\frac{1}{4}\sigma_3}
\end{equation*}
are invertible and meromorphic for $z\in D(\frac{1}{2},r)$. The factorizations \eqref{fac:1}, \eqref{fac:2} above, i.e.
\begin{equation*}
	G_R(z;s,\gamma)\equiv E^{\pm}(z)E_0^{\pm}(z),\ \ \ z\in\partial D\left(\frac{1}{2},r\right),
\end{equation*}
motivate our next move.
\subsection{Singular Riemann-Hilbert problem and iterative solution}\label{Airysing}
In this step we introduce
\begin{equation*}
	Q(z)=R(z)\begin{cases}\begin{cases} E^+(z),&|z-\frac{1}{2}|<r\\ I,&|z-\frac{1}{2}|>r \end{cases},&\,\,\,\,\,0\leq\alpha<\frac{1}{2}\bigskip\\ \begin{cases} E^-(z),&|z-\frac{1}{2}|<r \\ I,&|z-\frac{1}{2}|>r \end{cases},&-\frac{1}{2}\leq\alpha\leq 0 \end{cases}
\end{equation*}
which leads us to a singular RHP.
\begin{problem}\label{QRH} Determine $Q(z)=Q(z;s,\gamma)\in\mathbb{C}^{2\times 2}$ such that
\begin{enumerate}
	\item $Q(z)$ is analytic for $z\in\mathbb{C}\backslash(\Sigma_R\cup\{\frac{1}{2}\})$ with square integrable boundary values on the jump contour $\Sigma_R$ shown in Figure \ref{figure4}.
	\item The jumps on $\Sigma_R$ are as follows,
	\begin{eqnarray*}
		Q_+(z)&=&Q_-(z)G_R(z;s,\gamma),\ \ \ \ \ z\in\Sigma_R\backslash\partial D\left(\frac{1}{2},r\right);\\
		Q_+(z)&=&Q_-(z)\begin{cases} E_0^+(z),&\,\,\,\,\,0\leq\alpha<\frac{1}{2}\smallskip\\ E_0^-(z),&-\frac{1}{2}\leq\alpha\leq 0\end{cases},\ \ \ \ \ z\in\partial D\left(\frac{1}{2},r\right).
	 \end{eqnarray*}
	 \item The function $Q(z)$ has a first order pole at $z=\frac{1}{2}$. In more detail, near $z=\frac{1}{2}$,
	 \begin{equation}\label{sing:1}
	 	Q(z)=\widehat{Q}(z)\begin{pmatrix} 1 & \frac{\sigma^{\pm}}{z-\frac{1}{2}} \smallskip\\ 0 & 1 \end{pmatrix}\big(T^{\pm}\big)^{-1},\ \ \ \left|z-\frac{1}{2}\right|<r;
	 \end{equation}
	 where $\widehat{Q}(z)$ is analytic at $z=\frac{1}{2}$ and we have introduced
	 \begin{equation}\label{sing:2}
	 	\sigma^+=\frac{-\im\gamma_k^{-1}2^{-\frac{5k}{2}-\frac{5}{4}}t^{\alpha-\frac{1}{2}}}{1-\im\gamma_k^{-1}2^{-\frac{5k}{2}-\frac{5}{4}}t^{\alpha-\frac{1}{2}}},\ \ \alpha\in\left[0,\frac{1}{2}\right);\ \ \ \ \ \sigma^-=\frac{\im\gamma_{k-1}2^{\frac{5k}{2}-\frac{5}{4}}t^{-\alpha-\frac{1}{2}}}{1+\im\gamma_{k-1}2^{\frac{5k}{2}-\frac{5}{4}}t^{-\alpha-\frac{1}{2}}},\ \ \alpha\in\left[-\frac{1}{2},0\right];
	\end{equation}
	as well as
	\begin{equation}\label{sing:3}
		T^+=\begin{pmatrix} 1 & 1\\ -\e^{2\pi\im\alpha}\big(\frac{|s|}{2}\big)^{\frac{1}{2}} & 0 \end{pmatrix},\ \ \alpha\in\left[0,\frac{1}{2}\right);\ \ \ \ \ T^-=\begin{pmatrix}
		1 & 1 \\ \e^{2\pi\im\alpha}\big(\frac{|s|}{2}\big)^{\frac{1}{2}} & 0 \end{pmatrix},\ \ \alpha\in\left[-\frac{1}{2},0\right].
	\end{equation}
	 \item As $z\rightarrow\infty$
	 \begin{equation*}
	 	Q(z)=I+\mathcal{O}\left(z^{-\frac{1}{2}}\right).
	\end{equation*}
\end{enumerate}
\end{problem}
\begin{remark} The singular structure \eqref{sing:1}, \eqref{sing:2}, \eqref{sing:3} follows from the observation that
\begin{equation}\label{sing:4}
	R(z)=Q(z)\big(E^{\pm}(z)\big)^{-1},\ \ \left|z-\frac{1}{2}\right|<r
\end{equation}
is analytic at $z=\frac{1}{2}$. Hence the singular part of $Q(z)$ can be derived by comparison in \eqref{sing:4} which leads to \eqref{sing:1}.
\end{remark}
Note that all jump matrices in the $Q$-RHP are close to unity at the cost of an isolated singularity at $z=\frac{1}{2}$. This will now be resolved by a final transformation. Define $L(z),z\in\mathbb{C}\backslash\Sigma_R$ such that
\begin{equation}\label{res:1}
	Q(z)=\left\{\left(z-\frac{1}{2}\right)I+B^{\pm}\right\}L(z)\frac{1}{z-\frac{1}{2}},
\end{equation}
with $B^{\pm}\in\mathbb{C}^{2\times 2}$ constant in $z$.
\begin{problem}\label{final} Determine $L(z)=L(z;s,\gamma)\in\mathbb{C}^{2\times 2}$ such that
\begin{enumerate}
	\item $L(z)$ is analytic for $z\in\mathbb{C}\backslash\Sigma_R$ with square integrable boundary values on the jump contour $\Sigma_R$ shown in Figure \ref{figure4}.
	\item The jumps are identical to the ones in the previous $Q$-RHP \ref{QRH}, i.e.
	\begin{equation*}
		L_+(z)=L_-(z)G_Q(z;s,\gamma),\ \ \ z\in\Sigma_R.
	\end{equation*}
	\item The function $L(z)$ is analytic at $z=\frac{1}{2}$ provided
	\begin{equation}\label{Bmat}
		B^{\pm}=\sigma^{\pm}L\left(\frac{1}{2}\right)T^{\pm}\begin{pmatrix} 0 & 1\\ 0 & 0 \end{pmatrix}\big(T^{\pm}\big)^{-1}\left\{L\left(\frac{1}{2}\right)-\sigma^{\pm}L'\left(\frac{1}{2}\right)T^{\pm}\begin{pmatrix} 0 & 1\\ 0 & 0 \end{pmatrix}\big(T^{\pm}\big)^{-1}\right\}^{-1},
	\end{equation}
	which follows directly from \eqref{res:1} and \eqref{sing:1}.
	\item As $z\rightarrow\infty$, we have that $L(z)\rightarrow I$.
\end{enumerate}
\end{problem}
\begin{remark} Note that
\begin{equation*}
	B^{\pm}=\sigma^{\pm}L\left(\frac{1}{2}\right)T^{\pm}\begin{pmatrix} 0 & 1\\ 0 & 0 \end{pmatrix}\big(T^{\pm}\big)^{-1}\left\{I-\sigma^{\pm}L^{-1}\left(\frac{1}{2}\right)L'\left(\frac{1}{2}\right)T^{\pm}\begin{pmatrix} 0 & 1 \\ 0 & 0 \end{pmatrix}\big(T^{\pm}\big)^{-1}\right\}^{-1}L^{-1}\left(\frac{1}{2}\right)
\end{equation*}
and thus 
\begin{equation*}
	\textnormal{tr}\,B^{\pm}=0=\det B^{\pm}.
\end{equation*}
This implies that $\det((z-\frac{1}{2})I+B^{\pm})=(z-\frac{1}{2})^2$ and hence \eqref{res:1} is consistent with the identities 
\begin{equation*}
	\det L(z)=\det Q(z)\equiv 1.
\end{equation*}
\end{remark}
At this point it is clear that the $L$-RHP admits direct asymptotic analysis, indeed we have
\begin{prop}\label{DZ:2} Given $\chi=k+\alpha\in\mathbb{R}_{\geq 0}$ with $k\in\mathbb{Z}_{\geq 0},-\frac{1}{2}\leq\alpha<\frac{1}{2}$, there exist positive $t_0=t_0(\chi),v_0=v_0(\chi)$ and $c=c(\chi)$ such that
\begin{equation*}
	\|G_Q(\cdot;s,\gamma)-I\|_{L^2\cap L^{\infty}(\partial D(\frac{1}{2},r))}\leq c\,t^{-\frac{1}{6}-|\alpha|},\ \ \forall\,t\geq t_0,\ v=-\ln(1-\gamma)\geq v_0:\ v=\frac{2}{3}\sqrt{2}\,t-\chi\ln t.
\end{equation*}
\end{prop}
Hence together with Proposition \ref{DZ:1}, by standard arguments \cite{DZ}, the singular integral equation
\begin{equation}\label{fin:1}
	L(z)=I+\frac{1}{2\pi\im}\int_{\Sigma_R}L_-(w)\big(G_Q(w)-I\big)\frac{\d w}{w-z},\ \ \ z\in\mathbb{C}\backslash\Sigma_R
\end{equation}
which is equivalent to RHP \ref{final} is solvable for sufficiently large $t\geq t_0,v\geq v_0$. In fact
\begin{prop}\label{DZ:3} Given $\chi\in\mathbb{R}_{\geq 0}$ there exist positive $t_0=t_0(\chi),v_0=v_0(\chi)$ and $c=c(\chi)$ such that RHP \ref{final} is uniquely solvable for $t\geq t_0,v\geq v_0:v=\frac{2}{3}\sqrt{2}\,t-\chi\ln t$. The solution $L=L(z;s,\gamma)$ satisfies
\begin{equation*}
	\|L_-(\cdot;s,\gamma)-I\|_{L^2(\Sigma_R)}\leq c\,t^{-\frac{1}{6}-|\alpha|},\ \ \forall\,t\geq t_0,\ v\geq v_0:\ \ v=\frac{2}{3}\sqrt{2}\,t-\chi\ln t.
\end{equation*}
\end{prop}
The last Proposition can be used to derive an asymptotic expansion for the logarithmic $s$-derivative through \eqref{diff:1}. In the derivation of such an expansion a certain structural information will prove useful:\smallskip

Note that all jump matrices in the $Q$-RHP are in fact (formally) conjugated by $|s|^{-\frac{1}{4}\sigma_3}$, compare for instance \eqref{fac:1}, \eqref{fac:2}, i.e. we can write
\begin{equation*}
	G_Q(z;s,\gamma)=|s|^{-\frac{1}{4}\sigma_3}\widehat{G}_Q(z;s,\gamma)|s|^{\frac{1}{4}\sigma_3},\ \ \ z\in\Sigma_R.
\end{equation*}
This leads to a RHP for the function $\widehat{L}(z)=\widehat{L}(z;s,\gamma)=|s|^{\frac{1}{4}\sigma_3}L(z;s,\gamma)|s|^{-\frac{1}{4}\sigma_3}$ which is uniquely solvable by iteration,
\begin{equation}\label{DZ:better}
	\|\widehat{L}_-(\cdot;s,\gamma)-I\|_{L^2(\Sigma_R)}\leq c\, t^{-\frac{1}{2}-|\alpha|},\ \ \ \forall\, t\geq t_0,\ v\geq v_0:\ \ \ v=\frac{2}{3}\sqrt{2}\,t-\chi\ln t.
\end{equation}
\section{Extraction of large gap asymptotics}\label{Airyp3} We recall \eqref{diff:1},
\begin{equation*}
	F(J_{\textnormal{Ai}};\gamma)\equiv\frac{\partial}{\partial s}\ln D(J_{\textnormal{Ai}};\gamma)=\frac{\gamma}{2\pi\im}\big(X^{-1}(z)X'(z)\big)_{21}\Big|_{z\rightarrow s}
\end{equation*}
where the limit is carried out for $\textnormal{arg}(z-s)\in(0,\frac{2\pi}{3})$.
\subsection{Asymptotics for the derivative}\label{asyderiv} Through the sequence of transformations
\begin{equation*}
	X(z)\mapsto T(z)\mapsto S(z)\mapsto R(z)\mapsto Q(z)\mapsto L(z),
\end{equation*}
identity \eqref{diff:1} leads us to
\begin{align*}
	F(J_{\textnormal{Ai}};\gamma)=&\,\frac{\gamma}{2\pi\im|s|}\big(T^{-1}(w)T'(w)\big)_{21}\Big|_{w\rightarrow 0}=\frac{\gamma\e^{-2tg_+(0)}}{2\pi\im|s|}\big(S^{-1}(w)S'(w)\big)_{21}\Big|_{w\rightarrow 0}\\
	=&\,\frac{\gamma}{2\pi\im|s|}\left\{\big(P^{(0)}(w)\big)^{-1}\big(P^{(0)}(w)\big)'\right\}_{21}\Big|_{w\rightarrow 0}+\frac{\gamma}{2\pi\im|s|}\left\{\big(P^{(0)}(w)\big)^{-1}Q^{-1}(w)Q'(w)P^{(0)}(w)\right\}_{21}\Big|_{w\rightarrow 0}
\end{align*}
and the first summand is computed with the help of RHP \ref{origin} as
\begin{equation*}
	\frac{\gamma}{2\pi\im|s|}\left\{\big(P^{(0)}(w)\big)^{-1}\big(P^{(0)}(w)\big)'\right\}_{21}\Big|_{w\rightarrow 0}=\frac{1}{4}\gamma|s|^2\left(1-2^{\frac{3}{2}}\frac{k}{t}\right)^2.
\end{equation*}
For the second, with \eqref{res:1},
\begin{align*}
	\frac{\gamma}{2\pi\im|s|}\Big\{\big(P^{(0)}(w)\big)^{-1}&Q^{-1}(w)Q'(w)P^{(0)}(w)\Big\}_{21}\Big|_{w\rightarrow 0}=-\frac{2\gamma}{\pi\im|s|}\left\{\big(P^{(0)}(w)\big)^{-1}L^{-1}(w)B^{\pm}L(w)P^{(0)}(w)\right\}_{21}\Big|_{w\rightarrow 0}\\
	&+\frac{\gamma}{2\pi\im|s|}\left\{\big(P^{(0)}(w)\big)^{-1}L^{-1}(w)L'(w)P^{(0)}(w)\right\}_{21}\Big|_{w\rightarrow 0}\\
	&\,=2\gamma\e^{-2\pi\im\alpha}\left(1-2^{\frac{3}{2}}\frac{k}{t}\right)\Big\{L^{-1}(0)B^{\pm}L(0)\Big\}_{21}-\frac{\gamma}{2}\e^{-2\pi\im\alpha}\left(1-2^{\frac{3}{2}}\frac{k}{t}\right)\Big\{L^{-1}(0)L'(0)\Big\}_{21}.
\end{align*}
We now begin to compute the contributions from $L^{\pm 1}(0),L'(0)$ and $L^{\pm 1}(\frac{1}{2}),L'(\frac{1}{2})$. In order to achieve this, we use the integral equation corresponding to $\widehat{L}(z)$. Modulo exponentially small contributions, for $z\in\mathbb{C}\backslash\Sigma_R$,
\begin{align}\label{best:1}
	\widehat{L}(z)=&\,I+\frac{1}{2\pi\im}\oint_{\partial D(0,r)}\!\!\big(\widehat{G}_Q(w)-I\big)\frac{\d w}{w-z}+\frac{1}{2\pi\im}\oint_{\partial D(\frac{1}{2},r)}\!\!\big(\widehat{G}_Q(w)-I\big)\frac{\d w}{w-z}\\
	&\,+\frac{1}{2\pi\im}\oint_{\partial D(0,r)}\!\!\big(\widehat{L}_-(w)-I\big)\big(\widehat{G}_Q(w)-I\big)\frac{\d w}{w-z}+\frac{1}{2\pi\im}\oint_{\partial D(\frac{1}{2},r)}\!\!\big(\widehat{L}_-(w)-I\big)\big(\widehat{G}_Q(w)-I\big)\frac{\d w}{w-z}.\nonumber
\end{align}
With \eqref{DZ:better},
\begin{equation*}\label{best:2}
	\frac{1}{2\pi\im}\oint_{\partial D(0,r)}\!\!\big(\widehat{L}_-(w)-I\big)\big(\widehat{G}_Q(w)-I\big)\frac{\d w}{w-z}=\mathcal{O}\left(t^{-\frac{3}{2}-|\alpha|}\right),\ \ z\in\mathbb{C}\backslash\Sigma_R.
\end{equation*}
For an analogous estimation on $\partial D(\frac{1}{2},r)$ which we require in \eqref{best:1}, we use 
\begin{equation*}
	\frac{1}{2\pi\im}\oint_{\partial D(\frac{1}{2},r)}\big(\widehat{L}_-(w)-I\big)\big(\widehat{G}_Q(w)-I\big)\frac{\d w}{w-z}=\mathcal{O}\left(t^{-1-2|\alpha|}\right),\ \ z\in\mathbb{C}\backslash\Sigma_R;
\end{equation*}
and therefore back in \eqref{best:1}, for $z\in\mathbb{C}\backslash\Sigma_R$, since $1\leq1+2|\alpha|\leq\frac{3}{2}+|\alpha|\leq 2$,
\begin{equation}\label{best:3}
	\widehat{L}(z)=I+\frac{1}{2\pi\im}\oint_{\partial D(0,r)}\!\!\big(\widehat{G}_Q(w)-I\big)\frac{\d w}{w-z}+\frac{1}{2\pi\im}\oint_{\partial D(\frac{1}{2},r)}\!\!\big(\widehat{G}_Q(w)-I\big)\frac{\d w}{w-z}+\mathcal{O}\left(t^{-1-2|\alpha|}\right).
\end{equation}
Starting from \eqref{best:3} it is now straightforward to derive the following estimates,
\begin{align}
	|s|^{\frac{1}{4}\sigma_3}L^{-1}(0)|s|^{-\frac{1}{4}\sigma_3}=\,&I-\e^{-\im\pi\alpha\sigma_3}\Big\{A_0^{\pm}t^{-\frac{1}{2}-|\alpha|}+A_0t^{-1}+A_1^{\pm}t^{-1}+\mathcal{O}\left(t^{-\min\{1+2|\alpha|,\frac{3}{2}-|\alpha|\}}\right)\Big\}\e^{\im\pi\alpha\sigma_3},\label{co:1}\\
	|s|^{\frac{1}{4}\sigma_3}L'(0)|s|^{-\frac{1}{4}\sigma_3}=\,&\e^{-\im\pi\alpha\sigma_3}\Big\{A_2^{\pm}t^{-\frac{1}{2}-|\alpha|}+A_2t^{-1}+A_3^{\pm}t^{-1}+\mathcal{O}\left(t^{-\min\{1+2|\alpha|,\frac{3}{2}-|\alpha|\}}\right)\Big\}\e^{\im\pi\alpha\sigma_3},\label{co:2}
\end{align}
with
\begin{equation*}
	A_0=\frac{1}{8}\left(1-2^{\frac{3}{2}}\frac{k}{t}\right)^{-1}\begin{pmatrix} 0 & \frac{2}{3}\frac{1+2^{\frac{3}{2}}\frac{k}{t}}{1-2^{\frac{3}{2}}\frac{k}{t}}\\ -3 & 0 \end{pmatrix};\ \ \ \ \ \ A_1^{\pm}=2^{-\frac{3}{2}}\frac{k}{3}(k\pm 1)2^{\frac{1}{4}\sigma_3}\begin{pmatrix} 0 & 7\\ 1 & 0 \end{pmatrix}2^{-\frac{1}{4}\sigma_3};
\end{equation*}
and
\begin{equation*}
	A_2=-\frac{1}{4}\left(1-2^{\frac{3}{2}}\frac{k}{t}\right)^{-1}\begin{pmatrix} 0 & \star\\ \frac{1+2^{\frac{3}{2}}\frac{k}{t}}{1-2^{\frac{3}{2}}\frac{k}{t}} & 0 \end{pmatrix};\ \ \ \ \ \ A_3^{\pm}=2^{-\frac{1}{2}}\frac{k}{3}(k\pm 1)2^{\frac{1}{4}\sigma_3}\begin{pmatrix} 0 & 13\\ 7 & 0 \end{pmatrix}2^{-\frac{1}{4}\sigma_3}.
\end{equation*}
\begin{remark} Explicit expressions for the matrices $A_0^{\pm}$ and $A_2^{\pm}$, for instance
\begin{equation*}
	A_2^{\pm}=2^{\pm\frac{5k}{2}+\frac{3}{4}}2^{\frac{1}{4}\sigma_3}\begin{pmatrix} 1 & \mp 1\\ \pm 1 & -1 \end{pmatrix}2^{-\frac{1}{4}\sigma_3}\begin{cases}-\im\gamma_{k-1},&(+)\\ \im\gamma_k^{-1},&(-) \end{cases},
\end{equation*}
are not going to be important later on: As we shall see shortly, all powers of $t$ which explicitly contain the parameter $\alpha$ will not contribute to leading orders after integration, see \eqref{int:clue} and \eqref{int:clue2} below.
\end{remark}
We now obtain
\begin{align}\label{best:4}
	\Big\{L^{-1}(0)&L'(0)\Big\}_{21}=\e^{2\pi\im\alpha}|s|^{\frac{1}{2}}\Big\{A_2^{\pm}t^{-\frac{1}{2}-|\alpha|}+A_2t^{-1}+A_3^{\pm}t^{-1}+\mathcal{O}\left(t^{-\min\{1+2|\alpha|,\frac{3}{2}-|\alpha|\}}\right)\Big\}_{21}\\
	&=\e^{2\pi\im\alpha}t^{\frac{1}{3}}\left[-\frac{1}{4t}+\frac{7k}{6t}(k\pm 1)\pm2^{\pm\frac{5k}{2}+\frac{1}{4}}t^{-\frac{1}{2}-|\alpha|}\begin{cases}-\im\gamma_{k-1},&(+)\\ \im\gamma_k^{-1},&(-) \end{cases}\,\,+\mathcal{O}\left(t^{-\min\{1+2|\alpha|,\frac{3}{2}-|\alpha|\}}\right)\right].\nonumber
\end{align}
Since furthermore
\begin{equation*}
	L^{\pm 1}\left(\frac{1}{2}\right)=\e^{-\im\pi\alpha\sigma_3}|s|^{-\frac{1}{4}\sigma_3}\left[I+\mathcal{O}\left(t^{-\frac{1}{2}-|\alpha|}\right)\right]|s|^{\frac{1}{4}\sigma_3}\e^{\im\pi\alpha\sigma_3},
\end{equation*}
we also deduce from \eqref{Bmat},
\begin{equation*}
	B^{\pm}=\sigma^{\pm}\e^{-\im\pi\alpha\sigma_3}|s|^{-\frac{1}{4}\sigma_3}\Big\{I+\mathcal{O}\left(t^{-\frac{1}{2}-|\alpha|}\right)\Big\}2^{\frac{1}{4}\sigma_3}\begin{pmatrix} 1 & \pm 1\\ \mp 1 & -1 \end{pmatrix}2^{-\frac{1}{4}\sigma_3}\Big\{I+\mathcal{O}\left(t^{-\frac{1}{2}-|\alpha|}\right)\Big\}|s|^{\frac{1}{4}\sigma_3}\e^{\im\pi\alpha\sigma_3},
\end{equation*}
and thus
\begin{equation}\label{best:5}
	\Big\{L^{-1}(0)B^{\pm}L(0)\Big\}_{21}=\e^{2\pi\im\alpha}|s|^{\frac{1}{2}}\sigma^{\pm}\Big\{2^{\frac{1}{4}\sigma_3}\begin{pmatrix} 1 & \pm 1\\ \mp 1 & -1 \end{pmatrix}2^{-\frac{1}{4}\sigma_3} +\mathcal{O}\left(t^{-\frac{1}{2}-|\alpha|}\right)\Big\}_{21}.
\end{equation}
Summarizing (recall that $\gamma=1-\e^{-\varkappa_{\textnormal{Ai}}t}=1+\mathcal{O}\left(t^{-\infty}\right)$),
\begin{eqnarray}
	F(J_{\textnormal{Ai}};\gamma)&=&\frac{1}{4}|s|^2\left(1-2^{\frac{3}{2}}\frac{k}{t}\right)^2\mp\sqrt{2}\,t^{\frac{1}{3}}\sigma^{\pm}-\frac{1}{2}t^{\frac{1}{3}}\left(-\frac{1}{4t}+\frac{7k}{6t}(k\pm 1)\right)+\mathcal{O}\left(t^{-\frac{1}{6}-|\alpha|}\right)\nonumber\\
	&=&\frac{|s|^2}{4}-\frac{1}{8s}\mp\sqrt{2}|s|^{\frac{1}{2}}\sigma^{\pm}-\sqrt{2}|s|^{\frac{1}{2}}k-\frac{2k^2}{s}+\frac{7k}{12s}(k\pm 1)+\mathcal{O}\left(t^{-\frac{1}{6}-|\alpha|}\right),\label{best:6}
\end{eqnarray}
which is the central estimation for the upcoming integration.
\subsection{Integration of expansion \eqref{best:6}} Let us first work out the details in the special case $k=0$ and $0\leq\alpha<\frac{1}{2}$. Note that in this situation
\begin{equation}\label{int:1}
	F(J_{\textnormal{Ai}};\gamma)=\frac{s^2}{4}-\frac{1}{8s}-\sqrt{2}|s|^{\frac{1}{2}}\sigma^++\mathcal{O}\left(t^{-\frac{1}{6}-\alpha}\right),\ \ \ \ \sigma^+=\frac{\pi^{-\frac{1}{2}}2^{-\frac{9}{4}}t^{\alpha-\frac{1}{2}}}{1+\pi^{-\frac{1}{2}}2^{-\frac{9}{4}}t^{\alpha-\frac{1}{2}}},
\end{equation}
and thus the following Lemma will be useful.
\begin{lem} Let $t=t(s)=(-s)^{\frac{3}{2}}$ and $\hat{s}_0<s<0$. For any $f\in L^1(\hat{s}_0,s)$, we have
\begin{equation*}
	\int_{\hat{s}_0}^sf\big(t(u)\big)\d u=-\frac{2}{3}\int_{\hat{t}_0}^tf(w)\frac{\d w}{w^{\frac{1}{3}}};\ \ \ \ \ \hat{t}_0=(-\hat{s}_0)^{\frac{3}{2}}>0.
\end{equation*}
\end{lem}
Recall that
\begin{equation}\label{int:2}
	\varkappa_{_\textnormal{Ai}}=\frac{v}{t}=\frac{2}{3}\sqrt{2}-\chi\frac{\ln t}{t};\ \ \ \ v=-\ln(1-\gamma)>0,\ \ \chi=k+\alpha;
\end{equation}
so that
\begin{equation}\label{int:clue}
	t^{\alpha}=t^{\chi-k}=t^{-k}\e^{-v+\frac{2}{3}\sqrt{2}\,t},
\end{equation}
and thus back in \eqref{int:1},
\begin{align*}
	-\sqrt{2}&\int_{\hat{s}_0}^s(-u)^{\frac{1}{2}}\sigma^+\big(t(u)\big)\d u=\frac{2}{3}\sqrt{2}\int_{\hat{t}_0}^t\frac{\pi^{-\frac{1}{2}}2^{-\frac{9}{4}}w^{-\frac{1}{2}}\e^{-v}\e^{\frac{2}{3}\sqrt{2}\,w}}{1+\pi^{-\frac{1}{2}}2^{-\frac{9}{4}}w^{-\frac{1}{2}}\e^{-v}\e^{\frac{2}{3}\sqrt{2}\,w}}\d w\\
	&\,=\ln\left(1+\pi^{-\frac{1}{2}}2^{-\frac{9}{4}}w^{-\frac{1}{2}}\e^{-v}\e^{\frac{2}{3}\sqrt{2}\,w}\right)\bigg|_{w=\hat{t}_0}^t+\frac{1}{2}\int_{\hat{t}_0}^t\frac{\pi^{-\frac{1}{2}}2^{-\frac{9}{4}}w^{-\frac{1}{2}}\e^{-v}\e^{\frac{2}{3}\sqrt{2}\,w}}{1+\pi^{-\frac{1}{2}}2^{-\frac{9}{4}}w^{-\frac{1}{2}}\e^{-v}\e^{\frac{2}{3}\sqrt{2}\,w}}\frac{\d w}{w}.
\end{align*}
Let us now choose $(v,t)$ sufficiently large positive such that
\begin{equation*}
	\frac{3}{2\sqrt{2}}\,v\leq t\leq\frac{3}{2\sqrt{2}}\left(v+\frac{1}{2}\ln t\right),
\end{equation*}
i.e. the base point of integration is $\hat{s}_0=-(\frac{3v}{2\sqrt{2}})^{\frac{2}{3}}$ or equivalently $\hat{t}_0=\frac{3v}{2\sqrt{2}}$. Then
\begin{equation*}
	-\sqrt{2}\int_{\hat{s}_0}^s(-u)^{\frac{1}{2}}\sigma^+\big(t(u)\big)\d u=\ln\left(1+\pi^{-\frac{1}{2}}2^{-\frac{9}{4}}t^{-\frac{1}{2}}\e^{\frac{2}{3}\sqrt{2}\,t-v}\right)+\mathcal{O}\left(t^{-\frac{1}{2}}\right).
\end{equation*}
Summarizing, we can integrate \eqref{int:1} as follows,
\begin{equation}\label{int:3}
	\ln\left(\frac{\det(I-\gamma K_{\textnormal{Ai}})\big|_{L^2(s,\infty)}}{\det(I-\gamma K_{\textnormal{Ai}})\big|_{L^2(\hat{s}_0,\infty)}}\right)=\frac{1}{12}(s^3-\hat{s}_0^3)-\frac{1}{8}\ln\left|\frac{s}{\hat{s}_0}\right|+\ln\left(1+\pi^{-\frac{1}{2}}2^{-\frac{9}{4}}t^{-\frac{1}{2}}\e^{\frac{2}{3}\sqrt{2}\,t-v}\right)+\mathcal{O}\left(t^{-\frac{1}{2}}\right),
\end{equation}
where we used that
\begin{equation}\label{int:clue2}
	\int_{\hat{s}_0}^s\mathcal{O}\left(t^{-\frac{1}{6}-\alpha}(u)\right)\d u=\int_{\hat{t}_0}^t\mathcal{O}\left(w^{-\frac{1}{2}-\alpha}\right)\d w=\int_{\hat{t}_0}^t\mathcal{O}\left(w^{-\frac{1}{2}}\e^{-\frac{2}{3}\sqrt{2}\,w}e^v\right)\d w=\mathcal{O}\left(t^{-\frac{1}{2}}\right),
\end{equation}
uniformly as $t\rightarrow+\infty,\gamma\uparrow 1$ such that $\frac{3}{2\sqrt{2}}\,v\leq t\leq\frac{3}{2\sqrt{2}}(v+\frac{1}{2}\ln t)$. However, with \eqref{A:1}, we have that
\begin{equation*}
	\ln\det(I-\gamma K_{\textnormal{Ai}})\big|_{L^2(\hat{s}_0,\infty)} = \frac{\hat{s}_0^3}{12}-\frac{1}{8}\ln|\hat{s}_0|+\ln c_0+\mathcal{O}\left(\hat{t}_0^{-\frac{1}{2}}\right),\ \ \ \ c_0=\exp\left[\frac{1}{24}\ln 2+\z'(-1)\right]
\end{equation*}
since for the latter determinant $\varkappa_{_\textnormal{Ai}}=\frac{v}{t(\hat{s}_0)} = \frac{2}{3}\sqrt{2}\geq \frac{2}{3}\sqrt{2}$. This substituted back into \eqref{int:3} we obtain
\begin{prop}\label{k0} There exist $t_0>0$ and $v_0>0$ such that
\begin{equation*}
	\ln\det(I-\gamma K_{\textnormal{Ai}})\Big|_{L^2(s,\infty)}=\frac{s^3}{12}-\frac{1}{8}\ln|s|+\ln c_0+\ln\left(1+\frac{1}{\sqrt{\pi}}\,2^{-\frac{9}{4}}t^{-\frac{1}{2}}\e^{\frac{2}{3}\sqrt{2}\,t-v}\right)+\mathcal{O}\left(t^{-\frac{1}{2}}\right)
\end{equation*}
uniformly for
\begin{equation}\label{condi}
	 t\geq t_0,v\geq v_0:\ \ \ \ \ \frac{3}{2\sqrt{2}}\,v\leq t\leq\frac{3}{2\sqrt{2}}\left(v+\frac{1}{2}\ln t\right).
\end{equation}
\end{prop}
\begin{remark} Note that the result of Proposition \ref{k0} matches the leading order expansion derived in \cite{B}, Corollary $1.17$. We only have to use that \eqref{condi} implies 
\begin{equation*}
	t\geq t_0,\ v\geq v_0:\ \ \ \ \frac{2}{3}\sqrt{2}\geq\frac{v}{t}=\varkappa_{_\textnormal{Ai}}\geq \frac{2}{3}\sqrt{2}-\chi\frac{\ln t}{t}\ \ \ \ \textnormal{with}\ \ \ \chi=\frac{1}{2}\left(1+\mathcal{O}\left(\frac{\ln t}{t}\right)\right).
\end{equation*}
\end{remark}
For the general case $k\in\mathbb{Z}_{\geq 0}$ we require the following notation: Let
\begin{equation*}
	\hat{t}_k=\frac{3}{2\sqrt{2}}\left(v+\left(k-\frac{1}{2}\right)\ln t\right),\ \ \ \hat{t}_k'=\frac{3}{2\sqrt{2}}\big(v+k\ln t\big),\ \ k\in\mathbb{Z}_{\geq 1};\ \ \ \ \ \ \ \hat{t}_0=\hat{t}_0'=v,
\end{equation*}
or equivalently $\hat{s}_k=-\hat{t}_k^{\frac{2}{3}}, \hat{s}_k'=-\hat{t}_k'^{\frac{2}{3}}$. In this case \eqref{int:3} generalizes to the following two expansions
\begin{lem}\label{lem:1} For $\hat{t}_k\leq t\leq \hat{t}_k'$, i.e. $\alpha\in[-\frac{1}{2},0]$,
\begin{align}
	\ln&\bigg(\frac{\det(I-\gamma K_{\textnormal{Ai}})\big|_{L^2(s,\infty)}}{\det(I-\gamma K_{\textnormal{Ai}})\big|_{L^2(\hat{s}_k,\infty)}}\bigg)=\frac{1}{12}\big(s^3-\hat{s}_k^3\big)-\frac{1}{8}\ln\left|\frac{s}{\hat{s}_k}\right|+\ln\left(1+\im\gamma_{k-1}2^{\frac{5k}{2}-\frac{5}{4}}t^{k-\frac{1}{2}}\e^{-\frac{2}{3}\sqrt{2}\,t+v}\right)\nonumber\\
	&-\ln\left(1+\im\gamma_{k-1}2^{\frac{5k}{2}-\frac{5}{4}}\right)+\frac{7k}{12}(k-1)\ln\left|\frac{s}{\hat{s}_k}\right|-2k^2\ln\left|\frac{s}{\hat{s}_k}\right|+\frac{2}{3}\sqrt{2}\,k(t-\hat{t}_k)+\mathcal{O}\left(t^{-\frac{1}{2}}\right).\label{int:4}
\end{align}
\end{lem}
\begin{proof}  We use that
\begin{align*}
	\sqrt{2}&\int_{\hat{s}_k}^s(-u)^{\frac{1}{2}}\sigma^-\big(t(u)\big)\d u=-\frac{2}{3}\sqrt{2}\int_{\hat{t}_k}^t\frac{\im\gamma_{k-1}2^{\frac{5k}{2}-\frac{5}{4}}w^{k-\frac{1}{2}}\e^v\e^{-\frac{2}{3}\sqrt{2}\,w}}{1+\im\gamma_{k-1}2^{\frac{5k}{2}-\frac{5}{4}}w^{k-\frac{1}{2}}\e^v\,\e^{-\frac{2}{3}\sqrt{2}\,w}}\,\d w\\
	&\,=\ln\left(1+\im\gamma_{k-1}2^{\frac{5k}{2}-\frac{5}{4}}t^{k-\frac{1}{2}}\e^{-\frac{2}{3}\sqrt{2}\,t+v}\right)-\ln\left(1+\im\gamma_{k-1}2^{\frac{5k}{2}-\frac{5}{4}}\right)+\mathcal{O}\left(t^{-\frac{1}{2}}\right)
\end{align*}
and the stated identity follows from \eqref{best:6}.
\end{proof}
\begin{lem}\label{lem:2} For $\hat{t}_k'\leq t\leq\hat{t}_{k+1}$, i.e. $\alpha\in[0,\frac{1}{2}]$,
\begin{align}
	\ln\bigg(\frac{\det(I-\gamma K_{\textnormal{Ai}})\big|_{L^2(s,\infty)}}{\det(I-\gamma K_{\textnormal{Ai}})\big|_{L^2(\hat{s}_k',\infty)}}\bigg)=&\,\frac{1}{12}\big(s^3-\hat{s}_k'^3\big)-\frac{1}{8}\ln\left|\frac{s}{\hat{s}_k'}\right|+\ln\left(1-\im\gamma_k^{-1}2^{-\frac{5k}{2}-\frac{5}{4}}t^{-k-\frac{1}{2}}\e^{\frac{2}{3}\sqrt{2}\,t-v}\right)\nonumber\\
	&+\frac{7k}{12}(k+1)\ln\left|\frac{s}{\hat{s}_k'}\right|-2k^2\ln\left|\frac{s}{\hat{s}_k'}\right|+\frac{2}{3}\sqrt{2}\,k(t-\hat{t}_k')+\mathcal{O}\left(t^{-\frac{1}{2}}\right).\label{int:5}
\end{align}
\end{lem}
\begin{proof} Use 
\begin{align*}
	-\sqrt{2}&\int_{\hat{s}_k'}^s(-u)^{\frac{1}{2}}\sigma^+\big(t(u)\big)\d u=\frac{2}{3}\sqrt{2}\int_{\hat{t}_k'}^t\frac{-\im\gamma_k^{-1}2^{-\frac{5k}{2}-\frac{5}{4}}w^{-k-\frac{1}{2}}\e^{-v}\e^{\frac{2}{3}\sqrt{2}\,w}}{1-\im\gamma_k^{-1}2^{-\frac{5k}{2}-\frac{5}{4}}w^{-k-\frac{1}{2}}\e^{-v}\e^{\frac{2}{3}\sqrt{2}\,w}}\d w\\
	&\,=\ln\left(1-\im\gamma_k^{-1}2^{-\frac{5k}{2}-\frac{5}{4}}t^{-k-\frac{1}{2}}\e^{\frac{2}{3}\sqrt{2}\,t-v}\right)+\mathcal{O}\left(t^{-\frac{1}{2}}\right),
\end{align*}
back in \eqref{best:6}.
\end{proof}
The idea at this point is to successively improve Proposition \ref{k0} with the help of Lemma \ref{lem:1} and \ref{lem:2}. For instance, start with $k=1$ in Lemma \ref{lem:1}, then with Proposition \ref{k0},
\begin{equation*}
	\ln\det(I-\gamma K_{\textnormal{Ai}})\Big|_{L^2(\hat{s}_1,\infty)} = \frac{\hat{s}_1^3}{12}-\frac{1}{8}\ln|\hat{s}_1|+\ln c_0+\ln\left(1+\pi^{-\frac{1}{2}}2^{-\frac{9}{4}}\right)+\mathcal{O}\left(\hat{t}_1^{-\frac{1}{2}}\right)
\end{equation*}
and back into \eqref{int:4},
\begin{align*}
	\ln&\det(I-\gamma K_{\textnormal{Ai}})\Big|_{L^2(s,\infty)} =\frac{s^3}{12}-\frac{1}{8}\ln|s|+\ln c_0+\ln\left(1+\pi^{\frac{1}{2}}2^{\frac{9}{4}}t^{\frac{1}{2}}\e^{-\frac{2}{3}\sqrt{2}\,t+v}\right)-\ln\left(\pi^{\frac{1}{2}}2^{\frac{9}{4}}\right)\\
	&\,+\frac{2}{3}\sqrt{2}\,(t-\hat{t}_1)+\mathcal{O}\left(t^{-\frac{1}{2}}\right)=\frac{s^3}{12}-\frac{1}{8}\ln|s|+\ln c_0+\ln\left(1+\pi^{-\frac{1}{2}}2^{-\frac{9}{4}}t^{-\frac{1}{2}}\e^{\frac{2}{3}\sqrt{2}\,t-v}\right)+\mathcal{O}\left(t^{-\frac{1}{2}}\right)
\end{align*}
where we used that
\begin{equation*}
	\frac{2}{3}\sqrt{2}\,(t-\hat{t}_1)=\ln\left(t^{-\frac{1}{2}}\e^{\frac{2}{3}\sqrt{2}\,t-v}\right).
\end{equation*}
Summarizing,
\begin{prop}\label{k1} There exists $t_0>0$ and $v_0>0$ such that
\begin{equation*}
	\ln\det(I-\gamma K_{\textnormal{Ai}})\Big|_{L^2(s,\infty)}=\frac{s^3}{12}-\frac{1}{8}\ln|s|+\ln c_0+\ln\left(1+\frac{1}{\sqrt{\pi}}\,2^{-\frac{9}{4}}t^{-\frac{1}{2}}\e^{\frac{2}{3}\sqrt{2}\,t-v}\right)+\mathcal{O}\left(t^{-\frac{1}{2}}\right)
\end{equation*}
uniformly for
\begin{equation*}
	 t\geq t_0,v\geq v_0:\ \ \ \ \ \frac{3}{2\sqrt{2}}\left(v+\frac{1}{2}\ln t\right)\leq t\leq\frac{3}{2\sqrt{2}}(v+\ln t).
\end{equation*}
\end{prop}
Next, with $k=1$ in Lemma \ref{lem:2} and Proposition \ref{k1},
\begin{align*}
	\ln\det(I-\gamma K_{\textnormal{Ai}})\Big|_{L^2(\hat{s}_1',\infty)}=&\,\frac{\hat{s}_1'^3}{12}-\frac{1}{8}\ln|\hat{s}_1'|+\ln c_0+\ln\left(1+\pi^{-\frac{1}{2}}2^{-\frac{9}{4}}\hat{t}_1'^{-\frac{1}{2}}\e^{\frac{2}{3}\sqrt{2}\,\hat{t}_1'-v}\right)+\mathcal{O}\left(\hat{t}_1'^{-\frac{1}{2}}\right)\\
	=&\,\frac{\hat{s}_1'^3}{12}-\frac{1}{8}\ln|\hat{s}_1'|+\ln c_0+\ln\left(1+\pi^{-\frac{1}{2}}2^{-\frac{9}{4}}t^{\frac{1}{2}}\right)+\mathcal{O}\left(t^{-\frac{1}{2}}\right).
\end{align*}
We substitute this back into \eqref{int:5},
\begin{align*}
	\ln\det(I-\gamma K_{\textnormal{Ai}})\Big|_{L^2(s,\infty)}=&\,\frac{s^3}{12}-\frac{1}{8}\ln|s|+\ln c_0+\ln\left(1+\pi^{-\frac{1}{2}}2^{-\frac{7}{2}-\frac{9}{4}}t^{-1-\frac{1}{2}}\e^{\frac{2}{3}\sqrt{2}\,t-v}\right)\\
	&\,+\ln\left(1+\pi^{-\frac{1}{2}}2^{-\frac{9}{4}}t^{\frac{1}{2}}\right)+\frac{2}{3}\sqrt{2}\,(t-\hat{t}_1')+\mathcal{O}\left(t^{-\frac{1}{2}}\right).
\end{align*}
and note that
\begin{equation*}
	\ln\left(1+\pi^{-\frac{1}{2}}2^{-\frac{9}{4}}t^{\frac{1}{2}}\right)+\frac{2}{3}\sqrt{2}\,(t-\hat{t}_1')=\ln\left(1+\pi^{-\frac{1}{2}}2^{-\frac{9}{4}}t^{-\frac{1}{2}}\e^{\frac{2}{3}\sqrt{2}\,t-v}\right)+\mathcal{O}\left(t^{-\frac{1}{2}}\right),\ \ \ \hat{t}_1'\leq t\leq \hat{t}_2.
\end{equation*}
Hence,
\begin{prop}\label{k2} There exists $t_0>0$ and $v_0>0$ such that
\begin{equation*}
	\ln\det(I-\gamma K_{\textnormal{Ai}})\Big|_{L^2(s,\infty)}=\frac{s^3}{12}-\frac{1}{8}\ln|s|+\ln c_0+\ln\prod_{j=0}^1\left(1+\frac{j!}{\sqrt{\pi}}\,2^{-\frac{7}{2}j-\frac{9}{4}}t^{-j-\frac{1}{2}}\e^{\frac{2}{3}\sqrt{2}\,t-v}\right)+\mathcal{O}\left(t^{-\frac{1}{2}}\right)
\end{equation*}
uniformly for
\begin{equation*}
	t\geq t_0,\ v\geq v_0:\ \ \ \ \frac{3}{2\sqrt{2}}(v+\ln t)\leq t\leq\frac{3}{2\sqrt{2}}\left(v+\frac{3}{2}\ln t\right).
\end{equation*}
\end{prop}
\begin{remark}\label{arti} Note that the lower constraint on $t$ is artificial: if we were to fix $t<\frac{3}{2\sqrt{2}}(v+\ln t)$, the second factor in the product moves to the error term and we reproduce the result of Proposition \ref{k1}, after adjusting the error term.
\end{remark}
So far repeated integration has lead us to a sequence of results,
\begin{equation*}
	\eqref{A:1}\ \ \mapsto\ \ t\in\underbrace{[\hat{t}_0,\hat{t}_1]}_{\textnormal{Prop.}\,\ref{k0}}\ \ \mapsto\ \ t\in\underbrace{[\hat{t}_1,\hat{t}_1']}_{\textnormal{Prop.}\,\ref{k1}}\ \ \mapsto\ \ t\in\underbrace{[\hat{t}_1',\hat{t}_2]}_{\textnormal{Prop.}\,\ref{k2}},
\end{equation*}
but this strategy can be continued indefinitely
leading to the following result (here we also use Remark \ref{arti}).
\begin{theo}\label{nice:2} Given $q\in\mathbb{Z}_{\geq 1}$, there exists $t_0=t_0(q)>0$ and $v_0=v_0(q)>0$ such that
\begin{equation*}
\ln\det(I-\gamma K_{\textnormal{Ai}})\Big|_{L^2(s,\infty)}=\frac{s^3}{12}-\frac{1}{8}\ln|s|+\ln c_0+\ln\prod_{j=0}^{q-1}\left(1+\frac{j!}{\sqrt{\pi}}\,2^{-\frac{7}{2}j-\frac{9}{4}}t^{-j-\frac{1}{2}}\e^{\frac{2}{3}\sqrt{2}\,t-v}\right)+\mathcal{O}\left(t^{-\frac{1}{2}}\right)
\end{equation*}
uniformly for
\begin{equation*}
	t\geq t_0,\ v\geq v_0:\ \ \ \ \frac{3}{2\sqrt{2}}\,v\leq t\leq\frac{3}{2\sqrt{2}}\left(v+q\ln t\right).
\end{equation*}
\end{theo}

\section{Proof of Theorem \ref{Airymain} and asymptotics for eigenvalues}\label{Airyp4}
The content of Theorem \ref{Airymain} is simply a combination of Theorems \ref{nice:2} and \ref{bettererror}:
\begin{cor}\label{nicecor} Given $\chi\in\mathbb{R}$ determine $p\in\mathbb{Z}_{\geq 0}$ such that $p=0$ for $\chi<-\frac{1}{2}$ and $\chi+\frac{1}{2}<p\leq\chi+\frac{3}{2}$ for $\chi\geq -\frac{1}{2}$. There exist positive $t_0=t_0(\chi),v_0=v_0(\chi)$ such that
\begin{equation*}
	D(J_{\textnormal{Ai}};\gamma)=\exp\left[\frac{s^3}{12}\right]|s|^{-\frac{1}{8}}c_0\prod_{j=0}^{p-1}\left(1+\frac{j!}{\sqrt{\pi}}\,2^{-\frac{7}{2}j-\frac{9}{4}}t^{-j-\frac{1}{2}}\e^{\frac{2}{3}\sqrt{2}\,t-v}\right)\left(1+\mathcal{O}\left(t^{-\min\{p-\frac{1}{2}-\chi,\frac{1}{2}\}}\right)\right)
\end{equation*}
uniformly for $t\geq t_0,v\geq v_0$ and $\varkappa_{_\textnormal{Ai}}=\frac{v}{t}\geq\frac{2}{3}\sqrt{2}-\chi\frac{\ln t}{t}$. In case $p=0$, we take $\prod_{j=0}^{p-1}(\ldots)\equiv 1$.
\end{cor}
We now begin to derive asymptotic information for individual eigenvalues and employ techniques which have occurred previously in \cite{B}. First, by positivity of $\e^{-v}\lambda_j(1-\lambda_j)^{-1}$,
\begin{equation*}
	\forall\,p\in\mathbb{Z}_{\geq 0}:\ \ \ 1+\e^{-v}\frac{\lambda_p(s)}{1-\lambda_p(s)}\geq 1=\frac{\det(I-\gamma K_{\textnormal{Ai}})}{\det(I-\gamma K_{\textnormal{Ai}})}\bigg|_{L^2(s,\infty)}.\!\!\!
\end{equation*}
Now, we use Corollary \ref{nicecor} for the determinant in the denominator (with $\chi'=\chi,\chi\geq 0$ and thus $p'=p\in\mathbb{Z}_{\geq 0}$) as well as in the numerator (with $\chi''=\chi+1,\chi\geq 0$ and thus $p''=p'+1=p+1$), i.e.
\begin{equation*}
	1+\e^{-v}\frac{\lambda_p(s)}{1-\lambda_p(s)}\geq\left(1+\frac{p!}{\sqrt{\pi}}2^{-\frac{7}{2}p-\frac{9}{4}}t^{-p-\frac{1}{2}}\e^{\frac{2}{3}\sqrt{2}\,t-v}\right)\left(1+\mathcal{O}\left(t^{-\min\{p-\frac{1}{2}-\chi,\frac{1}{2}\}}\right)\right)
\end{equation*}
uniformly for $t\geq t_0,v\geq v_0$ such that $\varkappa_{_\textnormal{Ai}}\geq\frac{2}{3}\sqrt{2}-\chi\frac{\ln t}{t},\chi\geq 0$. Hence, after algebra,
\begin{equation}\label{first:es}
	\forall\,p\in\mathbb{Z}_{\geq 0}:\ \ \ \frac{\lambda_p(s)}{1-\lambda_p(s)}\geq C_pt^{-p-\frac{1}{2}}\e^{\frac{2}{3}\sqrt{2}\,t},\ \ t\rightarrow+\infty;
\end{equation}
and the constant $C_p>0$ can be chosen independent of $v$. Next, we make use of Lidskii's Theorem: Valid for any $\ell\in\mathbb{Z}_{\geq 0}$,
\begin{eqnarray*}
	\frac{\det(I-\gamma K_{\textnormal{Ai}})}{\det(I-K_{\textnormal{Ai}})}\bigg|_{L^2(s,\infty)} &=& \det\big(I+\e^{-v}K_{\textnormal{Ai}}(I-K_{\textnormal{Ai}})^{-1}\big)\\
	&=&\prod_{j=0}^{\ell-1}\left(1+\e^{-v}\frac{\lambda_{\ell}(s)}{1-\lambda_{\ell}(s)}\right)\det\big(I+\e^{-v}K_{\ell}(I-K_{\ell})^{-1}\big),\nonumber
\end{eqnarray*}
with $K_{\ell}=K_{\textnormal{Ai}}\cdot P_{\ell}$ and $P_{\ell}$ projects on the space of eigenvectors of $K_{\textnormal{Ai}}$ with corresponding eigenvalues $\{\lambda_j:\ j\geq \ell\}$. In this exact identity we can use the expansion of Corollary \ref{nicecor}, i.e. as $t\rightarrow+\infty,\gamma\uparrow 1$ such that $\varkappa_{_\textnormal{Ai}}\geq\frac{2}{3}\sqrt{2}-\chi\frac{\ln t}{t}$,
\begin{align}
	\prod_{j=0}^{p-1}\left(1+\frac{j!}{\sqrt{\pi}}\,2^{-\frac{7}{2}j-\frac{9}{4}}t^{-j-\frac{1}{2}}\e^{\frac{2}{3}\sqrt{2}\,t-v}\right)&\left(1+\mathcal{O}\left(t^{-\min\{p-\chi-\frac{1}{2},\frac{1}{2}\}}\right)\right)\label{eig:0}\\
	&\ \ \ \ =\prod_{j=0}^{\ell-1}\left(1+\e^{-v}\frac{\lambda_j(s)}{1-\lambda_j(s)}\right)\det\big(I+\e^{-v}K_{\ell}(I-K_{\ell})^{-1}\big).\nonumber
\end{align}
In here we first choose $\chi=\frac{1}{2}$, i.e. we take $p=2$ and, say, $\ell=1$:
\begin{equation}
	\left(1+\frac{1}{\sqrt{\pi}}2^{-\frac{9}{4}}t^{-\frac{1}{2}}\e^{\frac{2}{3}\sqrt{2}\,t-v}\right)\left(1+\mathcal{O}\left(t^{-\frac{1}{2}}\right)\right)
	=\left(1+\e^{-v}\frac{\lambda_0(s)}{1-\lambda_0(s)}\right)\det\big(I+\e^{-v}K_1(I-K_1)^{-1}\big).\label{eig:1}
\end{equation}
For $\frac{2}{3}\sqrt{2}\,t\geq v\geq\frac{2}{3}\sqrt{2}\,t-\frac{1}{2}\ln t$ we now see that $\det(I+\e^{-v}K_1(I-K_1)^{-1})\geq 1$ together with \eqref{minimax} implies in \eqref{eig:1} that
%
\begin{equation}\label{eig:3}
	\frac{t^{\frac{1}{2}}\e^{-\frac{2}{3}\sqrt{2}\,t}}{1-\lambda_0(s)}=\mathcal{O}(1),\ \ \ t\geq t_0,\ v\geq v_0:\ \ \frac{2}{3}\sqrt{2}\geq\varkappa_{_\textnormal{Ai}}\geq\frac{2}{3}\sqrt{2}-\frac{1}{2}\frac{\ln t}{t}.
\end{equation}
Next we choose $\chi=\frac{3}{2}$ in \eqref{eig:0}, i.e. we take $p=3$ and, say, $\ell=2$:
\begin{align}
	\left(1+\frac{1}{\sqrt{\pi}}2^{-\frac{9}{4}}t^{-\frac{1}{2}}\e^{\frac{2}{3}\sqrt{2}\,t-v}\right)&\left(1+\frac{1}{\sqrt{\pi}}2^{-\frac{7}{2}-\frac{9}{4}}t^{-\frac{3}{2}}\e^{\frac{2}{3}\sqrt{2}\,t-v}\right)\left(1+\mathcal{O}\left(t^{-\frac{1}{2}}\right)\right)\nonumber\\
	&\,=\left(1+\e^{-v}\frac{\lambda_0(s)}{1-\lambda_0(s)}\right)\left(1+\e^{-v}\frac{\lambda_1(s)}{1-\lambda_1(s)}\right)\det\big(I+\e^{-v}K_2(I-K_2)^{-1}\big)\label{eig:5}
\end{align}
Multiplying through with the second summand of the first factor in the left hand side of \eqref{eig:5}, we obtain
\begin{align*}
	&\left(1+\sqrt{\pi}\,2^{\frac{9}{4}}t^{\frac{1}{2}}\e^{-\frac{2}{3}\sqrt{2}\,t+v}\right)\left(1+\frac{1}{\sqrt{\pi}}2^{-\frac{7}{2}-\frac{9}{4}}t^{-\frac{3}{2}}\e^{\frac{2}{3}\sqrt{2}\,t-v}\right)\left(1+\mathcal{O}\left(t^{-\frac{1}{2}}\right)\right)\\
	&\,=\left(\sqrt{\pi}\,2^{\frac{9}{4}}t^{\frac{1}{2}}\e^{-\frac{2}{3}\sqrt{2}\,t+v}+\frac{\sqrt{\pi}\,2^{\frac{9}{4}}t^{\frac{1}{2}}\e^{-\frac{2}{3}\sqrt{2}\,t}\lambda_0(s)}{1-\lambda_0(s)}\right)\left(1+\e^{-v}\frac{\lambda_1(s)}{1-\lambda_1(s)}\right)\det\big(I+\e^{-v}K_2(I-K_2)^{-1}\big).
\end{align*}
For $\frac{2}{3}\sqrt{2}\,t-\frac{1}{2}\ln t\geq v\geq\frac{2}{3}\sqrt{2}\,t-\frac{3}{2}\ln t$, all factors in the left hand side are bounded, hence by positivity all factors in the right hand side have to be bounded and with $\det(I+\e^{-v}K_2(I-K_2)^{-1})\geq 1$, this leads us to
\begin{equation}\label{eig:8}
	\frac{t^{\frac{1}{2}}\e^{-\frac{2}{3}\sqrt{2}\,t}}{1-\lambda_0(s)}=\mathcal{O}(1),\ \ \ \ \frac{t^{\frac{1}{2}}\e^{-\frac{2}{3}\sqrt{2}\,t}}{1-\lambda_1(s)}=\mathcal{O}(1);\ \ \ \ \ t\geq t_0,\ v\geq v_0:\ \ \frac{2}{3}\sqrt{2}-\frac{1}{2}\frac{\ln t}{t}\geq\varkappa_{_\textnormal{Ai}}\geq\frac{2}{3}\sqrt{2}-\frac{3}{2}\frac{\ln t}{t}.
\end{equation}
After that, we let $\chi=\frac{5}{2}$, i.e. $p=4$ and, say, $\ell=3$: After simplification in \eqref{eig:0},
\begin{align*}
	\Big(1+&\sqrt{\pi}\,2^{\frac{9}{4}}t^{\frac{1}{2}}\e^{-\frac{2}{3}\sqrt{2}\,t+v}\Big)\left(1+\sqrt{\pi}\,2^{\frac{7}{2}+\frac{9}{4}}t^{\frac{3}{2}}\e^{-\frac{2}{3}\sqrt{2}\,t+v}\right)\left(1+\frac{2!}{\sqrt{\pi}}\,2^{-\frac{7}{2}2-\frac{9}{4}}t^{-2-\frac{1}{2}}\e^{\frac{2}{3}\sqrt{2}\,t-v}\right)\left(1+\mathcal{O}\left(t^{-\frac{1}{2}}\right)\right)\\
	&=\left(\sqrt{\pi}\,2^{\frac{9}{4}}t^{\frac{1}{2}}\e^{-\frac{2}{3}\sqrt{2}t+v}+\frac{\sqrt{\pi}\,2^{\frac{9}{4}}t^{\frac{1}{2}}\e^{-\frac{2}{3}\sqrt{2}\,t}\lambda_0(s)}{1-\lambda_0(s)}\right)\left(\sqrt{\pi}\,2^{\frac{7}{2}+\frac{9}{4}}t^{\frac{3}{2}}\e^{-\frac{2}{3}\sqrt{2}\,t+v}+\frac{\sqrt{\pi}\,2^{\frac{7}{2}+\frac{9}{4}}t^{\frac{3}{2}}\lambda_1(s)}{1-\lambda_1(s)}\right)\\
	&\,\times\left(1+\e^{-v}\frac{\lambda_2(s)}{1-\lambda_2(s)}\right)\det(I+\e^{-v}K_3(I-K_3)^{-1})
\end{align*}
and by boundedness of the left hand side, also with the help of \eqref{first:es}, for $t\geq t_0,v\geq v_0$,
\begin{equation}\label{eig:6}
	\frac{t^{\frac{1}{2}}\e^{-\frac{2}{3}\sqrt{2}\,t}}{1-\lambda_0(s)}=\mathcal{O}(1),\ \ \frac{t^{\frac{3}{2}}\e^{-\frac{2}{3}\sqrt{2}\,t}}{1-\lambda_1(s)}=\mathcal{O}(1),\ \ \frac{t^{\frac{3}{2}}\e^{-\frac{2}{3}\sqrt{2}\,t}}{1-\lambda_2(s)}=\mathcal{O}(1);\ \ \frac{2}{3}\sqrt{2}-\frac{3}{2}\frac{\ln t}{t}\geq\varkappa_{_\textnormal{Ai}}\geq\frac{2}{3}\sqrt{2}-\frac{5}{2}\frac{\ln t}{t}.
\end{equation}
Iterating this approach for general $\chi=p-\frac{1}{2},\ell=p-1,p\in\mathbb{Z}_{\geq 1}$ we derive a sequence of estimates,
\begin{equation*}
	\frac{t^{j+\frac{1}{2}}\e^{-\frac{2}{3}\sqrt{2}\,t}}{1-\lambda_j(s)}=\mathcal{O}(1),\ \ 0\leq j\leq p-2;\ \ \ \ \frac{t^{p-\frac{3}{2}}\e^{-\frac{2}{3}\sqrt{2}\,t}}{1-\lambda_{p-1}(s)}=\mathcal{O}(1);\ \ \ \frac{2}{3}\sqrt{2}-\left(p-\frac{3}{2}\right)\frac{\ln t}{t}\geq\varkappa_{_\textnormal{Ai}}\geq\frac{2}{3}\sqrt{2}-\left(p-\frac{1}{2}\right)\frac{\ln t}{t}.
\end{equation*}
Since these estimates are valid for any $p\in\mathbb{Z}_{\geq 1}$, we have in fact
\begin{equation}\label{second:es}
	\forall\,j\in\mathbb{Z}_{\geq 0}:\ \ \ \frac{\lambda_j(s)}{1-\lambda_j(s)}\leq D_jt^{-j-\frac{1}{2}}\e^{\frac{2}{3}\sqrt{2}\,t},\ \ t\rightarrow+\infty
\end{equation}
and the constant $D_j>0$ can be chosen independent of $v$. With \eqref{first:es} and \eqref{second:es} thus
\begin{prop}\label{almost} For any $j\in\mathbb{Z}_{\geq 0}$ fixed, as $t\rightarrow+\infty$,
\begin{equation*}
	1-\lambda_j(s)=c_jt^{j+\frac{1}{2}}\e^{-\frac{2}{3}\sqrt{2}\,t}\big(1+o(1)\big),\ \ \ c_j>0.
\end{equation*}
\end{prop}

We now recall a few general facts about the trace class operator $K_{\textnormal{Ai}}:L^2\big((s,\infty);\d\lambda\big)\circlearrowleft$. From the identity
\begin{equation}\label{convu}
	K_{\textnormal{Ai}}(\lambda,\mu)=\frac{\textnormal{Ai}(\lambda)\textnormal{Ai}'(\mu)-\textnormal{Ai}(\mu)\textnormal{Ai}'(\lambda)}{\lambda-\mu}=\int_0^{\infty}\textnormal{Ai}(\lambda+t)\textnormal{Ai}(t+\mu)\d t
\end{equation}
it follows that $K_{\textnormal{Ai}}$ is positive definite with finite operator norm $\|K_{\textnormal{Ai}}\|<1$ and trace norm
\begin{equation}\label{eig:10}
	\|K_{\textnormal{Ai}}\|_1=\int_s^{\infty}\textnormal{K}_{\textnormal{Ai}}(\lambda,\lambda)\d\lambda\leq ct,\ \ \ c>0,\ \ t=|s|^{\frac{3}{2}}.
\end{equation}
Note that for any positive definite, trace class operator $B$,
\begin{equation*}
	\big|\det(I+B)-1\big|\leq\|B\|_1\exp\left(\|B\|_1\right),
\end{equation*}
and thus, for any $p\in\mathbb{Z}_{\geq 1}$
\begin{equation*}
	\big|\det(I+\e^{-v}K_p(I-K_p)^{-1})-1\big|\leq \e^{-v}\|K_p(I-K_p)^{-1}\|_1\exp\left(\e^{-v}\|K_p(I-K_p)^{-1}\|_1\right).
\end{equation*}
But with \eqref{second:es},
\begin{equation*}
	\e^{-v}\|K_p(I-K_p)^{-1}\|_1\leq \e^{-v}\|K_p\|_1\|(I-K_p)^{-1}\|\leq \e^{-v}\frac{\|K_p\|_1}{1-\lambda_p}\leq C_p\,t^{\frac{1}{2}-p}\e^{\frac{2}{3}\sqrt{2}\,t-v}.
\end{equation*}
With this back to \eqref{eig:0} where $\varkappa_{_\textnormal{Ai}}\geq\frac{2}{3}\sqrt{2}-\chi\frac{\ln t}{t}$, and we now let $\ell=p$,
\begin{equation*}
	\det\big(I+\e^{-v}K_p(I-K_p)^{-1}\big)=1+\mathcal{O}\left(t^{-(p-\chi-\frac{1}{2})}\right)=1+o(1).
\end{equation*}
Summarizing,
\begin{prop}\label{eig:imp2} Given $\chi\in\mathbb{R}$, determine $p\in\mathbb{Z}_{\geq 0}$ as in Corollary \ref{nicecor}. There exist positive $t_0=t_0(\chi),v_0=v_0(\chi)$ such that
\begin{equation*}
	\prod_{j=0}^{p-1}\left(1+\frac{j!}{\sqrt{\pi}}2^{-\frac{7}{2}j-\frac{9}{4}}t^{-j-\frac{1}{2}}\e^{\frac{2}{3}\sqrt{2}\,t-v}\right)=\prod_{j=0}^{p-1}\left(1+\e^{-v}\frac{\lambda_j(s)}{1-\lambda_j(s)}\right)\left(1+\mathcal{O}\left(t^{-\min\{p-\chi-\frac{1}{2},\frac{1}{2}\}}\right)\right)
\end{equation*}
uniformly for $t\geq t_0,v\geq v_0$ and $\varkappa_{_\textnormal{Ai}}\geq\frac{2}{3}\sqrt{2}-\chi\frac{\ln t}{t}$. Also here we take $\prod_{j=0}^{p-1}(\ldots)\equiv 1$, if $p=0$.
\end{prop}
At this point we use Proposition \ref{almost}, take $\ln(\ldots)$ of both sides and compare powers in $t$, this gives
\begin{equation*}
	c_j=\frac{j!}{\sqrt{\pi}}2^{-\frac{7}{2}j-\frac{9}{4}},\ \ j\in\mathbb{Z}_{\geq 0},
\end{equation*}
and Corollary \ref{Aispecex} follows.
\section{Nonlinear steepest descent analysis associated with $K_{\textnormal{Bess}}^{(a)}$ -- part 1}\label{Bessp1}
The Bessel kernel \eqref{Bessk} is of type \eqref{IIKSker} with
\begin{equation}\label{BIIKS}
	\phi(z)=J_a(\sqrt{z}),\ \ \ \psi(z)=\frac{1}{2}\sqrt{z}J_a'(\sqrt{z});\ \ \ \ \ J_{\textnormal{Bess}}=(0,s),\ \ s>0
\end{equation}
and we will derive an asymptotic solution of RHP \ref{masterIIKS} for sufficiently large (positive) $s$ and $\gamma$ close to $1$ such that
\begin{equation}\label{Bscale}
	\varkappa_{_\textnormal{Bess}}\equiv\frac{v}{t}=2-2\left(\chi+\frac{a}{2}\right)\frac{\ln t}{t},\ \ \ \ v=-\ln(1-\gamma)>0,\ \ \ t=s^{\frac{1}{2}};\ \ \ \ \chi\in\mathbb{R}_{\geq 0},\ \ a\in\mathbb{R}_{>-1}.
\end{equation}
\subsection{Preliminary transformations}\label{Besspre:1} We introduce the unimodular function
\begin{equation*}
	\Psi_a(\z)=\sqrt{\pi}\,\e^{-\im\frac{\pi}{4}}\begin{pmatrix} I_a\big((\e^{-\im\pi}\z)^{\frac{1}{2}}\big)& -\frac{\im}{\pi}K_a\big((\e^{-\im\pi}\z)^{\frac{1}{2}}\big)\smallskip\\ (\e^{-\im\pi}\z)^{\frac{1}{2}}I_a'\big((\e^{-\im\pi}\z)^{\frac{1}{2}}\big) &
	-\frac{\im}{\pi}(\e^{-\im\pi}\z)^{\frac{1}{2}}K_a'\big((\e^{-\im\pi}\z)^{\frac{1}{2}}\big) \end{pmatrix}\e^{\im\frac{\pi}{2}a\sigma_3},\ \ \ \z\in\mathbb{C}\backslash[0,\infty)
\end{equation*}
in terms of the modified Bessel functions $I_{\nu}(z),K_{\nu}(z)$, cf. \cite{NIST}. As before, we choose principal branches for $\z^{\frac{1}{2}}:\,\textnormal{arg}\,\z\in(-\pi,\pi]$. Next we assemble
\begin{equation}\label{bess:1}
	\Psi(\z;a)=\Psi_a(\z)\e^{-\im\frac{\pi}{2}a\sigma_3}\begin{cases}\bigl(\begin{smallmatrix} 1 & 0 \\ -\e^{-\im\pi a}& 1 \end{smallmatrix}\bigr),&\textnormal{arg}\,\z\in(0,\frac{\pi}{3})\\ I,&\textnormal{arg}\,\z\in(\frac{\pi}{3},\frac{5\pi}{3})\\ \bigl(\begin{smallmatrix} 1 & 0 \\ \e^{\im\pi a}& 1 \end{smallmatrix}\bigr),&\textnormal{arg}\,\z\in(\frac{5\pi}{3},2\pi)\end{cases}
\end{equation}
which leads us to the problem below.
\begin{problem} The parametrix $\Psi(\z;a)$ defined in \eqref{bess:1} has the following properties:
\begin{enumerate}
	\item $\Psi(\z;a)$ is analytic for $\z\in\mathbb{C}\backslash\big(\bigcup_{j=1}^3\widehat{\Gamma}_j\cup\{0\}\big)$ with
	\begin{equation*}
		\widehat{\Gamma}_1=\e^{\im\frac{\pi}{3}}(0,\infty),\ \ \ \ \widehat{\Gamma}_2=(0,\infty),\ \ \ \ \widehat{\Gamma}_3=\e^{\im\frac{5\pi}{3}}(0,\infty)
	\end{equation*}
	and all three rays are oriented as shown in Figure \ref{figure6B}.
	\item We observe the following jumps,
	\begin{eqnarray*}
		\Psi_+(\z;a)&=&\Psi_-(\z;a)\bigl(\begin{smallmatrix} 1 & 0\smallskip\\ \e^{-\im\pi a} & 1 \end{smallmatrix}\bigr),\ \ \z\in\widehat{\Gamma}_1;\ \ \ \ 
		\Psi_+(\z;a)=\Psi_-(\z;a)\bigl(\begin{smallmatrix} 1 & 0\smallskip\\ \e^{\im\pi a} & 1 \end{smallmatrix}\bigr),\ \ \z\in\widehat{\Gamma}_3;\\
		\Psi_+(\z;a)&=&\Psi_-(\z;a)\bigl(\begin{smallmatrix} 0 & 1\smallskip\\ -1 & 0 \end{smallmatrix}\bigr),\ \ \z\in\widehat{\Gamma}_2.
	\end{eqnarray*}
	\item Near $\z=0$, in case $a\notin\mathbb{Z}$,
	\begin{equation*}
		\Psi(\z;a)=\widehat{\Psi}(\z;a)(\e^{-\im\pi}\z)^{\frac{a}{2}\sigma_3}\begin{pmatrix} 1 & \frac{\im}{2}\frac{1}{\sin\pi a}\\ 0 & 1 \end{pmatrix}\begin{cases}\bigl(\begin{smallmatrix} 1 & 0 \\ -\e^{-\im\pi a}& 1 \end{smallmatrix}\bigr),&\textnormal{arg}\,\z\in(0,\frac{\pi}{3})\\ I,&\textnormal{arg}\,\z\in(\frac{\pi}{3},\frac{5\pi}{3})\\ \bigl(\begin{smallmatrix} 1 & 0 \\ \e^{\im\pi a}& 1 \end{smallmatrix}\bigr),&\textnormal{arg}\,\z\in(\frac{5\pi}{3},2\pi)\end{cases};
	\end{equation*}
	and for $a\in\mathbb{Z}$,
	\begin{equation*}
		\Psi(\z;a)=\widehat{\Psi}(\z;a)(\e^{-\im\pi}\z)^{\frac{a}{2}\sigma_3}\begin{pmatrix} 1 & -\frac{\e^{\im\pi a}}{2\pi\im}\ln(\e^{-\im\pi}\z)\\ 0 & 1 \end{pmatrix}\begin{cases}\bigl(\begin{smallmatrix} 1 & 0 \\ -\e^{-\im\pi a}& 1 \end{smallmatrix}\bigr),&\textnormal{arg}\,\z\in(0,\frac{\pi}{3})\\ I,&\textnormal{arg}\,\z\in(\frac{\pi}{3},\frac{5\pi}{3})\\ \bigl(\begin{smallmatrix} 1 & 0 \\ \e^{\im\pi a}& 1 \end{smallmatrix}\bigr),&\textnormal{arg}\,\z\in(\frac{5\pi}{3},2\pi)\end{cases}.
	\end{equation*}
	Here, $\widehat{\Psi}(\z;a)$ is analytic at $\z=0$ and we choose principal branches for fractional exponents and logarithms.
	\item As $\z\rightarrow\infty,$ valid in a full neighborhood of infinity off the jump contours,
	\begin{align*}
		\Psi(\z;a)=&\,\big(\e^{-\im\pi}\z\big)^{-\frac{1}{4}\sigma_3}\frac{1}{\sqrt{2}}\begin{pmatrix} 1 & -1\\ 1 & 1 \end{pmatrix}\e^{-\im\frac{\pi}{4}\sigma_3}\left\{I+\frac{1}{8(\e^{-\im\pi}\z)^{\frac{1}{2}}}\begin{pmatrix}-(1+4a^2) & 2\im\\ 2\im & 1+4a^2 \end{pmatrix}+\mathcal{O}\left(\z^{-1}\right)\right\}\\
		&\,\times\exp\left[(\e^{-\im\pi}\z)^{\frac{1}{2}}\sigma_3\right]
	\end{align*}
	where $\z^{\alpha}$ is defined and analytic for $\z\in\mathbb{C}\backslash(-\infty,0]$ such that $\z^{\alpha}>0$ for $\z>0$.
\end{enumerate}
\end{problem}
\begin{figure}[tbh]
\begin{minipage}{0.4\textwidth} 
\begin{center}
\resizebox{0.6\textwidth}{!}{\includegraphics{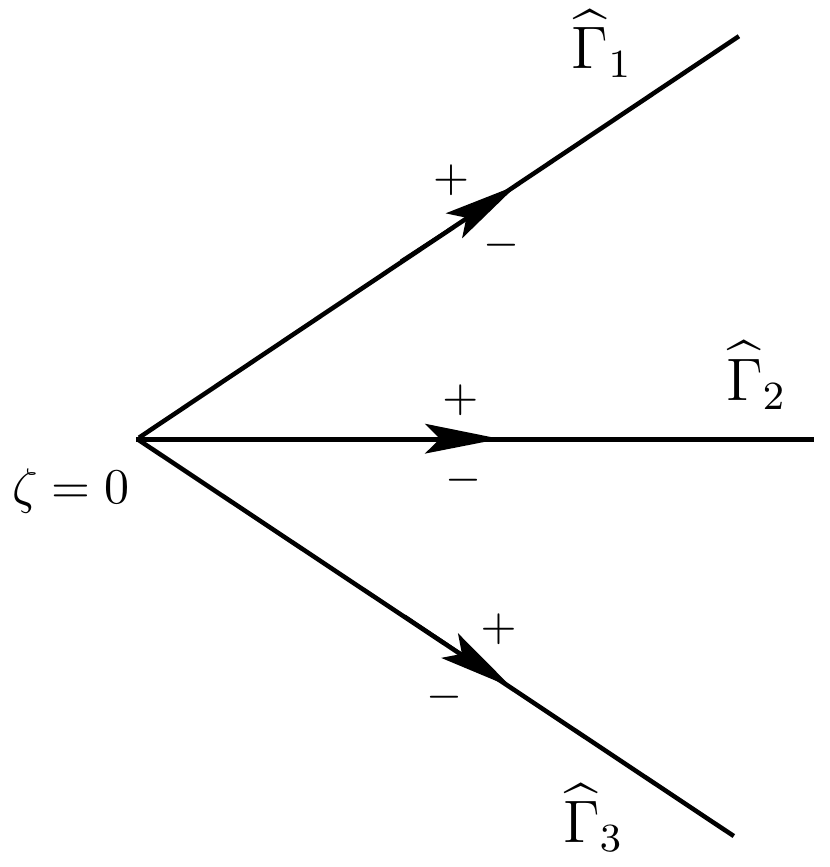}}
\caption{The oriented jump contours for the Bessel parametrix $\Psi(\z;a)$ in the complex $\z$-plane.}
\label{figure6B}
\end{center}
\end{minipage}
\begin{minipage}{0.4\textwidth}
\begin{center}
\resizebox{0.8\textwidth}{!}{\includegraphics{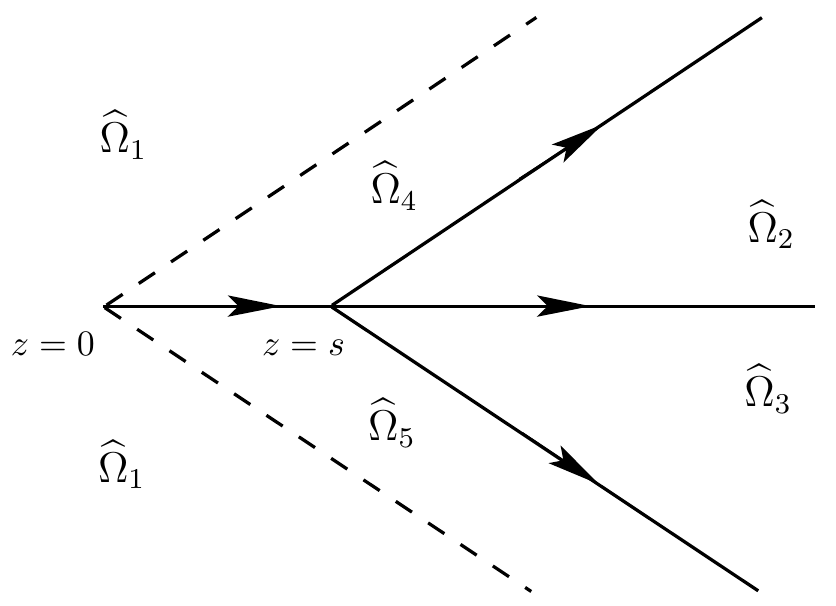}}
\caption{``Undressing" of RHP \ref{masterIIKS} with $s>0$. Jump contours of $X(z)$ as solid lines.}
\label{figure7}
\end{center}
\end{minipage}
\end{figure}

\begin{remark}\label{Besselconnect} We observe that, cf. \cite{NIST},
\begin{equation*}
	\e^{\im\frac{\pi}{2}a}I_a\big((\e^{-\im\pi}\z)^{\frac{1}{2}}\big)=J_a(\z^{\frac{1}{2}}),\ \ \ \z\in\mathbb{C}\backslash[0,\infty)
\end{equation*}
and hence for $\z>0$, we have
\begin{equation*}
	\big(\Psi_{a11}(\z)\big)_+=\sqrt{\pi}\,\e^{-\im\frac{\pi}{4}}J_a(\sqrt{\z}),\ \ \ \ \ \ \ \big(\Psi_{a21}(\z)\big)_+=\sqrt{\pi}\,\e^{-\im\frac{\pi}{4}}\sqrt{\z}J_a'(\sqrt{\z}).
\end{equation*}
This enables us to rewrite the Bessel kernel in terms of the entries of $\Psi_a(\z)$,
\begin{equation*}
	K_{\textnormal{Bess}}^{(a)}(\lambda,\mu)=\frac{\im}{2\pi}\frac{(\Psi_{a11}(\lambda))_+(\Psi_{a21}(\mu))_+-(\Psi_{a11}(\mu))_+(\Psi_{a21}(\lambda))_+}{\lambda-\mu},\ \ \ \lambda,\mu\in(0,s)
\end{equation*}
in which all limits are taken coming from the upper half-plane.
\end{remark}
In our next move we reduce RHP \ref{masterIIKS} with \eqref{BIIKS} to a problem with $z$-independent jumps and this step is the analogue of \eqref{center:1}, \eqref{center:2} for the Bessel kernel determinant. Let $s>0$ and set with help of Figure \ref{figure7},
\begin{equation*}
	X(z)=Y(z)\Psi(z;a)\begin{cases} I,&z\in\widehat{\Omega}_1\cup\widehat{\Omega}_2\cup\widehat{\Omega}_3\smallskip\\
	\bigl(\begin{smallmatrix} 1 & 0\smallskip\\ \e^{-\im\pi a} & 1 \end{smallmatrix}\bigr),&z\in\widehat{\Omega}_4\smallskip\\
	\bigl(\begin{smallmatrix} 1 & 0\smallskip\\ -\e^{\im\pi a} & 1 \end{smallmatrix}\bigr),&z\in\widehat{\Omega}_5 \end{cases}.
\end{equation*}

We arrive at the problem below:
\begin{problem}\label{Xstart} Determine $X(z)=X(z;s,\gamma;a)\in\mathbb{C}^{2\times 2}$ such that
\begin{enumerate}
	\item $X(z)$ is analytic for $z\in\mathbb{C}\backslash\big(\bigcup_{j=1}^4\widehat{\Gamma}_j^{(s)}\cup\{0,1\}\big)$ with
	\begin{equation*}
		\widehat{\Gamma}_1^{(s)}=(0,s),\ \ \ \widehat{\Gamma}_2^{(s)}=s+\e^{\im\frac{\pi}{3}}(0,\infty),\ \ \ \widehat{\Gamma}_3^{(s)}=(s,+\infty),\ \ \ \widehat{\Gamma}_4^{(s)}=s+\e^{\im\frac{5\pi}{3}}(0,\infty)
	\end{equation*}
	\item The jump conditions are as follows,
	\begin{equation*}
		X_+(z)=X_-(z)\bigl(\begin{smallmatrix}\e^{-\im\pi a} & 1-\gamma\\ 0 & \e^{\im\pi a} \end{smallmatrix}\bigr),\ \ z\in\widehat{\Gamma}_1^{(s)};\ \ \ \ \ \ \ X_+(z)=X_-(z)\bigl(\begin{smallmatrix}0 & 1\\ -1 & 0 \end{smallmatrix}\bigr),\ \ z\in\widehat{\Gamma}_3^{(s)};
	\end{equation*}
	and
	\begin{equation*}
		X_+(z)=X_-(z)\bigl(\begin{smallmatrix} 1 & 0\\ \e^{-\im\pi a} & 1 \end{smallmatrix}\bigr),\ \ z\in\widehat{\Gamma}_2^{(s)};\ \ \ \ \ \ \ X_+(z)=X_-(z)\bigl(\begin{smallmatrix} 1 & 0 \\ \e^{\im\pi a} & 1 \end{smallmatrix}\bigr),\ \ z\in\widehat{\Gamma}_4^{(s)}.
	\end{equation*}
	\item Near $z=s$, with $\mathbb{H}^{\pm}=\{z\in\mathbb{C}:\,\Im z\gtrless 0\}$,
	\begin{equation*}
		X(z)=\widehat{X}(z)\left[I+\frac{\gamma}{2\pi\im}\bigl(\begin{smallmatrix}-1 & -\e^{-\im\pi a}\\ \e^{\im\pi a} & 1 \end{smallmatrix}\bigr)\ln(z-s)\right]\begin{cases}\bigl(\begin{smallmatrix} 0 & 1\\ -1 & 0 \end{smallmatrix}\bigr),&\!\!\!z\in\mathbb{H}^+\\ I,&\!\!\!z\in\mathbb{H}^- \end{cases}\begin{cases}\bigl(\begin{smallmatrix}1 & 0\\ \e^{-\im\pi a} & 1 \end{smallmatrix}\bigr),&\!\!\!\textnormal{arg}(z-s)\in(\frac{\pi}{3},\pi)\\ \bigl(\begin{smallmatrix} 1 & 0\\ -\e^{\im\pi a} & 1 \end{smallmatrix}\bigr),&\!\!\!\textnormal{arg}(z-s)\in(\pi,\frac{5\pi}{3})\\ I,&\!\!\!\textnormal{else}\end{cases}		
	\end{equation*}
	where $\widehat{X}(z)$ is analytic at $z=s$ and we choose the principal branch for the logarithm. On the other hand, near $z=0$,
%
%
	\begin{eqnarray*}
		X(z)&=&\widehat{X}(z)(\e^{-\im\pi}z)^{\frac{a}{2}\sigma_3}\bigl(\begin{smallmatrix}1&\frac{\im}{2}\frac{1-\gamma}{\sin\pi a}\\ 0 & 1 \end{smallmatrix}\bigr),\ \ a\notin\mathbb{Z};\\
		X(z)&=&\widehat{X}(z)(\e^{-\im\pi}z)^{\frac{a}{2}\sigma_3}\left[I-\frac{\e^{\im\pi a}}{2\pi\im}(1-\gamma)\ln(\e^{-\im\pi}z)\bigl(\begin{smallmatrix}0 & 1\\ 0 & 0 \end{smallmatrix}\bigr)\right],\ a\in\mathbb{Z}.
	\end{eqnarray*}
	\item As $z\rightarrow\infty$,
	\begin{equation*}
		X(z)=(\e^{-\im\pi}z)^{-\frac{1}{4}\sigma_3}\frac{1}{\sqrt{2}}\begin{pmatrix}1 & -1\\ 1 & 1 \end{pmatrix}\e^{-\im\frac{\pi}{4}\sigma_3}\left\{I+X_{\infty}(\e^{-\im\pi}z)^{-\frac{1}{2}}+\mathcal{O}\left(z^{-1}\right)\right\}\exp\left[(\e^{-\im\pi}z)^{\frac{1}{2}}\sigma_3\right]
%
	\end{equation*}
	with
	\begin{equation*}
		X_{\infty}=-\frac{1}{2}\begin{pmatrix} 1 & \im\\ \im & -1 \end{pmatrix}Y_{\infty}^{12}+\frac{1}{8}\begin{pmatrix}-(1+4a^2) & 2\im \\ 2\im & 1+4a^2 \end{pmatrix}.
	\end{equation*}
\end{enumerate}
\end{problem}
\subsection{Differential identity} The required identity is derived along the same lines as for $D(J_{\textnormal{Ai}};\gamma)$, cf. \cite{CIK}. We skip details and summarize the final formula in the Proposition below.
\begin{prop} For fixed $\gamma\leq 1$, we have
\begin{equation}\label{diff:2}
	\frac{\partial}{\partial s}\ln D(J_{\textnormal{Bess}};\gamma)=-\frac{\gamma}{2\pi\im}\e^{\im\pi a}\big(X^{-1}(z)X'(z)\big)_{21}\Big|_{z\rightarrow s};\ \ \ \ (')=\frac{\d}{\d z}
\end{equation}
in terms of the solution $X(z)$ to RHP \ref{Xstart} and the limit is carried out for $\textnormal{arg}(z-s)\in(\frac{\pi}{3},\pi)$.
\end{prop}

\section{Nonlinear steepest descent analysis associated with $K_{\textnormal{Bess}}^{(a)}$ -- part 2} \label{Bessp2}
\subsection{Initial transformation}\label{Besspre:2} We choose $s>0$ and define
\begin{equation*}
	T(z)=X(sz),\ \ \ z\in\mathbb{C}\backslash(\Sigma_T\cup\{0,1\}),
\end{equation*}
where the contour $\Sigma_T$ consists of four line segment,
\begin{equation*}
	\Sigma_T=\bigcup_{j=1}^4\widehat{\Gamma}_{j,T};\ \ \ \ \widehat{\Gamma}_{1,T}=(0,1),\ \ \ \widehat{\Gamma}_{2,T}=1+\e^{\im\frac{\pi}{3}}(0,\infty),\ \ \ \widehat{\Gamma}_{3,T}=(1,+\infty),\ \ \ \widehat{\Gamma}_{4,T}=1+\e^{\im\frac{5\pi}{3}}(0,\infty);
\end{equation*}
oriented ``from left to right", see Figure \ref{figure7}. This leads us to
\begin{problem} Determine a function $T(z)=T(z;s,\gamma;a)\in\mathbb{C}^{2\times 2}$ which is uniquely determined by the following properties:
\begin{enumerate}
	\item $T(z)$ is analytic for $z\in\mathbb{C}\backslash(\Sigma_T\cup\{0,1\})$.
	\item We have the jump conditions
	\begin{equation*}
		T_+(z)=T_-(z)\bigl(\begin{smallmatrix} \e^{-\im\pi a} & 1-\gamma \\ 0 & \e^{\im\pi a} \end{smallmatrix}\bigr),\ \ z\in\widehat{\Gamma}_{1,T};\ \ \ \ \ \ T_+(z)=T_-(z)\bigl(\begin{smallmatrix} 0 & 1\\ -1 & 0 \end{smallmatrix}\bigr),\ \ z\in\widehat{\Gamma}_{3,T};
	\end{equation*}
	and
	\begin{equation*}
		T_+(z)=T_-(z)\bigl(\begin{smallmatrix} 1 & 0\\ \e^{-\im\pi a} & 1 \end{smallmatrix}\bigr),\ \ z\in\widehat{\Gamma}_{2,T};\ \ \ \ \ \ T_+(z)=T_-(z)\bigl(\begin{smallmatrix} 1 & 0 \\ \e^{\im\pi a} & 1 \end{smallmatrix}\bigr),\ \ z\in\widehat{\Gamma}_{4,T}.
	\end{equation*}
	\item The singular behavior near $z=0$ and $z=1$ is unchanged from the one stated in RHP \ref{Xstart}, modulo the change of variables $T(z)=X(sz)$.
	\item As $z\rightarrow\infty$, with $t=\sqrt{s}$,
	\begin{equation*}
		T(z)=(\e^{-\im\pi}sz)^{-\frac{1}{4}\sigma_3}\frac{1}{\sqrt{2}}\begin{pmatrix} 1 & -1\\ 1& 1 \end{pmatrix}\e^{-\im\frac{\pi}{4}\sigma_3}\Big\{I+X_{\infty}(\e^{-\im\pi}sz)^{-\frac{1}{2}}+\mathcal{O}\left(z^{-1}\right)\Big\}
		\exp\left[t(\e^{-\im\pi}z)^{\frac{1}{2}}\sigma_3\right].
	\end{equation*}
\end{enumerate}
\end{problem}
\subsection{Normalization transformation} We introduce the $g$-function, for $z\in\mathbb{C}\backslash\mathbb{R}$,
\begin{equation}\label{g:2n}
	g(z)=-\big(\e^{-\im\pi}(z-1)\big)^{\frac{1}{2}}+V\ln\left(\frac{1+(\e^{-\im\pi}(z-1))^{\frac{1}{2}}}{1-(\e^{-\im\pi}(z-1))^{\frac{1}{2}}}\right),\ \ \ V=\frac{\chi}{t}
\end{equation}
where $\z^{\alpha}$ as well as $\ln\z$ are defined and analytic for $\z\in\mathbb{C}\backslash(-\infty,0]$. We also choose
\begin{equation*}
	-\pi<\textnormal{arg}\left(\frac{1+(\e^{-\im\pi}(z-1))^{\frac{1}{2}}}{1-(\e^{-\im\pi}(z-1))^{\frac{1}{2}}}\right)\leq\pi.
\end{equation*}
The relevant analytical properties of the $g$-function are summarized below.
\begin{prop} The function $g(z)$ defined in \eqref{g:2n} is analytic for $z\in\mathbb{C}\backslash\big((-\infty,0)\cup(1,+\infty)\big)$, more precisely,
\begin{equation*}
	g_{\pm}(z)=-\sqrt{1-z}+V\ln\left(\frac{\sqrt{1-z}+1}{\sqrt{1-z}-1}\right)\mp\im\pi V,\ \ \ z\in(-\infty,0);
\end{equation*}
and
\begin{equation*}
	g_{\pm}(z)=-\sqrt{1-z}+V\ln\left(\frac{1+\sqrt{1-z}}{1-\sqrt{1-z}}\right),\ \ \ z\in(0,1);
\end{equation*}
as well as
\begin{equation*}
	g_{\pm}(z)=\pm\im\sqrt{z-1}\pm\im V\textnormal{arg}\left(\frac{1-\im\sqrt{z-1}}{1+\im\sqrt{z-1}}\right),\ \ \ z\in(1,+\infty).
\end{equation*}
Also for $z\rightarrow\infty,z\notin\mathbb{R}$,
\begin{equation*}
	(\e^{-\im\pi}z)^{\frac{1}{2}}+g(z)=\im\pi V+\left(2V-\frac{1}{2}\right)(\e^{-\im\pi}z)^{-\frac{1}{2}}+\mathcal{O}\left(z^{-1}\right).
\end{equation*}
\end{prop}
In the next transformation we set
\begin{equation*}
	S(z)=\e^{-\im\pi tV\sigma_3}T(z)\e^{tg(z)\sigma_3},\ \ \ z\in\mathbb{C}\backslash\big(\Sigma_T\cup\{0,1\}\big)
\end{equation*}
and are lead to the problem below.
\begin{problem}\label{g2RHP} Determine $S(z)=S(z;s,\gamma;a)\in\mathbb{C}^{2\times 2}$ such that
\begin{enumerate}
	\item  $S(z)$ is analytic for $z\in\mathbb{C}\backslash(\Sigma_T\cup(-\infty)\cup\{0,1\})$.
	\item The limiting values $S_{\pm}(z),z\in\Sigma_T$ are related by the equations
	 \begin{equation*}
		S_+(z)=S_-(z)\begin{pmatrix} 1 & 0\\ \e^{-\im\pi a}\e^{2tg(z)} & 1 \end{pmatrix},\ \ z\in\widehat{\Gamma}_{2,T};\ \ \ \ \ \ \ S_+(z)=S_-(z)\begin{pmatrix} 1 & 0\\ \e^{\im\pi a}\e^{2tg(z)} & 1 \end{pmatrix},\ \ z\in\widehat{\Gamma}_{4,T};
	\end{equation*}
	and
	\begin{equation*}
		S_+(z)=S_-(z)\begin{pmatrix} 0 & 1\\ -1 & 0 \end{pmatrix},\ \ z\in\widehat{\Gamma}_{3,T};\ \ \ \ \ \ \ S_+(z)=S_-(z)\e^{-2\pi\im tV\sigma_3},\ \ z\in(-\infty,0);
	\end{equation*}
	as well as
	\begin{equation*}
		S_+(z)=S_-(z)\begin{pmatrix}\e^{-\im\pi a} & \e^{-t(\varkappa_{\textnormal{Bess}}+2g(z))}\\ 0 & \e^{\im\pi a}\end{pmatrix},\ \ z\in\widehat{\Gamma}_{1,T}.
	\end{equation*}
	\item Near $z=1$ and $z=0$ with $\widehat{S}(z)$ analytic at either point,
	\begin{eqnarray*}
		\e^{\im\pi tV\sigma_3}S(z)\e^{-tg(z)\sigma_3}&=&\widehat{S}(z)\left[I+\frac{\gamma}{2\pi\im}\bigl(\begin{smallmatrix}-1 & -1\\ 1 & 1 \end{smallmatrix}\bigr)\ln(z-1)\right]\e^{\im\frac{\pi}{2}a\sigma_3}\begin{cases}
		\bigl(\begin{smallmatrix}0 & 1\\ -1 & 0 \end{smallmatrix}\bigr),&z\in\mathbb{H}^+\\ I,&z\in\mathbb{H}^- \end{cases}\\
		&&\,\,\times\begin{cases} \bigl(\begin{smallmatrix} 1 & 0\\ \e^{-\im\pi a} & 1 \end{smallmatrix}\bigr),&\textnormal{arg}(z-1)\in(\frac{\pi}{3},\pi)\\ \bigl(\begin{smallmatrix} 1 & 0\\ -\e^{\im\pi a} & 1\end{smallmatrix}\bigr),&\textnormal{arg}(z-1)\in(\pi,\frac{5\pi}{3})\\ I,&\textnormal{else}\end{cases},\ \ \ \ \ z\rightarrow 1,\ z\notin\Sigma_T;
%
%
	\end{eqnarray*}
	and, as $z\rightarrow 0,z\notin\Sigma_T$,
	\begin{eqnarray*}
		\e^{\im\pi tV\sigma_3}S(z)\e^{-tg(z)\sigma_3}&=&\widehat{S}(z)(\e^{-\im\pi}sz)^{\frac{a}{2}\sigma_3}\begin{pmatrix} 1 & \frac{\im}{2}\frac{1-\gamma}{\sin\pi a}\\ 0 & 1 \end{pmatrix},\ \ \ a\notin\mathbb{Z}\\
		 \e^{\im\pi tV\sigma_3}S(z)\e^{-tg(z)\sigma_3}&=&\widehat{S}(z)(\e^{-\im\pi}sz)^{\frac{a}{2}\sigma_3}\left[I-\frac{\e^{\im\pi a}}{2\pi\im}(1-\gamma)\ln(\e^{-\im\pi}sz)\bigl(\begin{smallmatrix}0 & 1\\ 0 & 0 \end{smallmatrix}\bigr)\right],\ \ \ a\in\mathbb{Z}.
	\end{eqnarray*}
	\item As $z\rightarrow\infty$,
	\begin{equation*}
		S(z)=(\e^{-\im\pi}sz)^{-\frac{1}{4}\sigma_3}\e^{-\im\pi tV\sigma_3}\frac{1}{\sqrt{2}}\begin{pmatrix}1 & -1\\ 1 & 1 \end{pmatrix}\e^{-\im\frac{\pi}{4}\sigma_3}\e^{\im\pi tV\sigma_3}\left\{I+S_{\infty}(\e^{-\im\pi}z)^{-\frac{1}{2}}+\mathcal{O}\left(z^{-1}\right)\right\}
	\end{equation*}
	with
	\begin{equation*}
		S_{\infty}=\e^{-\im\pi tV\sigma_3}X_{\infty}\e^{\im\pi tV\sigma_3}s^{-\frac{1}{2}}+t\left(2V-\frac{1}{2}\right)\sigma_3.
	\end{equation*}
\end{enumerate}
\end{problem}
As before, the following observations are crucial: With $0<r<\frac{1}{4}$ fixed,
\begin{equation}\label{crux:4}
	\Re\big(g(z)\big)<0,\ \ \ z\in\big(\widehat{\Gamma}_{2,T}\cup\widehat{\Gamma}_{4,T}\big)\backslash D(1,r)
\end{equation}
and for $z\in(0,1)\backslash D(0,r)$ and sufficiently large $t\geq t_0,v\geq v_0$ subject to \eqref{Bscale},
\begin{equation}\label{crux:5}
	\varkappa_{_\textnormal{Bess}}+2g(z)=\varkappa_{_{\textnormal{Bess}}}-2\sqrt{1-z}+2V\ln\left(\frac{1+\sqrt{1-z}}{1-\sqrt{1-z}}\right)\geq\delta>0.
\end{equation}
Hence, we expect the major contribution to the asymptotic solution of RHP \ref{g2RHP} to arise from the line segments $(-\infty,0)\cup(0,1)\cup(1,+\infty)$ and two small neighborhoods of $z=0$ as well as $z=1$.
\begin{remark} We will again represent $tV=\chi$ as
\begin{equation*}
	\chi=k+\alpha;\ \ \ \ \ k\in\mathbb{Z}_{\geq 0},\ \ -\frac{1}{2}\leq\alpha<\frac{1}{2},
\end{equation*}
and hence the jump on $(-\infty,0)$ in RHP \ref{g2RHP} equals
\begin{equation*}
	S_+(z)=S_-(z)\e^{-2\pi\im\alpha\sigma_3},\ \ \ z<0.
\end{equation*}
\end{remark}
\subsection{Analysis of model Riemann-Hilbert problems} The outer parametrix, for $z\in\mathbb{C}\backslash\mathbb{R}$,
\begin{equation}\label{Besselout}
	P^{(\infty)}(z)=\big(\e^{-\im\pi}s(z-1)\big)^{-\frac{1}{4}\sigma_3}\e^{-\im\pi\alpha\sigma_3}\frac{1}{\sqrt{2}}\begin{pmatrix}1 & -1\\ 1 & 1 \end{pmatrix}\e^{-\im\frac{\pi}{4}\sigma_3}
	\left(\frac{(\e^{-\im\pi}(z-1))^{\frac{1}{2}}+1}{(\e^{-\im\pi}(z-1))^{\frac{1}{2}}-1}\right)^{-\frac{a}{2}\sigma_3}\big(\mathcal{D}(z)\big)^{\sigma_3},
\end{equation}
with the Szeg\H{o} function
\begin{equation*}
	\mathcal{D}(z)=\left(\frac{1+\big(\e^{-\im\pi}(z-1)\big)^{\frac{1}{2}}}{1-\big(\e^{-\im\pi}(z-1)\big)^{\frac{1}{2}}}\right)^{\alpha},\ \ \ z\in\mathbb{C}\backslash\mathbb{R},\ \ \ \ \ \ \mathcal{D}(z)=\e^{\im\pi\alpha}\left(1+\frac{2\alpha}{(\e^{-\im\pi}z)^{\frac{1}{2}}}+\mathcal{O}\left(z^{-1}\right)\right),\ \ z\rightarrow\infty;
\end{equation*}
displays the properties summarized below.
\begin{problem} The parametrix $P^{(\infty)}(z)$ has the following analytical properties
\begin{enumerate}
	\item $P^{(\infty)}(z)$ is analytic for $z\in\mathbb{C}\backslash\mathbb{R}$ and we orient the real axis from left to right.
	\item The limiting values $P^{(\infty)}_{\pm}(z),z\in\mathbb{R}$ are square integrable and related by the jump conditions
	\begin{eqnarray*}
		P_+^{(\infty)}(z)&=&P_-^{(\infty)}(z)\e^{-2\pi\im\alpha\sigma_3},\ \ z\in(-\infty,0);\ \ \ \ \ \ P_+^{(\infty)}(z)=P_-^{(\infty)}(z)\e^{-\im\pi a\sigma_3},\ \ z\in(0,1);\\
		P_+^{(\infty)}(z)&=&P_-^{(\infty)}(z)\bigl(\begin{smallmatrix} 0 & 1\\ -1 & 0 \end{smallmatrix}\bigr),\ \ z\in(1,+\infty)
	\end{eqnarray*}
	\item As $z\rightarrow\infty,z\notin\mathbb{R}$,
	\begin{equation*}
		P^{(\infty)}(z)=(\e^{-\im\pi}sz)^{-\frac{1}{4}\sigma_3}\e^{-\im\pi\alpha\sigma_3}\frac{1}{\sqrt{2}}\begin{pmatrix} 1 & -1\\ 1 & 1 \end{pmatrix}\e^{-\im\frac{\pi}{4}\sigma_3}\e^{\im\pi\alpha\sigma_3}\Big\{I+\mathcal{O}\left((\e^{-\im\pi}z)^{-\frac{1}{2}}\right)\Big\}.
	\end{equation*}
\end{enumerate}
\end{problem}
For the local parametrix near $z=1$ we use ideas from \cite{BDIK}, Section $3.4$. Consider the unimodular function
\begin{equation}\label{hankelbare}
	B(\z)=\frac{1}{2}\sqrt{\pi}\begin{pmatrix} 1 & 0\\ 0 & \im\end{pmatrix}\begin{pmatrix} H_0^{(1)}(\z^{\frac{1}{2}}) & -H_0^{(2)}(\z^{\frac{1}{2}})\smallskip\\ -\z^{\frac{1}{2}}\big(H_0^{(1)}\big)'(\z^{\frac{1}{2}}) & \z^{\frac{1}{2}}\big(H_0^{(2)}\big)'(\z^{\frac{1}{2}}) \\ \end{pmatrix},\ \ \ \z\in\mathbb{C}\backslash(-\infty,0]
\end{equation}
which uses the Hankel functions $H_0^{(1)}$ and $H_0^{(2)}$ and the principal branch for $\z^{\frac{1}{2}}:\,\textnormal{arg}\,\z\in(-\pi,\pi]$. 
\begin{problem}\label{Hankelmodel} The function $B(\z)\in\mathbb{C}^{2\times 2}$ defined in \eqref{hankelbare} has the following properties:
\begin{enumerate}
	\item $B(\z)$ is analytic for $\z\in\mathbb{C}\backslash(-\infty,0]$.
	\item Provided we orient the negative half-ray from $-\infty$ to the origin, we have
	\begin{equation*}
		B_+(\z)=B_-(\z)\begin{pmatrix} 0 & -1\\ 1 & 2\end{pmatrix},\ \ \ \ \z\in(-\infty,0).
	\end{equation*}
	\item As $\z\rightarrow 0,\z\notin(-\infty,0]$,
	\begin{equation*}
		B(\z)=\frac{1}{2}\sqrt{\pi}\begin{pmatrix}1 & 0\\ 0 & \im \end{pmatrix}\widehat{B}(\z)\left[I+\frac{\ln\z}{2\pi\im}\begin{pmatrix} -1 & -1\\ 1 & 1 \end{pmatrix}\right],
	\end{equation*}
	with the principal branch for the logarithm and
	\begin{eqnarray*}
		\widehat{B}(\z)&=&\begin{pmatrix} J_0(\z^{\frac{1}{2}})(1-\frac{2\im}{\pi}\ln 2) & -J_0(\z^{\frac{1}{2}})(1+\frac{2\im}{\pi}\ln 2)\smallskip\\ -\z^{\frac{1}{2}}J_0'(\z^{\frac{1}{2}})(1-\frac{2\im}{\pi}\ln 2)-\frac{2\im}{\pi}J_0(\z^{\frac{1}{2}}) & \z^{\frac{1}{2}}J_0'(\z^{\frac{1}{2}})(1+\frac{2\im}{\pi}\ln 2)-\frac{2\im}{\pi}J_0(\z^{\frac{1}{2}}) \end{pmatrix}\\
		&&-\frac{2\im}{\pi}\sum_{k=0}^{\infty}\psi(k+1)\begin{pmatrix}1 & 1\\ -2k& -2k \end{pmatrix}\frac{(-\frac{\z}{4})^k}{(k!)^2},\ \ |\z|<r;\ \ \ \ \begin{cases}J_0(\z^{\frac{1}{2}})=\sum_{0}^{\infty}\frac{(-\frac{1}{4}\z)^k}{(k!)^2}&\\ \z^{\frac{1}{2}}J_0'(\z^{\frac{1}{2}})=\sum_{1}^{\infty}2k\frac{(-\frac{1}{4}\z)^k}{(k!)^2}&\\ \end{cases}.
	\end{eqnarray*}
	\item The function $B(\z)$ is normalized so that as $\z\rightarrow\infty,\z\notin\mathbb{R}$,
	\begin{equation}\label{Hankelasy}
		B(\z)=\z^{-\frac{1}{4}\sigma_3}\frac{1}{\sqrt{2}}\begin{pmatrix} 1 & -1\\ 1 & 1 \end{pmatrix}\e^{-\im\frac{\pi}{4}\sigma_3}\left\{I+\frac{\im}{8\z^{\frac{1}{2}}}\begin{pmatrix}1&-2\im\\ -2\im & -1\end{pmatrix}+\frac{3}{128\z}\begin{pmatrix}1 & 4\im\\ -4\im & 1\end{pmatrix}+\mathcal{O}\left(\z^{-\frac{3}{2}}\right)\right\}\e^{\im\z^{\frac{1}{2}}\sigma_3}.
	\end{equation}
\end{enumerate}
\end{problem}
In terms of $B(\z)$ we can now define the local parametrix as follows: consider
\begin{align}\label{Hankelpara}
	\e^{\im\pi tV\sigma_3}&P^{(1)}(z)\e^{-tg(z)\sigma_3}=E^{(1)}(z)B\big(\z(z)\big)\left[I-\frac{1-\gamma}{2\pi\im}\begin{pmatrix}-1 & -1 \\ 1 & 1 \end{pmatrix}\ln\big(\z(z)\big)\right]\e^{\im\frac{\pi}{2}a\sigma_3}\nonumber\\
	&\times\,\begin{cases}\bigl(\begin{smallmatrix}0 & 1\\ -1 & 0 \end{smallmatrix}\bigr),&z\in \mathbb{H}^+\\ I,&z\in \mathbb{H}^-\end{cases}\,\times\begin{cases} \bigl(\begin{smallmatrix} 1 & 0\\ \e^{-\im\pi a} & 1 \end{smallmatrix}\bigr),&\textnormal{arg}(z-1)\in(\frac{\pi}{3},\pi)\\ \bigl(\begin{smallmatrix} 1 & 0\\ -\e^{\im\pi a} & 1\end{smallmatrix}\bigr),&\textnormal{arg}(z-1)\in(\pi,\frac{5\pi}{3})\\ I,&\textnormal{else}\end{cases};\ \ \ \ \ 0<|z-1|<\frac{1}{4}
\end{align}
where we make use of the locally conformal change of variables
\begin{equation*}
	\z(z)=t^2\e^{\im\pi}\left\{g(z)-\frac{\alpha}{t}\ln\left(\frac{1+(\e^{-\im\pi}(z-1))^{\frac{1}{2}}}{1-(\e^{-\im\pi}(z-1))^{\frac{1}{2}}}\right)\right\}^2=t^2(z-1)\left(1-\frac{2k}{t}\right)^2\Big\{1+\frac{4k(z-1)}{3(t-2k)}+\mathcal{O}\left((z-1)^2\right)\Big\},
\end{equation*}
so that
\begin{equation}\label{localtrick}
	\z^{\frac{1}{2}}(z)=\e^{-\im\frac{\pi}{2}}t\,\textnormal{sgn}(\Im z)\left\{g(z)-\frac{\alpha}{t}\ln\left(\frac{1+(\e^{-\im\pi}(z-1))^{\frac{1}{2}}}{1-(\e^{-\im\pi}(z-1))^{\frac{1}{2}}}\right)\right\}, \ \ z\in D\left(1,\frac{1}{4}\right)\Big\backslash\left(\frac{3}{4},\frac{5}{4}\right)
\end{equation}
and the locally analytic multiplier, for $|z-1|<\frac{1}{4}$,
\begin{equation*}
	E^{(1)}(z)=\e^{\im\pi tV\sigma_3}P^{(\infty)}(z)\big(\mathcal{D}(z)\big)^{-\sigma_3}\begin{cases}\bigl(\begin{smallmatrix} 0 & -1\\ 1 & 0 \end{smallmatrix}\bigr),&\!\!\!z\in\mathbb{H}^+\\ I,&\!\!\!z\in\mathbb{H}^-\end{cases}
	\times\,\e^{-\im\frac{\pi}{2}a\sigma_3}\e^{\im\frac{\pi}{4}\sigma_3}\frac{1}{\sqrt{2}}\begin{pmatrix}1 & 1\\ -1 & 1 \end{pmatrix}\big(\z(z)\big)^{\frac{1}{4}\sigma_3}.
\end{equation*}
It is easy to verify the properties of $P^{(1)}(z)$ which are summarized below:
\begin{problem}\label{HankelRHP} The parametrix $P^{(1)}(z)$ defined in \eqref{Hankelpara} has the following properties:
\begin{enumerate}
	\item $P^{(1)}(z)$ is analytic for $z\in D(1,\frac{1}{4})\backslash(\Sigma_T\cup\{1\})$.
	\item Since the change of variables $\z\mapsto\z(z),z\in D(1,\frac{1}{4})$ is locally conformal, we obtain from RHP \ref{Hankelmodel} directly the jump behavior 
	\begin{eqnarray*}
		P^{(1)}_+(z)&=&P^{(1)}_-(z)\begin{pmatrix} 1 & 0\\ \e^{-\im\pi a}\e^{2tg(z)} & 1 \end{pmatrix},\ \ z\in\widehat{\Gamma}_{2,T}\cap D\left(1,\frac{1}{4}\right),\\
		P^{(1)}_+(z)&=&P^{(1)}_-(z)\begin{pmatrix} 1 & 0\\ \e^{\im\pi a}\e^{2tg(z)} & 1 \end{pmatrix},\ \ z\in\widehat{\Gamma}_{4,T}\cap D\left(1,\frac{1}{4}\right);
	\end{eqnarray*}
	and
	\begin{eqnarray*}
		P^{(1)}_+(z)&=&P^{(1)}_-(z)\begin{pmatrix} 0 & 1\\ -1 & 0 \end{pmatrix},\ \ z\in\left(1,\frac{5}{4}\right);\\
		P^{(1)}_+(z)&=&P^{(1)}_-(z)\begin{pmatrix}\e^{-\im\pi a} & \e^{-t(\varkappa_{\textnormal{Bess}}+2g(z))}\\ 0 & \e^{\im\pi a}\end{pmatrix},\ \ z\in\left(\frac{3}{4},1\right).
	\end{eqnarray*}
	This matches exactly the jump behavior near $z=1$ of RHP \ref{g2RHP}.
	\item Near $z=1$, with the principal branch for the logarithm,
	\begin{align*}
		\e^{\im\pi tV\sigma_3}P^{(1)}(z)\e^{-tg(z)\sigma_3}=&\,\widehat{P}^{(1)}(z)\left[I+\frac{\gamma}{2\pi\im}\begin{pmatrix}-1 & -1\\ 1 & 1\end{pmatrix}\ln(z-1)\right]\e^{\im\frac{\pi}{2}a\sigma_3}\begin{cases}\bigl(\begin{smallmatrix} 0 & 1\\ -1 & 0 \end{smallmatrix}\bigr),&z\in\mathbb{H}^+\cap D(1,\frac{1}{4})\\ I,&z\in\mathbb{H}^-\cap D(1,\frac{1}{4})\end{cases}\\
		&\,\times\begin{cases} \bigl(\begin{smallmatrix} 1 & 0\\ \e^{-\im\pi a} & 1 \end{smallmatrix}\bigr),&\textnormal{arg}(z-1)\in(\frac{\pi}{3},\pi)\\ \bigl(\begin{smallmatrix} 1 & 0\\ -\e^{\im\pi a} & 1\end{smallmatrix}\bigr),&\textnormal{arg}(z-1)\in(\pi,\frac{5\pi}{3})\\ I,&\textnormal{else}\end{cases},\ \ \ \ \ z\rightarrow 1,\ z\notin\Sigma_T,
	\end{align*}
	and we have thus the same singular structure near $z=1$ as in RHP \ref{g2RHP}.
	\item In case $t\rightarrow+\infty,\gamma\uparrow 1$ subject to \eqref{Bscale}, we obtain from \eqref{Hankelasy} and \eqref{localtrick},
	\begin{align}\label{Hmatch:1}
		P^{(1)}&(z)=P^{(\infty)}(z)\bigg\{I+\frac{\im}{8\z^{\frac{1}{2}}(z)}M(z)\begin{pmatrix}1 & -2\im\e^{-\im\pi a}\\ -2\im\e^{\im\pi a} & -1\end{pmatrix}\big(M(z)\big)^{-1}+\frac{3}{128\z(z)}\\
		&\times\,M(z)\begin{pmatrix}1 & 4\im\e^{-\im\pi a}\\ -4\im\e^{\im\pi a} & -1\end{pmatrix}\big(M(z)\big)^{-1}+\mathcal{O}\left(t^{-3}\right)\bigg\};\ \ \ M(z)=\big(\mathcal{D}(z)\big)^{-\sigma_3}
		\begin{cases}\bigl(\begin{smallmatrix}0&-1\\ 1&0\end{smallmatrix}\bigr),&z\in\mathbb{H}^+\\ I,&z\in\mathbb{H}^-\end{cases},\nonumber
	\end{align}
	uniformly for $0<r_1\leq|z-1|\leq r_2<\frac{1}{4}$. Here we use that on the annulus
	\begin{equation*}
		\exp\left[-\varkappa_{_\textnormal{Bess}}t\pm2\im\z^{\frac{1}{2}}(z)\right]=\mathcal{O}\left(t^{-\infty}\right),\ \ \ t\rightarrow+\infty,\ \gamma\uparrow 1:\ \ \varkappa_{_\textnormal{Bess}}=2-2\left(\chi+\frac{a}{2}\right)\frac{\ln t}{t}.
	\end{equation*}
\end{enumerate}
\end{problem}
\begin{remark}\label{cem:3} Observe that from properties $(2)$ and $(3)$ in RHP \ref{HankelRHP} we obtain
\begin{equation*}
	S(z)=N_1(z)P^{(1)}(z),\ \ \ 0\leq|z-1|<\frac{1}{4},
\end{equation*}
with $N_1(z)$ analytic at $z=1$.
\end{remark}
For the remaining parametrix near $z=0$ we shall use again classical orthogonal polynomials: define for $k\in\mathbb{Z}_{\geq 0}$, with principal branches for fractional exponents,
\begin{equation}\label{lagbare}
	L(\z;a)=\begin{pmatrix}q_k^{(a)}(\z) & \frac{1}{2\pi\im}\int_0^{\infty}q_k^{(a)}(t)\,t^a\e^{-t}\frac{\d t}{t-\z}\smallskip\\ \gamma_{k-1}q_{k-1}^{(a)}(\z) & \frac{\gamma_{k-1}}{2\pi\im}\int_0^{\infty}q_{k-1}^{(a)}(t)\,t^a\e^{-t}\frac{\d t}{t-\z}
	\end{pmatrix}(\e^{-\im\pi}\z)^{\frac{a}{2}\sigma_3}\e^{-\frac{1}{2}\z\sigma_3},\ \ \ \z\in\mathbb{C}\backslash[0,\infty),
\end{equation}
in terms of monic Laguerre polynomials $\{q_k^{(a)}(\z)\}_{k\in\mathbb{Z}_{\geq 0}}, q_{-1}^{(a)}(\z)\equiv 0=\gamma_{k-1}$, cf. \cite{NIST}:
\begin{equation*}
	q_k^{(a)}(\z)=\z^k+a_{k,k-1}\z^{k-1}+a_{k,k-2}\z^{k-2}+\mathcal{O}\left(\z^{k-3}\right),\ \ \z\rightarrow\infty;\ \ \ \ \int_0^{\infty}q_j^{(a)}(t)q_k^{(a)}(t)\,t^a\e^{-t}\d t=h_k\delta_{jk},
\end{equation*}
with
\begin{equation*}
	h_k=k!\,\Gamma(1+k+a),\ \ k\in\mathbb{Z}_{\geq 0};\ \ \ a_{k,k-1}=-k(k+a),\ \ a_{k,k-2}=\frac{k}{2}(k+a)(k-1)(k+a-1),\ \ k\in\mathbb{Z}_{\geq 0}.
\end{equation*}
\begin{remark} The normalization in \cite{NIST} of the Laguerre polynomials $\{L_k^{(a)}(\z)\}_{k\in\mathbb{Z}_{\geq 0}}$ is again different from the one we choose: we have to use the relation $q_k^{(a)}(\z)=(-1)^kk!L_k^{(a)}(\z),\z\in\mathbb{C}$.
\end{remark}
\begin{problem}\label{lagbareRHP} For any $k\in\mathbb{Z}_{\geq 0},a>-1$, the function $L(\z;a)\in\mathbb{C}^{2\times 2}$ defined in \eqref{lagbare} has the following properties
\begin{enumerate}
	\item $L(\z;a)$ is analytic for $\z\in\mathbb{C}\backslash[0,+\infty)$ and we orient the positive half-ray from $0$ to $+\infty$.
	\item Along this ray we have
	\begin{equation*}
		L_+(\z;a)=L_-(\z;a)\begin{pmatrix}\e^{-\im\pi a} & 1\\ 0 & \e^{\im\pi a} \end{pmatrix},\ \ \z\in(0,\infty).
	\end{equation*}
	\item Near $\z=0$,
	\begin{eqnarray*}
		L(\z;a)&=&\widehat{L}(\z;a)(\e^{-\im\pi}\z)^{\frac{a}{2}\sigma_3}\begin{pmatrix} 1 & \frac{\im}{2}\frac{1}{\sin\pi a}\\ 0 & 1\end{pmatrix},\ \ \ a\notin\mathbb{Z};\\
		L(\z;a)&=&\widehat{L}(\z;a)(\e^{-\im\pi}\z)^{\frac{a}{2}\sigma_3}\left[I-\frac{\e^{\im\pi a}}{2\pi\im}\ln\big(\e^{-\im\pi}\z\big)\bigl(\begin{smallmatrix}0 & 1\\ 0 & 0 \end{smallmatrix}\bigr)\right],\ \ \ a\in\mathbb{Z};
	\end{eqnarray*}
	with $\widehat{L}(\z;a)$ analytic at $\z=0$.
	\item As $\z\rightarrow\infty,\z\notin[0,\infty)$, with $\gamma_k=-\frac{2\pi\im}{h_k}$,
	\begin{equation}\label{Lagasy}
		L(\z;a)=\left\{I+\frac{1}{\z}\begin{pmatrix} -\frac{\gamma_{k-1}}{\gamma_k}& \gamma_k^{-1}\\ \gamma_{k-1}&\frac{\gamma_{k-1}}{\gamma_k}\end{pmatrix}+\frac{1}{\z^2}\begin{pmatrix} \frac{1}{2}\frac{\gamma_{k-2}}{\gamma_k} & \gamma_{k+1}^{-1}\\ -\gamma_{k-2}& \frac{1}{2}\frac{\gamma_{k-1}}{\gamma_{k+1}}\end{pmatrix}+\mathcal{O}\left(\z^{-3}\right)\right\}\z^{k\sigma_3}(\e^{-\im\pi}\z)^{\frac{a}{2}\sigma_3}\e^{-\frac{1}{2}\z\sigma_3}.
	\end{equation}
\end{enumerate}
\end{problem}
The required model function near $z=0$ is now given by
\begin{equation}\label{Lagpara}
	P^{(0)}(z)=E^{(0)}(z)L\big(\z(z);a\big)\e^{\frac{t}{2}\varkappa_{\textnormal{Bess}}\sigma_3}\e^{tg(z)\sigma_3},\ \ \ \ \ z\in D\left(0,\frac{1}{4}\right)\Big\backslash\left(-\frac{1}{4},\frac{1}{4}\right),
\end{equation}
where
\begin{equation*}
	E^{(0)}(z)=P^{(\infty)}(z)\big(\mathcal{D}(z)\big)^{-\sigma_3}\big(\e^{-\im\pi}\z(z)\big)^{-\frac{a}{2}\sigma_3}t^{\frac{a}{2}\sigma_3}\big(\beta(z)\big)^{-\sigma_3},\ \ \ 0\leq|z|<\frac{1}{4},
\end{equation*}
with the choice
\begin{equation*}
	\beta(z)=\left(\z(z)\frac{1+(\e^{-\im\pi}(z-1))^{\frac{1}{2}}}{1-(\e^{-\im\pi}(z-1))^{\frac{1}{2}}}\right)^kt^{-\chi}=2^{2k}t^{-\alpha}\left\{1-\frac{k}{4}z+\mathcal{O}\left(z^2\right)\right\},\ \ z\rightarrow0,
\end{equation*}
is analytic at $z=0$. In \eqref{Lagpara} we have employed the locally conformal change of variables
\begin{equation*}
	\z(z)=2t\left(g(z)-\frac{\chi}{t}\ln\left(\frac{1+(\e^{-\im\pi}(z-1))^{\frac{1}{2}}}{1-(\e^{-\im\pi}(z-1))^{\frac{1}{2}}}\right)+1\right)=tz\left\{1+\frac{z}{4}+\mathcal{O}\left(z^2\right)\right\},\ \ \ |z|<\frac{1}{4}
\end{equation*}
and we collect the relevant properties of $P^{(0)}(z)$ below.
\begin{problem}\label{LagparaRHP} The parametrix $P^{(0)}(z)$ has the following analytical properties
\begin{enumerate}
	\item $P^{(0)}(z)$ is analytic for $z\in D(0,\frac{1}{4})\backslash(-\frac{1}{4},\frac{1}{4})$.
	\item With local analyticity of $\z=\z(z)$,
	\begin{eqnarray*}
		P^{(0)}_+(z)&=&P^{(0)}_-(z)\e^{-2\pi\im tV\sigma_3},\ \ \ z\in\left(-\frac{1}{4},0\right);\\
		P^{(0)}_+(z)&=&P^{(0)}_-(z)\begin{pmatrix}\e^{-\im\pi a} & \e^{-t(\varkappa_{\textnormal{Bess}}+2g(z))}\\ 0 & \e^{\im\pi a}\end{pmatrix},\ \ \ z\in\left(0,\frac{1}{4}\right);
	\end{eqnarray*}
	which matches the jump behavior of $S(z)$ in a neighborhood of $z=0$, compare RHP \ref{g2RHP}.
	\item Near $z=0$, compare RHP \ref{lagbareRHP},
	\begin{eqnarray*}
		\e^{\im\pi tV\sigma_3}P^{(0)}(z)\e^{-tg(z)\sigma_3}&=&\widehat{P}^{(0)}(z)(\e^{-\im\pi}sz)^{\frac{a}{2}\sigma_3}\begin{pmatrix} 1 & \frac{\im}{2}\frac{1-\gamma}{\sin\pi a}\\ 0 & 1\end{pmatrix},\ \ \ a\notin\mathbb{Z};\\
		\e^{\im\pi tV\sigma_3}P^{(0)}(z)\e^{-tg(z)\sigma_3}&=&\widehat{P}^{(0)}(z)(\e^{-\im\pi}sz)^{\frac{a}{2}\sigma_3}\left[I-\frac{\e^{\im\pi a}}{2\pi\im}\ln\big(\e^{-\im\pi}sz\big)\bigl(\begin{smallmatrix} 0 & 1\\ 0 & 0\end{smallmatrix}\bigr)\right],\ \ \ a\in\mathbb{Z};
	\end{eqnarray*}
	i.e. we also have a matching of the singular behavior with the one stated in RHP \ref{g2RHP}.
	\item As $t\rightarrow+\infty,\gamma\uparrow 1$ subject to \eqref{Bscale}, we derive from \eqref{Lagasy} (using also that $\mathbb{R}_{\geq 0}\ni\chi=k+\alpha$),
	\begin{align}
		P^{(0)}(z)=&\,P^{(\infty)}(z)\bigg\{I+\frac{1}{\z(z)}\big(\mathcal{D}(z)\big)^{-\sigma_3}\begin{pmatrix}-\frac{\gamma_{k-1}}{\gamma_k} & \gamma_k^{-1}\beta^{-2}(z)(\e^{-\im\pi}t^{-1}\z(z))^{-a}\smallskip\\ \gamma_{k-1}\beta^2(z)(\e^{-\im\pi}t^{-1}\z(z))^a & \frac{\gamma_{k-1}}{\gamma_k}\end{pmatrix}\big(\mathcal{D}(z)\big)^{\sigma_3}\nonumber\\
		&\,+\frac{1}{\z^2(z)}\big(\mathcal{D}(z)\big)^{-\sigma_3}\begin{pmatrix}\frac{1}{2}\frac{\gamma_{k-2}}{\gamma_k} & \gamma_{k+1}^{-1}\beta^{-2}(z)(\e^{-\im\pi}t^{-1}\z(z))^{-a}\smallskip\\
		-\gamma_{k-2}\beta^2(z)(\e^{-\im\pi}t^{-1}\z(z))^a & \frac{1}{2}\frac{\gamma_{k-1}}{\gamma_{k+1}}\end{pmatrix}\big(\mathcal{D}(z)\big)^{\sigma_3}\nonumber\\
		&\,+\mathcal{O}\left(t^{-3+2|\alpha|}\right)\bigg\}\label{Lagparaasy}
	\end{align}
	uniformly for $0<r_1\leq|z|\leq r_2<\frac{1}{4}$.
\end{enumerate}
\end{problem}
\begin{remark}\label{cem:4} From properties $(2)$ and $(3)$ in RHP \ref{LagparaRHP} we deduce
\begin{equation*}
	S(z)=N_0(z)P^{(0)}(z),\ \ \ 0\leq|z|<\frac{1}{4}
\end{equation*}
where $N_0(z)$ is analytic at $z=0$.
\end{remark}
We have now completed the necessary local analysis and can move on to the comparison of the explicit model functions $P^{(\infty)}(z),P^{(1)}(z)$ and $P^{(0)}(z)$ to the unknown $S(z)$ of RHP \ref{g2RHP}.
\subsection{Ratio transformation and first small norm estimate} Use \eqref{Besselout}, \eqref{Hankelpara}, \eqref{Lagpara} and define
\begin{equation*}
	R(z)=\begin{pmatrix} 1 & 0\\ \omega & 1\end{pmatrix}S(z)\begin{cases}\big(P^{(0)}(z)\big)^{-1},&|z|<r\\ \big(P^{(1)}(z)\big)^{-1},&|z-1|<r\\ \big(P^{(\infty)}(z)\big)^{-1},&|z|>r,\,|z-1|>r\end{cases};\ \ \ \ \omega=-s^{\frac{1}{2}}\big(N(S_{\infty}-(a+2\alpha)\sigma_3)N^{-1}\big)_{21}
\end{equation*}
with $0<r<\frac{1}{4}$ fixed and we put, compare RHP \ref{g2RHP},
\begin{equation*}
	N=\e^{-\im\pi\alpha\sigma_3}\frac{1}{\sqrt{2}}\begin{pmatrix} 1 & -1\\ 1 & 1\end{pmatrix}\e^{-\im\frac{\pi}{4}\sigma_3}\e^{\im\pi\alpha\sigma_3}.
\end{equation*}
\begin{figure}[tbh]
\begin{center}
\resizebox{0.255\textwidth}{!}{\includegraphics{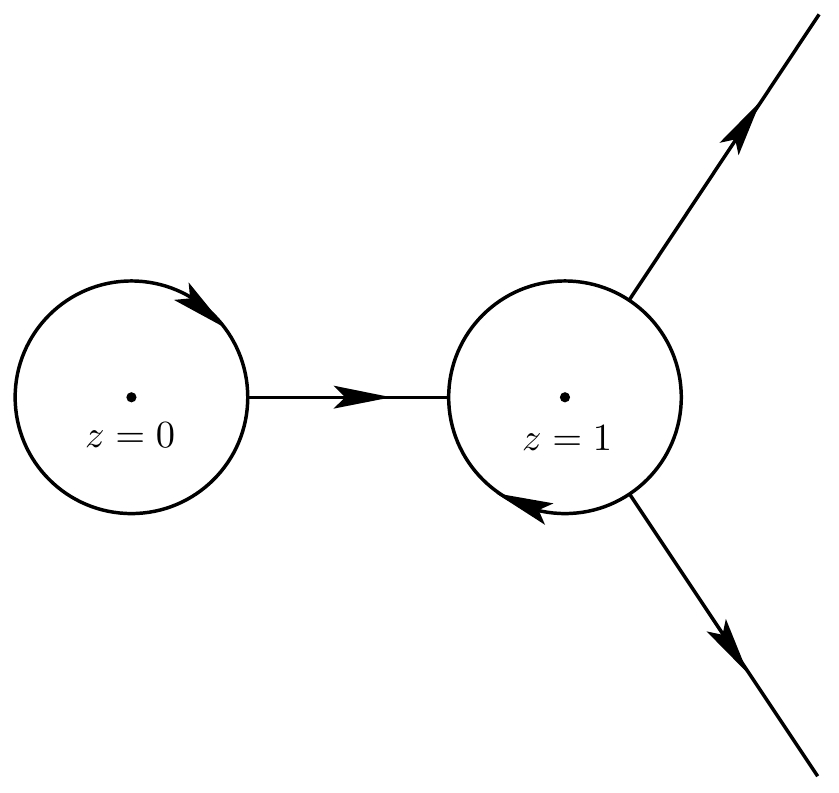}}
\caption{The oriented jump contours for the ratio function $R(z)$ in the complex $z$-plane.}
\label{figure8}
\end{center}
\end{figure}
\begin{problem} Determine $R(z)=R(z;s,\gamma;a)\in\mathbb{C}^{2\times 2}$ such that
\begin{enumerate}
	\item $R(z)$ is analytic for $z\in\mathbb{C}\backslash\Sigma_R$ with square integrable boundary values on the contour
	\begin{equation*}
		\Sigma_R=\partial D(0,r)\cup\partial D(1,r)\cup(r,1-r)\cup\Big(\big(\,\widehat{\Gamma}_{2,T}\cup\widehat{\Gamma}_{4,T}\big)\cap\{z\in\mathbb{C}:\,|z-1|>r\}\Big)
	\end{equation*}
	which is depicted in Figure \ref{figure8}.
	\item We have jumps $R_+(z)=R_-(z)G_R(z),z\in\Sigma_R$ with $G_R(z)=G_R(z;s,\gamma;a)$ and
	\begin{equation*}
		G_R(z)=P^{(0)}(z)\big(P^{(\infty)}(z)\big)^{-1},\ z\in\partial D(0,r);\ \ \ \ 
		G_R(z)=P^{(1)}(z)\big(P^{(\infty)}(z)\big)^{-1},\ z\in\partial D(1,r);
	\end{equation*}
	followed by
	\begin{align*}
		G_R(z)=&\,P^{(\infty)}(z)\begin{pmatrix}1 & 0\\\e^{-\im\pi a}\e^{2tg(z)} & 1\end{pmatrix}\big(P^{(\infty)}(z)\big)^{-1},\ \ z\in\widehat{\Gamma}_{2,T}\cap\{z\in\mathbb{C}:\,|z-1|>r\},\\
		G_R(z)=&\,P^{(\infty)}(z)\begin{pmatrix}1 & 0\\\e^{\im\pi a}\e^{2tg(z)} & 1\end{pmatrix}\big(P^{(\infty)}(z)\big)^{-1},\ \ z\in\widehat{\Gamma}_{4,T}\cap\{z\in\mathbb{C}:\,|z-1|>r\};
	\end{align*}
	as well as
	\begin{equation*}
		G_R(z)=P_-^{(\infty)}(z)\begin{pmatrix}1 & \e^{-\im\pi a}\e^{-t(\varkappa_{\textnormal{Bess}}+2g(z))}\\ 0 & 1\end{pmatrix}\big(P_-^{(\infty)}(z)\big)^{-1},\ \ z\in(r,1-r).
	\end{equation*}
	Note that by construction, see Remarks \ref{cem:3} and \ref{cem:4}, there are no jumps inside $D(0,r)\cup D(1,r)$ and on $(-\infty,-r)\cup(1+r,+\infty)$. Moreover $R(z)$ is bounded at $z=0$ and $z=1$.
	\item As $z\rightarrow\infty$,
	\begin{align*}
		R(z)=&\,\begin{pmatrix}1&0\\ \omega&1\end{pmatrix}(\e^{-\im\pi}sz)^{-\frac{1}{4}\sigma_3}\e^{-\im\pi tV\sigma_3}\frac{1}{\sqrt{2}}\begin{pmatrix}1 & -1\\ 1 & 1\end{pmatrix}\e^{-\im\frac{\pi}{4}\sigma_3}\e^{\im\pi tV\sigma_3}\left\{I+S_{\infty}(\e^{-\im\pi}z)^{-\frac{1}{2}}+\mathcal{O}\left(z^{-1}\right)\right\}\\
		&\,\times\big(P^{(\infty)}(z)\big)^{-1}=I+\mathcal{O}\left((\e^{-\im\pi}z)^{-\frac{1}{2}}\right).
	\end{align*}
\end{enumerate}
\end{problem}
Observe that we obtain from \eqref{crux:4}, \eqref{crux:5} and \eqref{Hmatch:1} the following small norm estimate for
\begin{equation}\label{jfac:1}
	\widehat{G}_R(z;s,\gamma;a)=s^{\frac{1}{4}\sigma_3}G_R(z;s,\gamma;a)s^{-\frac{1}{4}\sigma_3},\ \ z\in\Sigma_R.
\end{equation}
\begin{prop}\label{bDZ:1} Given $\chi\in\mathbb{R}_{\geq 0},a>-1$ there exist positive $t_0=t_0(\chi,a),v_0=v_0(\chi,a)$ and $c=c(\chi,a)$ such that
\begin{equation*}
	\|\widehat{G}_R(\cdot;s,\gamma;a)-I\|_{L^2\cap L^{\infty}(\Sigma_R\backslash\partial D(0,r))}\leq\frac{c}{t},\ \ \ \ \ \forall\,t\geq t_0,\ v=-\ln(1-\gamma)\geq v_0:\ \ v=2t-2\left(\chi+\frac{a}{2}\right)\ln t.
\end{equation*}
\end{prop}
For the remaining contour $\partial D(0,r)$ we have to be more careful:
From \eqref{Lagparaasy}, uniformly for $z\in\partial D(0,r)$,
\begin{align*}
	\widehat{G}_R&(z;s,\gamma;a)-I=(\e^{-\im\pi}(z-1))^{-\frac{1}{4}\sigma_3}\Big\{R_1^+(z)t^{-1+2\alpha}+R_1^-(z)t^{-1-2\alpha}+R_2(z)t^{-1}\\
	&+R_3^+(z)t^{-2+2\alpha}+R_3^-(z)t^{-2-2\alpha}+R_4(z)t^{-2}+\mathcal{O}\left(t^{-3+2|\alpha|}\right)\Big\}(\e^{-\im\pi}(z-1))^{\frac{1}{4}\sigma_3}
\end{align*}
with $t$-independent coefficients
\begin{equation*}
	R_1^+(z)=\im\gamma_k^{-1}\frac{\hat{\beta}^{-2}(z)}{2\hat{\z}(z)}\hat{\eta}^{-1}(z)\e^{-\im\pi\alpha\sigma_3}\begin{pmatrix}1 & -1\\ 1 & -1\end{pmatrix}\e^{\im\pi\alpha\sigma_3};\  R_1^-(z)=-\im\gamma_{k-1}\frac{\hat{\beta}^2(z)}{2\hat{\z}(z)}\hat{\eta}(z)\e^{-\im\pi\alpha\sigma_3}\begin{pmatrix}1 & 1\\ -1 & -1\end{pmatrix}\e^{\im\pi\alpha\sigma_3}
\end{equation*}
and
\begin{equation*}
	R_3^+(z)=\im\gamma_{k+1}^{-1}\frac{\hat{\beta}^{-2}(z)}{2\hat{\z}^2(z)}\hat{\eta}^{-1}(z)\e^{-\im\pi\alpha\sigma_3}\begin{pmatrix}1 & -1\\ 1 & -1\end{pmatrix}\e^{\im\pi\alpha\sigma_3};\ \ R_3^-(z)=\im\gamma_{k-2}\frac{\hat{\beta}^2(z)}{2\hat{\z}^2(z)}\hat{\eta}(z)\e^{-\im\pi\alpha\sigma_3}\begin{pmatrix}1 & 1\\ -1 & -1 \end{pmatrix}\e^{\im\pi\alpha\sigma_3};
\end{equation*}
all of rank one as well as
\begin{equation*}
	R_2(z)=\frac{1}{\hat{\z}(z)}\e^{-\im\pi\alpha\sigma_3}\begin{pmatrix} 0 & -\frac{\gamma_{k-1}}{\gamma_k}\\ -\frac{\gamma_{k-1}}{\gamma_k} & 0\end{pmatrix}\e^{\im\pi\alpha\sigma_3};\ \ R_4(z)=\frac{1}{4\hat{\z}^2(z)}\e^{-\im\pi\alpha\sigma_3}\begin{pmatrix}\frac{\gamma_{k-2}}{\gamma_k}+\frac{\gamma_{k-1}}{\gamma_{k+1}} & \frac{\gamma_{k-2}}{\gamma_k}-\frac{\gamma_{k-1}}{\gamma_{k+1}}\\ \frac{\gamma_{k-2}}{\gamma_k}-\frac{\gamma_{k-1}}{\gamma_{k+1}} & \frac{\gamma_{k-2}}{\gamma_k}+\frac{\gamma_{k-1}}{\gamma_{k+1}}\end{pmatrix}\e^{\im\pi\alpha\sigma_3}.
\end{equation*}
We have introduced $\beta(z)=t^{-\alpha}\hat{\beta}(z),\z(z)=t\hat{\z}(z)$ and 
\begin{equation*}
	\hat{\eta}(z)=\left(\frac{1+(\e^{-\im\pi}(z-1))^{\frac{1}{2}}}{1-(\e^{-\im\pi}(z-1))^{\frac{1}{2}}}\hat{\z}(z)\right)^a=2^{2a}\bigg\{1-\frac{a}{4}z+\mathcal{O}\left(z^2\right)\bigg\},\ \ z\rightarrow 0.
\end{equation*}
Similar to the approach for the Airy kernel we use also here matrix factorizations,
\begin{align*}
	\widehat{G}_R&(z;s,\gamma;a)=\widehat{E}^+(z)(\e^{-\im\pi}(z-1))^{-\frac{1}{4}\sigma_3}\Big\{I+R_1^-(z)t^{-1-2\alpha}+R_2(z)t^{-1}+\big(R_3^+(z)-R_1^+(z)R_2(z)\big)t^{-2+2\alpha}\\
	&+R_3^-(z)t^{-2-2\alpha}+\big(R_4(z)-R_1^+(z)R_1^-(z)\big)t^{-2}+\mathcal{O}\left(t^{-3+2|\alpha|}\right)\Big\}(\e^{-\im\pi}(z-1))^{\frac{1}{4}\sigma_3},\ \ 0\leq\alpha<\frac{1}{2};\\
	\widehat{G}_R&(z;s,\gamma;a)=\widehat{E}^-(z)(\e^{-\im\pi}(z-1))^{-\frac{1}{4}\sigma_3}\Big\{I+R_1^+(z)t^{-1+2\alpha}+R_2(z)t^{-1}+\big(R_3^-(z)-R_1^-(z)R_2(z)\big)t^{-2-2\alpha}\\
	&+R_3^+(z)t^{-2+2\alpha}+\big(R_4(z)-R_1^-(z)R_1^+(z)\big)t^{-2}+\mathcal{O}\left(t^{-3+2|\alpha|}\right)\Big\}(\e^{-\im\pi}(z-1))^{\frac{1}{4}\sigma_3},\ \ -\frac{1}{2}\leq\alpha\leq 0
\end{align*}
with the invertible meromorphic factors
\begin{eqnarray*}
	\widehat{E}^+(z)&=&(\e^{-\im\pi}(z-1))^{-\frac{1}{4}\sigma_3}\big(I+R_1^+(z)t^{-1+2\alpha}\big)(\e^{-\im\pi}(z-1))^{\frac{1}{4}\sigma_3},\\
	\widehat{E}^-(z)&=&(\e^{-\im\pi}(z-1))^{-\frac{1}{4}\sigma_3}\big(I+R_1^-(z)t^{-1-2\alpha}\big)(\e^{-\im\pi}(z-1))^{\frac{1}{4}\sigma_3}.
\end{eqnarray*}
We have thus
\begin{equation*}
	G_R(z;s,\gamma;a)=s^{-\frac{1}{4}\sigma_3}\widehat{E}^{\pm}(z)s^{\frac{1}{4}\sigma_3}E_0^{\pm}(z),\ \ \ z\in\partial D(0,r)
\end{equation*}
and can move on to the next transformation.
\subsection{Singular Riemann-Hilbert problem and iterative solution} Put
\begin{equation}\label{unde:1}
	Q(z)=R(z)\begin{cases}\begin{cases}s^{-\frac{1}{4}\sigma_3}\widehat{E}^+(z)s^{\frac{1}{4}\sigma_3},&|z|<r\\ I,&|z|>r\end{cases},&0\leq\alpha<\frac{1}{2}\smallskip\\
	\begin{cases}s^{-\frac{1}{4}\sigma_3}\widehat{E}^-(z)s^{\frac{1}{4}\sigma_3},&|z|<r\\ I,&|z|>r\end{cases},&-\frac{1}{2}\leq\alpha\leq 0\end{cases}
\end{equation}
and obtain
\begin{problem}\label{singBessel} Determine $Q(z)=Q(z;s,\gamma;a)\in\mathbb{C}^{2\times 2}$ such that
\begin{enumerate}
	\item $Q(z)$ is analytic for $z\in\mathbb{C}\backslash(\Sigma_R\cup\{0\})$ with square integrable boundary values on the contour $\Sigma_R$ depicted in Figure \ref{figure8}.
	\item The jump conditions are as follows,
	\begin{eqnarray*}
		Q_+(z)&=&Q_-(z)G_R(z;s,\gamma;a),\ \ \ \ z\in\Sigma_R\backslash\partial D(0,r);\\
		Q_+(z)&=&Q_-(z)\begin{cases} E_0^+(z),&\,\,\,\,\,0\leq\alpha<\frac{1}{2}\\ E_0^-(z),&-\frac{1}{2}\leq\alpha\leq 0\end{cases},\ \ \ \ z\in\partial D(0,r).
	\end{eqnarray*}
	\item We have a first order pole singularity at $z=0$,
	\begin{equation*}
		Q(z)=\widehat{Q}(z)\begin{pmatrix} 1 & \sigma^{\pm}z^{-1}\\ 0 & 1\end{pmatrix}\big(T^{\pm}\big)^{-1},\ \ \ |z|<r
	\end{equation*}
	with $\widehat{Q}(z)$ analytic at $z=0$ and
	\begin{equation*}
		\sigma^+=\frac{\im\gamma_k^{-1}2^{-4k-2a-1}t^{-1+2\alpha}}{1-\im\gamma_k^{-1}2^{-4k-2a-2}t^{-1+2\alpha}},\ \ \alpha\in\left[0,\frac{1}{2}\right);\ \ \ \sigma^-=\frac{-\im\gamma_{k-1}2^{4k+2a-1}t^{-1-2\alpha}}{1+\im\gamma_{k-1}2^{4k+2a-2}t^{-1-2\alpha}},\ \ \alpha\in\left[-\frac{1}{2},0\right];
	\end{equation*}
	as well as
	\begin{equation*}
		T^+=\begin{pmatrix} 1 & 1\\ \e^{2\pi\im\alpha}s^{\frac{1}{2}} & 0\end{pmatrix},\ \ \alpha\in\left[0,\frac{1}{2}\right);\ \ \ \ \ \ \ T^-=\begin{pmatrix}1 & 1\\ -\e^{2\pi\im\alpha}s^{\frac{1}{2}} & 0\end{pmatrix},\ \ \alpha\in\left[-\frac{1}{2},0\right].
	\end{equation*}
	\item As $z\rightarrow\infty$, 
	\begin{equation*}
		Q(z)=I+\mathcal{O}\left((\e^{-\im\pi}z)^{-\frac{1}{2}}\right).
	\end{equation*}
\end{enumerate}
\end{problem}
In order to complete the nonlinear steepest descent analysis we perform a final transformation, define for $z\in\mathbb{C}\backslash\Sigma_R$
\begin{equation}\label{unde:2}
	Q(z)=\big\{zI+B^{\pm}\big\}L(z)\frac{1}{z}
\end{equation}
and obtain the problem below
\begin{problem}\label{bLfinal} Determine $L(z)=L(z;s,\gamma;a)\in\mathbb{C}^{2\times 2}$ such that
\begin{enumerate}
	\item $L(z)$ is analytic for $z\in\mathbb{C}\backslash\Sigma_R$ with the contour $\Sigma_R$ shown in Figure \ref{figure8}.
	\item We have identical jumps as in the latest $Q$-RHP \ref{singBessel}, i.e.
	\begin{equation*}
		L_+(z)=L_-(z)G_Q(z;s,\gamma;a),\ \ \ z\in\Sigma_R.
	\end{equation*}
	\item $L(z)$ is analytic at $z=0$ provided that
	\begin{equation*}
		B^{\pm}=\sigma^{\pm}L(0)T^{\pm}\begin{pmatrix}0&1\\0&0\end{pmatrix}\big(T^{\pm}\big)^{-1}\left\{L(0)-\sigma^{\pm}L'(0)T^{\pm}\begin{pmatrix}0 & 1\\0&0\end{pmatrix}\big(T^{\pm}\big)^{-1}\right\}^{-1}.
	\end{equation*}
	\item As $z\rightarrow\infty$, we have $L(z)\rightarrow I$.
\end{enumerate}
\end{problem}
At this point we use a similar factorization as in \eqref{jfac:1},
\begin{equation*}
	\widehat{G}_Q(z;s,\gamma;a)=s^{\frac{1}{4}\sigma_3}G_Q(z;s,\gamma;a)s^{-\frac{1}{4}\sigma_3},\ \ \ z\in\Sigma_R
\end{equation*}
and summarize our results.
\begin{prop} Given $\chi=k+\alpha\in\mathbb{R}_{\geq 0}$ with $k\in\mathbb{Z}_{\geq 0},-\frac{1}{2}\leq\alpha<\frac{1}{2}$ as well as $a>-1$, there exist positive $t_0=t_0(\chi,a),v_0=v_0(\chi,a)$ and $c=c(\chi,a)$ such that
\begin{equation*}
	\|\widehat{G}_Q(\cdot;s,\gamma;a)-I\|_{L^2\cap L^{\infty}(\partial D(0,r))}\leq c\,t^{-1-2|\alpha|},\ \ \forall\,t\geq t_0,\ v=-\ln(1-\gamma)\geq v_0:\ \ v=2t-2\left(\chi+\frac{a}{2}\right)\ln t.
\end{equation*}
\end{prop}
Hence, from general theory \cite{DZ}, the latter Proposition together with Proposition \ref{bDZ:1} implies
\begin{prop}\label{bDZ:2} Given $\chi\in\mathbb{R}_{\geq 0}$ and $a>-1$ there exist positive $t_0=t_0(\chi,a),v_0=v_0(\chi,a)$ and $c=c(\chi,a)$ such that RHP \ref{bLfinal} is uniquely solvable for $t\geq t_0,v\geq v_0:v=2t-2\left(\chi+\frac{a}{2}\right)\ln t$. The solution $L=L(z;s,\gamma;a)$ can be computed iteratively via the integral equation
\begin{equation*}
	\widehat{L}(z)=I+\frac{1}{2\pi\im}\int_{\Sigma_R}\widehat{L}_-(w)\big(\widehat{G}_Q(w)-I\big)\frac{\d w}{w-z},\ \ z\in\mathbb{C}\backslash\Sigma_R;\ \ \ \ \ \widehat{L}(z;s,\gamma;a)=s^{\frac{1}{4}\sigma_3}L(z;s,\gamma;a)s^{-\frac{1}{4}\sigma_3}
\end{equation*}
using the estimate
\begin{equation*}
	\|\widehat{L}_-(\cdot;s,\gamma;a)-I\|_{L^2(\Sigma_R)}\leq c\,t^{-1-2|\alpha|},\ \ \ \forall\, t\geq t_0,\ v\geq v_0:\ \ v=2t-2\left(\chi+\frac{a}{2}\right)\ln t.
\end{equation*}
\end{prop}
\section{Extraction of large gap asymptotics}\label{Bessinteg}
We employ the central differential identity \eqref{diff:2}, i.e.
\begin{equation*}
	F(J_{\textnormal{Bess}};\gamma)\equiv\frac{\partial}{\partial s}\ln D(J_{\textnormal{Bess}};\gamma)=-\frac{\gamma}{2\pi\im}\e^{\im\pi a}\big(X^{-1}(z)X'(z)\big)_{21}\Big|_{z\rightarrow s}
\end{equation*}
and the limit is taken with $\textnormal{arg}(z-s)\in(\frac{\pi}{3},\pi)$.
\subsection{Asymptotics for the derivative} Tracing back the transformation sequence
\begin{equation*}
	X(z)\mapsto T(z)\mapsto S(z)\mapsto R(z)\mapsto Q(z)\mapsto L(z)
\end{equation*}
we derive for \eqref{diff:2} the following exact identity
\begin{align*}
	F(J_{\textnormal{Bess}};\gamma)=&\,-\frac{\gamma\e^{\im\pi a}}{2\pi\im s}\big(T^{-1}(w)T'(w)\big)_{21}\Big|_{w\rightarrow 1}=-\frac{\gamma\e^{\im\pi a}}{2\pi\im s}\e^{-2tg_+(1)}\big(S^{-1}(w)S'(w)\big)_{21}\Big|_{w\rightarrow 1}\\
	=&\,-\frac{\gamma\e^{\im\pi a}}{2\pi\im s}\Big\{\big(P^{(1)}(w)\big)^{-1}\big(P^{(1)}(w)\big)'\Big\}_{21}\Big|_{w\rightarrow 1}-\frac{\gamma\e^{\im\pi a}}{2\pi\im s}\Big\{\big(P^{(1)}(w)\big)^{-1}R^{-1}(w)R'(w)P^{(1)}(w)\Big\}_{21}\Big|_{w\rightarrow 1}.
\end{align*}
The first summand is derived from RHP \ref{HankelRHP}
\begin{equation*}
	-\frac{\gamma\e^{\im\pi a}}{2\pi\im s}\Big\{\big(P^{(1)}(w)\big)^{-1}\big(P^{(1)}(w)\big)'\Big\}_{21}\Big|_{w\rightarrow 1}=\gamma\left(-\frac{1}{4}+\frac{a}{2\sqrt{s}}+\frac{k}{\sqrt{s}}-\frac{k}{s}(a+k)\right)
\end{equation*}
and for the second we have from \eqref{unde:1} and \eqref{unde:2},
\begin{align*}
	-\frac{\gamma\e^{\im\pi a}}{2\pi\im s}\Big\{\big(P^{(1)}(w)\big)^{-1}&R^{-1}(w)R'(w)P^{(1)}(w)\Big\}_{21}\Big|_{w\rightarrow 1}=\frac{\gamma\e^{\im\pi a}}{2\pi\im s}\Big\{\big(P^{(1)}(w)\big)^{-1}L^{-1}(w)B^{\pm}L(w)P^{(1)}(w)\Big\}_{21}\Big|_{w\rightarrow 1}\\
	&-\frac{\gamma\e^{\im\pi a}}{2\pi\im s}\Big\{\big(P^{(1)}(w)\big)^{-1}L^{-1}(w)L'(w)P^{(1)}(w)\Big\}_{21}\Big|_{w\rightarrow 1}\\
	&\,=-\frac{\gamma}{2s}\left(1-\frac{2k}{t}\right)\e^{-2\pi\im\alpha}\Big\{L^{-1}(1)B^{\pm}L(1)\Big\}_{21}+\frac{\gamma}{2s}\left(1-\frac{2k}{t}\right)\e^{-2\pi\im\alpha}\Big\{L^{-1}(1)L'(1)\Big\}_{21}.
\end{align*}
Note that with Proposition \ref{bDZ:2}, for $z\in\mathbb{C}\backslash\Sigma_R$,
\begin{equation*}
	\widehat{L}(z)=I+\frac{1}{2\pi\im}\oint_{\partial D(0,r)}\big(\widehat{G}_Q(w)-I\big)\frac{\d w}{w-z}+\frac{1}{2\pi\im}\oint_{\partial D(1,r)}\big(\widehat{G}_Q(w)-I\big)\frac{\d w}{w-z}+\mathcal{O}\left(t^{-2-4|\alpha|}\right)
\end{equation*}
and hence by standard residue computations,
\begin{eqnarray*}
	\widehat{L}(1)&=&I-\e^{-\im\pi\alpha\sigma_3}\Big\{A_0^{\pm}t^{-1-2|\alpha|}+A_0t^{-1}+\mathcal{O}\left(t^{-\min\{2-2|\alpha|,2\}}\right)\Big\}\e^{\im\pi\alpha\sigma_3},\\
	\widehat{L}'(1)&=&\e^{-\im\pi\alpha\sigma_3}\Big\{A_2^{\pm}t^{-1-2|\alpha|}+A_2t^{-1}+\mathcal{O}\left(t^{-\min\{2-2|\alpha|,2\}}\right)\Big\}\e^{\im\pi\alpha\sigma_3}
\end{eqnarray*}
where
\begin{eqnarray*}
	A_0&=&\frac{1}{8}\left(1-\frac{2k}{t}\right)^{-1}\begin{pmatrix} -4a & -4a^2-\frac{2k}{3t-2k}\\ 3 & 4a\end{pmatrix}-\frac{\gamma_{k-1}}{\gamma_k}\begin{pmatrix}0&1\\1&0\end{pmatrix};\\
	 A_2&=&\frac{1}{8}\left(1-\frac{2k}{t}\right)^{-1}\begin{pmatrix}\ast & \ast\\ -4a^2-\frac{6k}{3t-2k} & \ast\end{pmatrix}+\frac{\gamma_{k-1}}{\gamma_k}\begin{pmatrix}0&1\\1&0\end{pmatrix}
\end{eqnarray*}
and
\begin{equation*}
	A_0^{\pm}=2^{\pm4k-1\pm 2a}\begin{pmatrix}1 & \pm 1\\ \mp 1 & -1\end{pmatrix}\begin{cases}-\im\gamma_{k-1},&(+)\\ \im\gamma_k^{-1},&(-)\end{cases};\ \ \ \ \ A_2^{\pm}=-A_0^{\pm}.
\end{equation*}
We now obtain
\begin{align*}
	&\Big\{L^{-1}(1)L'(1)\Big\}_{21}=\e^{2\pi\im\alpha}s^{\frac{1}{2}}\Big\{A_2^{\pm}t^{-1-2|\alpha|}+A_2t^{-1}+\mathcal{O}\left(t^{-\min\{2-2|\alpha|,2\}}\right)\Big\}_{21}\\
	&=\e^{2\pi\im\alpha}s^{\frac{1}{2}}\left[-\frac{a^2}{2t}+\frac{k}{t}(k+a)\pm 2^{\pm 4k-1\pm 2a}t^{-1-2|\alpha|}\begin{cases}-\im\gamma_{k-1},&(+)\\ \im\gamma_k^{-1},&(-)\end{cases}+
	\mathcal{O}\left(t^{-\min\{2-2|\alpha|,2\}}\right)\right\}
\end{align*}
and with
\begin{equation*}
	L^{\pm 1}(0)=\e^{-\im\pi\alpha\sigma_3}s^{-\frac{1}{4}\sigma_3}\left[I+\mathcal{O}\left(t^{-1-2|\alpha|}\right)\right]s^{\frac{1}{4}\sigma_3}\e^{\im\pi\alpha\sigma_3},
\end{equation*}
also
\begin{equation*}
	B^{\pm}=\sigma^{\pm}\e^{-\im\pi\alpha\sigma_3}s^{-\frac{1}{4}\sigma_3}\left\{I+\mathcal{O}\left(t^{-1-2|\alpha|}\right)\right\}\begin{pmatrix}1 & \mp 1\\ \pm 1 & -1 \end{pmatrix}\left\{I+\mathcal{O}\left(t^{-1-2|\alpha|}\right)\right\}s^{\frac{1}{4}\sigma_3}\e^{\im\pi\alpha\sigma_3},
\end{equation*}
so that
\begin{equation*}
	\Big\{L^{-1}(1)B^{\pm}L(1)\Big\}_{21}=\e^{2\pi\im\alpha}s^{\frac{1}{2}}\sigma^{\pm}\left\{\begin{pmatrix}1 & \mp 1\\ \pm 1 & -1\end{pmatrix}+\mathcal{O}\left(t^{-1-2|\alpha|}\right)\right\}_{21}.
\end{equation*}
We summarize (using that $\gamma=1-\e^{-\varkappa_{\textnormal{Bess}}t}=1+\mathcal{O}\left(t^{-\infty}\right)$),
\begin{equation}\label{Bessdifffinal}
	F(J_{\textnormal{Bess}};\gamma)=-\frac{1}{4}+\frac{a}{2\sqrt{s}}-\frac{a^2}{4s}\mp\frac{\sigma^{\pm}}{2\sqrt{s}}+\frac{k}{\sqrt{s}}-\frac{k}{2s}(k+a)+\mathcal{O}\left(t^{-2-2|\alpha|}\right)
\end{equation}
and can now move to the integration.
\subsection{Integration of expansion \eqref{Bessdifffinal}} Start with $k=0$ and $0\leq\alpha<\frac{1}{2}$, i.e. 
\begin{equation}\label{intb:1}
	F(J_{\textnormal{Bess}};\gamma)=-\frac{1}{4}+\frac{a}{2\sqrt{s}}-\frac{a^2}{4s}-\frac{\sigma^+}{2\sqrt{s}}+\mathcal{O}\left(t^{-2-2\alpha}\right),\ \ \ \sigma^+=-\frac{\Gamma(1+a)\pi^{-1}2^{-2a-2}\,t^{-1+2\alpha}}{1+\Gamma(1+a)\pi^{-1}2^{-2a-3}\,t^{-1+2\alpha}}
\end{equation}
and make use of
\begin{lem} Let $t=t(s)=s^{\frac{1}{2}}$ and $0<\hat{s}_0<s$. For any $f\in L^1(\hat{s}_0,s)$, we have
\begin{equation*}
	\int_{\hat{s}_0}^sf\big(t(u)\big)\d u = 2\int_{\hat{t}_0}^twf(w)\,\d w;\ \ \ \ \ \hat{t}_0=\hat{s}_0^{\frac{1}{2}}>0.
\end{equation*}
\end{lem}
Since
\begin{equation*}
	\varkappa_{_\textnormal{Bess}}=\frac{v}{t}=2-2\left(\chi+\frac{a}{2}\right)\frac{\ln t}{t};\ \ \ \ v=-\ln(1-\gamma)>0,\ \ \chi=k+\alpha
\end{equation*}
we have that
\begin{equation*}
	t^{2\alpha}=t^{-2k-a}\e^{-v+2t},
\end{equation*}
and thus in \eqref{intb:1},
\begin{align*}
	-\frac{1}{2}&\int_{\hat{s}_0}^s\sigma^+\big(t(u)\big)\frac{\d u}{u^{\frac{1}{2}}}=\int_{\hat{t}_0}^t\frac{2d_0(a)w^{-1-a}\e^{-v}\e^{2w}}{1+d_0(a)w^{-1-a}\e^{-v}\e^{2w}}\d w\\
	&\,=\ln\Big(1+d_0(a)w^{-1-a}\e^{-v}\e^{2w}\Big)\bigg|_{w=\hat{t}_0}^t+\frac{1}{2}(1+a)\int_{\hat{t}_0}^t\frac{2d_0(a)w^{-1-a}\e^{-v}\e^{2w}}{1+d_0(a)w^{-1-a}\e^{-v}\e^{2w}}\frac{\d w}{w}
\end{align*}
with $d_0(a)=\Gamma(1+a)\pi^{-1}2^{-2a-3}$. Suppose $t\geq t_0,v\geq v_0$ are sufficiently large such that
\begin{equation*}
	0<\frac{1}{2}\left(v+a\ln t\right)\leq t\leq\frac{1}{2}\big(v+(1+a)\ln t\big),
\end{equation*}
and thus the base point of integration equals $\hat{s}_0=\frac{1}{4}(v+a\ln t)^2$ or equivalently $\hat{t}_0=\frac{1}{2}(v+a\ln t)>0$. Then
\begin{equation*}
	-\frac{1}{2}\int_{\hat{s}_0}^s\sigma^+\big(t(u)\big)\frac{\d u}{u^{\frac{1}{2}}}=\ln\Big(1+d_0(a)t^{-1-a}\e^{2t-v}\Big)+\mathcal{O}\left(t^{-1}\right)
\end{equation*}
and we can integrate \eqref{intb:1} as follows,
\begin{eqnarray}
	\ln\left(\frac{\det(I-\gamma K_{\textnormal{Bess}})\big|_{L^2(0,s)}}{\det(I-\gamma K_{\textnormal{Bess}})\big|_{L^2(0,\hat{s}_0)}}\right)&=&-\frac{1}{4}(s-\hat{s}_0)+a\big(\sqrt{s}-\sqrt{\hat{s}_0}\,\big)-\frac{a^2}{4}\ln\left|\frac{s}{\hat{s}_0}\right|+\ln\Big(1+d_0(a)t^{-1-a}\e^{2t-v}\Big)\nonumber\\
	&&+\mathcal{O}\left(t^{-1}\right)\label{intb:2}
\end{eqnarray}
where we use that
\begin{equation*}
	\int_{\hat{s}_0}^s\mathcal{O}\left(t^{-2-2\alpha}(u)\right)\d u=\int_{\hat{t}_0}^t\mathcal{O}\left(w^{-1-2\alpha}\right)\d w=\int_{\hat{t}_0}^t\mathcal{O}\left(w^{-1+a}\e^{-2w}\e^v\right)\d w=\mathcal{O}\left(t^{-1}\right)
\end{equation*}
uniformly as $t\rightarrow+\infty,\gamma\uparrow 1$ such that $0<\frac{1}{2}(v+a\ln t)\leq t\leq\frac{1}{2}(v+(1+a)\ln t)$. But with \eqref{B:1}, we know that
\begin{equation*}
	\ln\det(I-\gamma K_{\textnormal{Bess}})\Big|_{L^2(0,\hat{s}_0)}=-\frac{1}{4}\hat{s}_0+a\sqrt{\hat{s}_0}-\frac{a^2}{4}\ln \hat{s}_0+\ln\tau_a+\mathcal{O}\left(\hat{t}_0^{-1}\right),\ \ \ \ \tau_a=\frac{G(1+a)}{(2\pi)^{\frac{a}{2}}}
\end{equation*}
since $\varkappa_{_\textnormal{Bess}}=\frac{v}{t(\hat{s}_0)}\geq 2-a\frac{\ln t}{t}$. Substituting the latter back into \eqref{intb:2} we have thus
\begin{prop}\label{propb:1} For any given $a>-1$ there exist $t_0(a)>0$ and $v_0(a)>0$ such that
\begin{equation*}
	\ln\det(I-\gamma K_{\textnormal{Bess}})\Big|_{L^2(0,s)}=-\frac{s}{4}+a\sqrt{s}-\frac{a^2}{4}\ln s+\ln\tau_a+\ln\Big(1+d_0(a)t^{-1-a}\e^{2t-v}\Big)+\mathcal{O}\left(t^{-1}\right)
\end{equation*}
uniformly for
\begin{equation*}
	t\geq t_0,\ v\geq v_0:\ \ \ \ \ 0<\frac{1}{2}(v+a\ln t)\leq t\leq\frac{1}{2}\big(v+(1+a)\ln t\big)
\end{equation*}
with $d_0(a)=\Gamma(1+a)\pi^{-1}2^{-2a-3}$.
\end{prop}
In the general case $k\in\mathbb{Z}_{\geq 0}$ we let
\begin{equation*}
	\hat{t}_k=\frac{1}{2}\left(v+2\left(k-\frac{1}{2}+\frac{a}{2}\right)\ln t\right),\ \ \ \ \ \hat{t}_k'=\frac{1}{2}\left(v+2\left(k+\frac{a}{2}\right)\ln t\right),\ \ k\in\mathbb{Z}_{\geq 1}
\end{equation*}
and first derive the following two expansions.
\begin{lem}\label{lemb:1} For $\hat{t}_k\leq t\leq \hat{t}_k'$, i.e. $\alpha\in[-\frac{1}{2},0]$,
\begin{align*}
	&\ln\left(\frac{\det(I-\gamma K_{\textnormal{Bess}})\big|_{L^2(0,s)}}{\det(I-\gamma K_{\textnormal{Bess}})\big|_{L^2(0,\hat{s}_k)}}\right)=-\frac{1}{4}(s-\hat{s}_k)+a\big(\sqrt{s}-\sqrt{\hat{s}_k}\,\big)-\frac{a^2}{4}\ln\left|\frac{s}{\hat{s}_k}\right|+2k\big(\sqrt{s}-\sqrt{\hat{s}_k}\,\big)\\
	&-\frac{k}{2}(k+a)\ln\left|\frac{s}{\hat{s}_k}\right|+\ln\Big(1+\im\gamma_{k-1}2^{4k+2a-2}t^{-1+a+2k}\e^{-2t+v}\Big)-\ln\Big(1+\im\gamma_{k-1}2^{4k+2a-2}\Big)+\mathcal{O}\left(t^{-1}\right).
\end{align*}
\end{lem}
\begin{proof}
We use that
\begin{equation*}
	\frac{1}{2}\int_{\hat{s}_k}^s\sigma^-\big(t(u)\big)\frac{\d u}{u^{\frac{1}{2}}} =\ln\Big(1+\im\gamma_{k-1}2^{4k+2a-2}t^{-1+a+2k}\e^{-2t+v}\Big)-\ln\Big(1+\im\gamma_{k-1}2^{4k+2a-2}\Big)+\mathcal{O}\left(t^{-1}\right)
\end{equation*}
and the stated identity follows now from \eqref{Bessdifffinal}.
\end{proof}
\begin{lem}\label{lemb:2} For $\hat{t}_k'\leq t\leq \hat{t}_{k+1}$, i.e. $\alpha\in[0,\frac{1}{2}]$,
\begin{align*}
	&\ln\left(\frac{\det(I-\gamma K_{\textnormal{Bess}})\big|_{L^2(0,s)}}{\det(I-\gamma K_{\textnormal{Bess}})\big|_{L^2(0,\hat{s}_k')}}\right)=-\frac{1}{4}(s-\hat{s}_k')+a\left(\sqrt{s}-\sqrt{\hat{s}_k'}\,\right)-\frac{a^2}{4}\ln\left|\frac{s}{\hat{s}_k'}\right|+2k\left(\sqrt{s}-\sqrt{\hat{s}_k'}\,\right)\\
	&-\frac{k}{2}(k+a)\ln\left|\frac{s}{\hat{s}_k'}\right|+\ln\Big(1-\im\gamma_k^{-1}2^{-4k-2a-2}t^{-1-a-2k}\e^{2t-v}\Big)+\mathcal{O}\left(t^{-1}\right)
\end{align*}
\end{lem}
\begin{proof} Note that
\begin{equation*}
	-\frac{1}{2}\int_{\hat{s}_k'}^s\sigma^+\big(t(u)\big)\frac{\d u}{u^{\frac{1}{2}}}=\ln\Big(1-\im\gamma_k^{-1}2^{-4k-2a-2}t^{-1-a-2k}\e^{2t-v}\Big)+\mathcal{O}\left(t^{-1}\right),
\end{equation*}
and substitute this back into \eqref{Bessdifffinal}.
\end{proof}
At this point we follow the logic carried out for the Airy kernel determinant, i.e. first choose $k=1$ in Lemma \ref{lemb:1} and use Proposition \ref{propb:1}, 
\begin{equation*}
	\ln\det(I-\gamma K_{\textnormal{Ai}})\Big|_{L^2(0,\hat{s}_1)}=-\frac{\hat{s}_1}{4}+a\sqrt{\hat{s}_1}-\frac{a^2}{4}\ln \hat{s}_1+\ln\tau_a+\ln\Big(1+d_0(a)\Big)+\mathcal{O}\left(t^{-1}\right),
\end{equation*}
so that
\begin{equation*}
	\ln\det(I-\gamma K_{\textnormal{Bess}})\Big|_{L^2(0,s)}=-\frac{s}{4}+a\sqrt{s}-\frac{a^2}{4}\ln s+\ln\tau_a+\ln\Big(1+d_0(a)t^{-1-a}\e^{2t-v}\Big)+\mathcal{O}\left(\frac{\ln t}{t}\right)
\end{equation*}	
%
and we used that
\begin{equation*}
	2\big(t-\hat{t}_1\big)=-\ln\left(t^{1+a}\e^{-2t+v}\right).
\end{equation*}
Summarizing,
\begin{prop}\label{propb:2} Given $a>-1$ there exist $t_0(a)>0$ and $v_0(a)>0$ such that
\begin{equation*}
	\ln\det(I-\gamma K_{\textnormal{Bess}})\Big|_{L^2(0,s)}=-\frac{s}{4}+a\sqrt{s}-\frac{a^2}{4}\ln s+\ln\tau_a+\ln\Big(1+d_0(a)t^{-1-a}\e^{2t-v}\Big)+\mathcal{O}\left(\frac{\ln t}{t}\right)
\end{equation*}
uniformly for
\begin{equation*}
	t\geq t_0,\ v\geq v_0:\ \ \ \ \frac{1}{2}\left(v+(1+a)\ln t\right)\leq t\leq \frac{1}{2}\left(v+(2+a)\ln t\right).
\end{equation*}
\end{prop}
Now we continue with $k=1$: use Proposition \ref{propb:2},
\begin{equation*}
	\ln\det(I-\gamma K_{\textnormal{Bess}})\Big|_{L^2(0,\hat{s}_1')}=-\frac{\hat{s}_1'}{4}+a\sqrt{\hat{s}_1'}-\frac{a^2}{4}\ln \hat{s}_1'+\ln\tau_a+\ln\Big(1+d_0(a)t\Big)+\mathcal{O}\left(\frac{\ln t}{t}\right),
\end{equation*}
let $d_1(a)=\Gamma(2+a)\pi^{-1}2^{-7-2a}$, and thus back in Lemma \ref{lemb:2},
\begin{align*}
	\ln\det(I-&\gamma K_{\textnormal{Bess}})\Big|_{L^2(0,s)}=-\frac{s}{4}+a\sqrt{s}-\frac{a^2}{4}\ln s+\ln\tau_a+2(t-\hat{t}_1')+\ln\Big(1+d_0(a)t\Big)\\
	&+\ln\Big(1+d_1(a)t^{-3-a}\e^{2t-v}\Big)+\mathcal{O}\left(\frac{\ln t}{t}\right). 
\end{align*}
But since
\begin{equation*}
	2(t-\hat{t}_1')=\ln\Big(t^{-2-a}\e^{2t-v}\Big),
\end{equation*}
we obtain in fact
\begin{prop} Given $a>-1$, there exist $t_0(a)>0$ and $v_0(a)>0$ such that
\begin{equation*}
	\ln\det(I-\gamma K_{\textnormal{Bess}})\Big|_{L^2(0,s)}=-\frac{s}{4}+a\sqrt{s}-\frac{a^2}{4}\ln s+\ln\tau_a+\prod_{j=0}^1\ln\Big(1+d_j(a)t^{-2j-1-a}\e^{2t-v}\Big)+\mathcal{O}\left(\frac{\ln t}{t}\right)
\end{equation*}
uniformly for
\begin{equation*}
	t\geq t_0,\ \ v\geq v_0:\ \ \ \frac{1}{2}\left(v+(2+a)\ln t\right)\leq t\leq \frac{1}{2}\left(v+(3+a)\ln t\right).
\end{equation*}
\end{prop}
\begin{remark} Also here the lower constraint on $t$ is artificial since for $t<\frac{1}{2}(v+2\ln t)$, the second factor of the product only contributes to the error term and we reproduce Proposition \ref{propb:2} with adjusted error estimate.
\end{remark}
At this point we would iterate the procedure (just as in the case of the Airy kernel) and obtain in the end the following result
\begin{theo}\label{nice:b} Given $q\in\mathbb{Z}_{\geq 1}$ and $a>-1$, there exist $t_0=t_0(q,a)>0$ and $v_0=v_0(q,a)>0$ such that
\begin{equation*}
	\ln\det(I-\gamma K_{\textnormal{Bess}})\Big|_{L^2(0,s)}=-\frac{s}{4}+a\sqrt{s}-\frac{a^2}{4}\ln s+\ln\tau_a+\prod_{j=0}^{q-1}\Big(1+d_j(a)t^{-2j-1-a}\e^{2t-v}\Big)+\mathcal{O}\left(\frac{\ln t}{t}\right)\\
\end{equation*}
with $d_j(a)=j!\,\Gamma(1+a+j)\pi^{-1}2^{-4j-2a-3}$, uniformly for
\begin{equation*}
	t\geq t_0,\ \ v\geq v_0:\ \ \ 0<\frac{1}{2}(v+a\ln t)\leq t\leq\frac{1}{2}\left(v+(2q+a)\ln t\right).
\end{equation*}
\end{theo}
\section{Proof of Theorem \ref{Besselmain} and asymptotics for eigenvalues}\label{Besseigexp} To obtain Theorem \ref{Besselmain} we combine Theorem \ref{nice:b} with expansion \eqref{B:1}, 
\begin{cor} Given $\chi\in\mathbb{R}$ and fixed $a>-1$ determine $p\in\mathbb{Z}_{\geq 0}$ such that $p=0$ for $\chi<-\frac{1}{2}$ and $\chi+\frac{1}{2}<p\leq\chi+\frac{3}{2}$ for $\chi\geq-\frac{1}{2}$. There exist positive $t_0=t_0(\chi,a),v_0=v_0(\chi,a)$ such that
\begin{equation*}
	D(J_{\textnormal{Bess}};\gamma)=\exp\left[-\frac{s}{4}+a\sqrt{s}\right]s^{-\frac{1}{4}a^2}\tau_a\prod_{j=0}^{p-1}\Big(1+d_j(a)t^{-2j-1-a}\e^{2t-v}\Big)
	\left(1+\mathcal{O}\left(\max\left\{t^{-2(p-\frac{1}{2}-\chi)},\frac{\ln t}{t}\right\}\right)\right)
\end{equation*}
with $d_j(a)=j!\,\Gamma(1+a+j)\pi^{-1}2^{-4j-2a-3}$, uniformly for $t\geq t_0,v\geq v_0$ and $\varkappa_{_\textnormal{Bess}}=\frac{v}{t}\geq 2-2(\chi+\frac{a}{2})\frac{\ln t}{t}$. In case $p=0$, we take $\prod_{j=0}^{p-1}(\ldots)\equiv 1$.
\end{cor}
At this point the logic is just as before we only have to use that $K_{\textnormal{Bess}}^{(a)}$ is positive definite with finite operator norm $\|K_{\textnormal{Bess}}^{(a)}\|<1$ and trace norm 
\begin{equation*}
	\|K_{\textnormal{Bess}}^{(a)}\|_1=\int_0^sK_{\textnormal{Bess}}^{(a)}(\lambda,\lambda)\d\lambda\leq ct,\ \ \forall\, t\geq t_0,\ \ t=s^{\frac{1}{2}}.
\end{equation*}
\begin{appendix}
\section{Refining Corollary $1.2$ in \cite{B}}\label{Airyapp} In this appendix we will improve the error estimate derived in \cite{B}, Corollary $1.2.$. In more detail, we shall prove the following Theorem.
\begin{theo}\label{bettererror} As $s\rightarrow-\infty,\gamma\uparrow 1$,
\begin{equation}\label{A:1}
	\det(I-\gamma K_{\textnormal{Ai}})\Big|_{L^2(s,\infty)}=\exp\left[\frac{s^3}{12}\right]|s|^{-\frac{1}{8}}c_0\Big(1+\mathcal{O}\left(s^{-\frac{3}{4}}\right)\Big),\ \ \ c_0=\exp\left[\frac{1}{24}\ln 2+\z'(-1)\right],
\end{equation}
uniformly for
\begin{equation*}
	\varkappa_{_\textnormal{Ai}}=-\frac{\ln(1-\gamma)}{(-s)^{\frac{3}{2}}}\geq \frac{2}{3}\sqrt{2}.
\end{equation*}
\end{theo}
Expansion \eqref{A:1} was proven in \cite{B} with a weaker error term. The result in loc. cit followed from the Tracy-Widom representation of $D(J_{\textnormal{Ai}};\gamma)$ in terms of a Painlev\'e transcendent and the derivation of transition asymptotics for the latter.\smallskip

Here, we will apply the method of integrable integral operators and use RHP \ref{masterIIKS}. In more detail, \eqref{A:1} will follow from a comparison of the nonlinear steepest descent analysis of RHP \ref{masterIIKS} for $\gamma=1, s\rightarrow-\infty$ (corresponding to $\varkappa_{_\textnormal{Ai}}=+\infty$), cf. \cite{CIK}, to the nonlinear steepest descent analysis of the same RHP for $\gamma\uparrow 1,s\rightarrow-\infty$ such that $\varkappa_{_\textnormal{Ai}}\geq\frac{2}{3}\sqrt{2}$ is fixed. We will see that the difference between $\varkappa_{_\textnormal{Ai}}\geq\frac{2}{3}\sqrt{2}$ fixed and $\varkappa_{_\textnormal{Ai}}=+\infty$ is ``almost beyond all orders".

\subsection{Comparison of nonlinear steepest descent analysis} We go back to RHP \ref{masterIIKS} tailored to \eqref{Airytailor} and observe its validity for any choice of $\gamma\leq 1$. The same also applies to the first two transformations
\begin{equation*}
	Y(z)\mapsto X(z)\mapsto T(z)
\end{equation*}
carried out in Section \ref{prelim}. After that, we can formally copy the $g$-function transformation \eqref{g:1}, however for the scale
\begin{equation*}
	\varkappa_{_\textnormal{Ai}}\in\left[\frac{2}{3}\sqrt{2},+\infty\right)\ \ \textnormal{fixed},\ \ \ \ \ \textnormal{or}\ \ \ \ \varkappa_{_\textnormal{Ai}}=+\infty
\end{equation*}
we shall see that setting $V\equiv 0$ in \eqref{g:1} is sufficient. With this, the next transformation
\begin{equation*}
	S_0(z)=T(z)\e^{t(g(z)|_{V=0})\sigma_3},\ \ \ z\in\mathbb{C}\backslash\Sigma_T
\end{equation*}
leads us to the following problem.
\begin{problem} The normalized function $S_0(z)=S_0(z;s,\gamma)\in\mathbb{C}^{2\times 2}$ is characterized by the following properties
\begin{enumerate}
	\item $S_0(z)$ is analytic for $z\in\mathbb{C}\backslash(\Sigma_T\cup\{0\})$.
	\item The limiting values $S_{\pm}(z),z\in\Sigma_T$ are related by the equations
	\begin{equation*}
		S_{0+}(z)=S_{0-}(z)\begin{pmatrix} 1 & 0\\ \e^{2tg(z)|_{V=0}} & 1 \end{pmatrix},\ \ z\in\big(\Gamma_2\cup\Gamma_4\big)\backslash\{0\};\ \ \ \ \ S_{0+}(z)=S_{0-}(z)\begin{pmatrix} 0 & 1\\ -1 & 0 \end{pmatrix},\ \ z\in\Gamma_3\backslash\{0\};
	\end{equation*}
	and, along the entire positive half-ray,
	\begin{equation*}
		S_{0+}(z)=S_{0-}(z)\begin{pmatrix} 1 & \e^{-t(\varkappa_{\textnormal{Ai}}+2g(z)|_{V=0})}\\ 0 & 1 \end{pmatrix},\ \ z\in(0,+\infty).
	\end{equation*}
	\item The singular behavior is identical to the one stated in RHP \ref{gRHP} subject to the replacement $V=0$.
	\item Also the asymptotic normalization at $z=\infty\notin\Sigma_T$ is the same as in RHP \ref{gRHP}, modulo $V=0$.
\end{enumerate}
\end{problem}
Note that, as before, with $0<r<\frac{1}{8}$ fixed,
\begin{equation}\label{A:es1}
	\Re\big(g(z)|_{V=0}\big)<0,\ \ \ z\in\big(\Gamma_2\cup\Gamma_4\big)\backslash D(0,r)
\end{equation}
and, since
\begin{equation*}
	g(z)|_{V=0}\geq-\frac{1}{3}\sqrt{2},\ \ \ z\in(0,+\infty);\ \ \ \ \ \ g\left(\frac{1}{2}\right)|_{V=0}=-\frac{1}{3}\sqrt{2},
\end{equation*}
we have now for $z\in(0,+\infty)\backslash D(0,r)$,
\begin{equation}\label{A:es2}
	\varkappa_{_\textnormal{Ai}}+2g(z)|_{V=0}\geq 0\ \ \ \textnormal{with equality iff}\ \ \ \varkappa_{_\textnormal{Ai}}=\frac{2}{3}\sqrt{2}\ \ \textnormal{and}\ \ z=\frac{1}{2}
\end{equation}
as a consequence of our choice $\varkappa_{_\textnormal{Ai}}\geq\frac{2}{3}\sqrt{2}$. But this means that we only require a simplified local Riemann-Hilbert analysis near $z=\frac{1}{2}$ when working with the scale $\varkappa_{_\textnormal{Ai}}\in[\frac{2}{3}\sqrt{2},+\infty]$. In fact, for $\varkappa_{_\textnormal{Ai}}\in[\frac{2}{3}\sqrt{2},+\infty)$ we take $k=0$ in \eqref{hermitebare}, i.e.
\begin{equation*}
	H(\z)\big|_{k=0}=\begin{pmatrix} 1 & \frac{1}{2\pi\im}\int_{\mathbb{R}}\e^{-t^2}\frac{\d t}{t-\z}\\ 0 & 1 \end{pmatrix}\e^{-\frac{1}{2}\z^2\sigma_3},\ \ \ \z\in\mathbb{C}\backslash\mathbb{R},
\end{equation*}
and adjusting all subsequent steps of the construction to $V=0=\alpha=k$ is sufficient. The matching \eqref{e:11} is then replaced by
\begin{equation}\label{nmatcrit}
	P^{(\frac{1}{2})}(z)=P^{(\infty)}(z)\big|_{V=0}\left\{I+\frac{\gamma_0^{-1}}{\z(z)}\bigl(\begin{smallmatrix} 0 & 1\\ 0 & 0 \end{smallmatrix}\bigr)+\mathcal{O}\left(t^{-\frac{3}{2}}\right)\right\};\ \ \ \ \z(z)=\sqrt{2t}\left(g(z)\big|_{V=0}+\frac{\sqrt{2}}{3}\right)^{\frac{1}{2}},
\end{equation}
uniformly for $0<r_1\leq|z-\frac{1}{2}|\leq r_2<\frac{1}{2}$. On the other hand, for $\varkappa_{_\textnormal{Ai}}=+\infty$ we do not need any separate analysis near $z=\frac{1}{2}$.
\begin{remark} If we were to keep $\varkappa_{_\textnormal{Ai}}>\frac{2}{3}\sqrt{2}$ fixed, there would be no need for a local analysis near $z=\frac{1}{2}$ at all, see \eqref{A:es2}. The difference between $\varkappa_{_\textnormal{Ai}}>\frac{2}{3}\sqrt{2}$ and $\varkappa_{_\textnormal{Ai}}=+\infty$ is thus ``beyond all orders" and between $\varkappa_{_\textnormal{Ai}}\geq\frac{2}{3}\sqrt{2}$ and $\varkappa_{_\textnormal{Ai}}=+\infty$ ``almost beyond all orders".
\end{remark}
The local model functions in \eqref{out} and \eqref{e:7} can be used again, subject to the replacements $V=0=\alpha=k$.
\subsection{Small norm theorem} In the given situation the ratio transformation consists of the replacement
\begin{equation*}
	R_0(z)=\begin{pmatrix} 1 & 0\\ \omega & 1 \end{pmatrix} S_0(z)\begin{cases}\big(P^{(0)}(z)\big)^{-1}\big|_{V=0=\alpha=k},&|z|<r\smallskip\\ \big(P^{(\frac{1}{2})}(z)\big)^{-1}\big|_{V=0=\alpha=k},&|z-\frac{1}{2}|<r\smallskip\\ \big(P^{(\infty)}(z)\big)^{-1}\big|_{V=0=\alpha=k},&|z|>r, |z-\frac{1}{2}|>r \end{cases};\ \ \ \omega=-|s|^{\frac{1}{2}}\big(NS_{\infty}N^{-1}\big)_{21}
\end{equation*}
when $\varkappa_{_\textnormal{Ai}}\in[\frac{2}{3}\sqrt{2},+\infty)$ is fixed and without $P^{(\frac{1}{2})}(z)\big|_{V=0=\alpha=k}$ for $\varkappa_{_\textnormal{Ai}}=+\infty$. In either case $0<r<\frac{1}{8}$ and we use
\begin{equation*}
	N=\frac{1}{\sqrt{2}}\begin{pmatrix} 1 & 1\\ -1 & 1 \end{pmatrix} \e^{-\im\frac{\pi}{4}\sigma_3}.
\end{equation*}
We are lead to the problem below which is formulated on the contour shown in Figure \ref{figure4} for $\varkappa_{_\textnormal{Ai}}\in[\frac{2}{3}\sqrt{2},+\infty)$ and in Figure \ref{figure5} for $\varkappa_{_\textnormal{Ai}}=+\infty$.
\begin{figure}[tbh]
\begin{center}
\resizebox{0.21\textwidth}{!}{\includegraphics{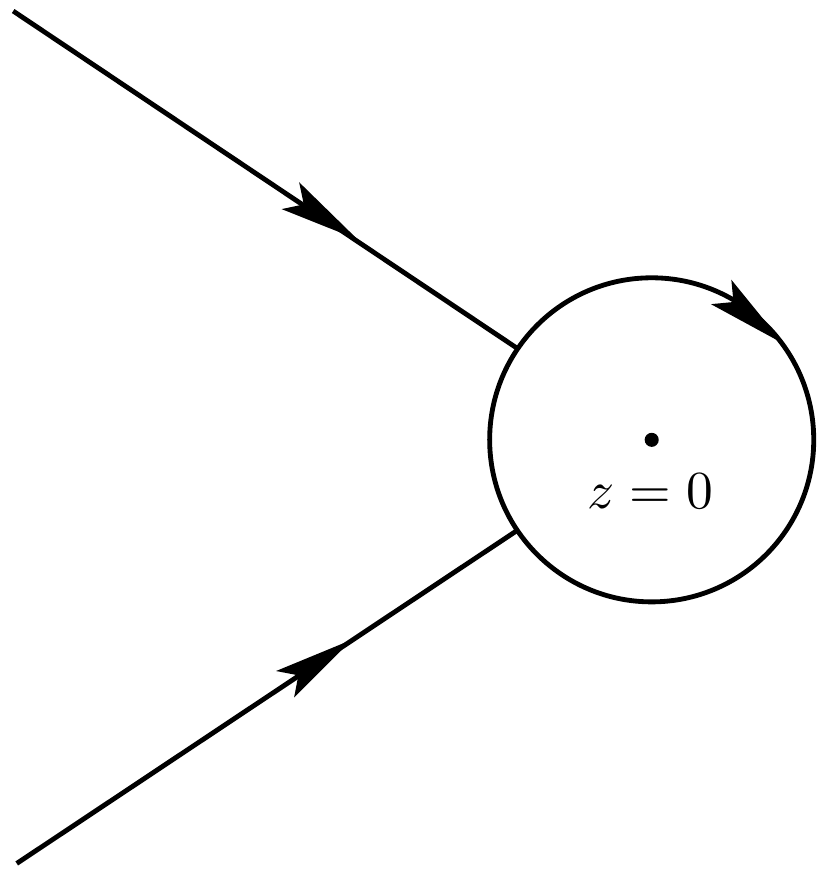}}
\caption{The oriented jump contours for the ratio function $R_0(z)$ in the complex $z$-plane, case $\varkappa_{_\textnormal{Ai}}=+\infty$.}
\label{figure5}
\end{center}
\end{figure}
\begin{problem} Determine $R_0(z)=R_0(z;t,\varkappa_{_\textnormal{Ai}})\in\mathbb{C}^{2\times 2}$ such that
\begin{enumerate}
	\item $R_0(z)$ is analytic for $z\in\mathbb{C}\backslash\Sigma_{R_0}$ with square integrable boundary values on the contours $\Sigma_{R_0}$,
	\begin{align*}
		\Sigma_{R_0}=&\,\,\partial D(0,r)\cup\big((\Gamma_2\cup\Gamma_4)\cap\{z\in\mathbb{C}:\,|z|>r\}\big);\\
		\Sigma_{R_0}=&\,\,\partial D(0,r)\cup\partial D\left(\frac{1}{2},r\right)\cup\left(r,\frac{1}{2}-r\right)\cup\left(\frac{1}{2}+r,\infty\right)\cup\big((\Gamma_2\cup\Gamma_4)\cap\{z\in\mathbb{C}:\,|z|>r\}\big);
	\end{align*}
	shown in Figures \ref{figure4} and \ref{figure5}.
	\item On the contour $\Sigma_{R_0}$ we have jumps $R_{0+}(z)=R_{0-}(z)G_{R_0}(z;s,\gamma),z\in\Sigma_{R}$ with
	\begin{align*}
		G_{R_0}(z)=&\,P^{(0)}(z)\big|_{V=0}\big(P^{(\infty)}(z)\big|_{V=0}\big)^{-1},\ z\in\partial D(0,r);\\
		G_{R_0}(z)=&\,P^{(\frac{1}{2})}(z)\big|_{V=0}\big(P^{(\infty)}(z)\big|_{V=0}\big)^{-1},\ z\in\partial D\left(\frac{1}{2},r\right)\\
		G_{R_0}(z)=&\,P^{(\infty)}(z)\big|_{V=0}\begin{pmatrix} 1 & \e^{-t(\varkappa_{\textnormal{Ai}}+2g(z)|_{V=0})}\\ 0 & 1 \end{pmatrix}\big(P^{(\infty)}(z)\big|_{V=0}\big)^{-1},\ \ z\in\left(r,\frac{1}{2}-r\right)\cup\left(r+\frac{1}{2},+\infty\right);\\
		G_{R_0}(z)=&\,P^{(\infty)}(z)\big|_{V=0}\begin{pmatrix} 1 & 0\\ \e^{tg(z)|_{V=0}} & 1 \end{pmatrix}\big(P^{(\infty)}(z)\big|_{V=0}\big)^{-1},\ \ z\in\big(\Gamma_2\cup\Gamma_4\big)\cap\{z\in\mathbb{C}:\,|z|>r\}.
	\end{align*}
	These jumps hold for $\varkappa_{_\textnormal{Ai}}\in[\frac{2}{3}\sqrt{2},+\infty)$.  For $\varkappa_{_\textnormal{Ai}}=+\infty$ we simply have to dispose of $\partial D(\frac{1}{2},r)\cup(r,\frac{1}{2}-r)\cup(r+\frac{1}{2},+\infty)$.
	\item As $z\rightarrow\infty$, we have $R_0(z)\rightarrow I$.
\end{enumerate}
\end{problem}
It is evident that the latter problem admits direct asymptotical analysis, in fact from Proposition \ref{DZ:1} and \eqref{nmatcrit} we obtain at once
\begin{prop} There exist $t_0$ and $c$ positive such that
\begin{equation*}
	\|G_{R_0}(\cdot;t,\varkappa_{_\textnormal{Ai}})-I\|_{L^2\cap L^{\infty}(\Sigma_{R_0})}\leq\frac{c}{t^{\frac{1}{6}}},\ \ \forall\,t\geq t_0,\ \varkappa_{_\textnormal{Ai}}\geq\frac{2}{3}\sqrt{2};\ \ 
	\|G_{R_0}(\cdot;t,+\infty)-I\|_{L^2\cap L^{\infty}(\Sigma_{R_0})}\leq\frac{c}{t^{\frac{2}{3}}},\ \ \forall\,t\geq t_0.
\end{equation*}
\end{prop}
This estimate implies unique solvability of the $R_0$-RHP in $L^2(\Sigma_{R_0})$. Moreover, the solution satisfies
\begin{eqnarray}
	R_0(z;t,\varkappa_{_\textnormal{Ai}})&=&|s|^{-\frac{1}{4}\sigma_3}\left\{I+\mathcal{O}\left(\frac{t^{-\frac{1}{2}}}{1+|z|}\right)\right\}|s|^{\frac{1}{4}\sigma_3},\ \ t\rightarrow\infty,\gamma\uparrow 1:\ \ \varkappa_{_\textnormal{Ai}}\geq\frac{2}{3}\sqrt{2}\label{A:es31}\\
	R_0(z;t,+\infty)&=&|s|^{-\frac{1}{4}\sigma_3}\left\{I+\mathcal{O}\left(\frac{t^{-1}}{1+|z|}\right)\right\}|s|^{\frac{1}{4}\sigma_3},\ \ t\rightarrow\infty\label{A:es32}
\end{eqnarray}
uniformly for $z\in\mathbb{C}\backslash\Sigma_{R_0}$.  We shall now consider
\begin{equation*}
	\chi(z;t,\varkappa_{_\textnormal{Ai}})=R_0(z;t,\varkappa_{_\textnormal{Ai}})\big(R_0(z;t,+\infty)\big)^{-1},\ \ z\in\mathbb{C}\backslash\Sigma_{R};\ \ \ t\geq t_0,\ \ \ \varkappa_{_\textnormal{Ai}}\in\left[\frac{2}{3}\sqrt{2},+\infty\right)\ \ \textnormal{fixed}.
\end{equation*}
Because of \eqref{A:es31} and \eqref{A:es32} the function $\chi(z;t,\varkappa_{_\textnormal{Ai}})$ is well defined and it again satisfies a RHP
\begin{problem} Determine $\chi(z)=\chi(z;t,\varkappa_{_\textnormal{Ai}})\in\mathbb{C}^{2\times 2}$ such that
\begin{enumerate}
	\item $\chi(z)$ is analytic for $z\in\mathbb{C}\backslash\Sigma_{\chi}$ with square integrable boundary values on the jump contour 
	\begin{equation*}
		\Sigma_{\chi}=\partial D(0,r)\cup\partial D\left(\frac{1}{2},r\right)\cup\left(r,\frac{1}{2}-r\right)\cup\left(r+\frac{1}{2},+\infty\right)
	\end{equation*}
	shown in Figure \ref{figure6}.
	\begin{figure}[tbh]
\begin{center}
\resizebox{0.3\textwidth}{!}{\includegraphics{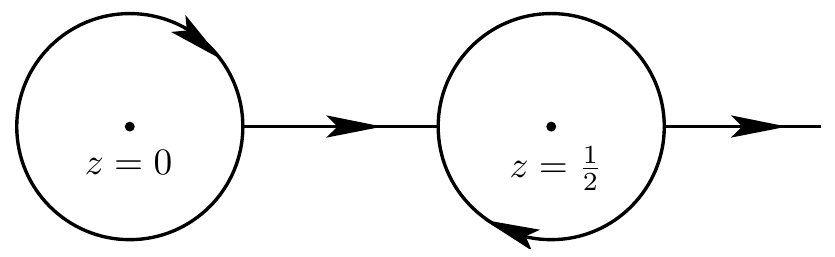}}
\caption{The oriented jump contours for the ratio function $\chi(z)$ in the complex $z$-plane.}
\label{figure6}
\end{center}
\end{figure}
	\item The limiting values are related via the jump conditions $\chi_+(z)=\chi_-(z)G_{\chi}(z;t,\varkappa)$ with
	\begin{align*}
		G_{\chi}(z)=&\,R_{0-}(z;t,+\infty)\Big[P^{(0)}(z)\big|_{V=0}\big(P^{(0)}(z)\big|_{\substack{V=0\\ \varkappa=+\infty}}\big)^{-1}\Big]R_{0-}^{-1}(z;t,+\infty),\ \ z\in\partial D(0,r);\\
		G_{\chi}(z)=&\,R_0(z;t,+\infty)\Big[P^{(\frac{1}{2})}(z)\big|_{V=0}\big(P^{(\infty)}(z)\big|_{V=0}\big)^{-1}\Big]R_0^{-1}(z;t,+\infty),\ \ z\in\partial D\left(\frac{1}{2},r\right)
	\end{align*}
	and for $z\in(r,\frac{1}{2}-r)\cup(r+\frac{1}{2},+\infty)$,
	\begin{equation*}
		G_{\chi}(z)=R_0(z;t,+\infty)\bigg[P^{(\infty)}(z)\big|_{V=0}\begin{pmatrix} 1 & \e^{-t(\varkappa_{_\textnormal{Ai}}+2g(z)|_{V=0})}\\ 0 & 1 \end{pmatrix}\big(P^{(\infty)}(z)\big|_{V=0}\big)^{-1}\bigg]R_0^{-1}(z;t,+\infty).	\end{equation*}
	\item As $z\rightarrow\infty$, we have $\chi(z)\rightarrow I$.
\end{enumerate}
\end{problem}
Using \eqref{e:7} together with \eqref{A:es32} we obtain the following important small-norm estimate
\begin{prop} There exists $t_0$ and $c$ positive such that
\begin{equation*}
	\|G_{\chi}(\cdot;t,\varkappa_{_\textnormal{Ai}})-I\|_{L^2\cap L^{\infty}(\Sigma_{\chi})}\leq \frac{c}{t^{\frac{1}{6}}},\ \ \ \forall\,t\geq t_0,\ \ \varkappa_{_\textnormal{Ai}}\geq\frac{2}{3}\sqrt{2}.
\end{equation*}
\end{prop}
But this implies again unique solvability of the $\chi$-RHP for $t\geq t_0,\varkappa>\frac{2}{3}\sqrt{2}$ and we have
\begin{equation*}
	\chi(z)=|s|^{-\frac{1}{4}\sigma_3}\left\{I+\mathcal{O}\left(\frac{t^{-\frac{1}{2}}}{1+|z|}\right)\right\}|s|^{\frac{1}{4}\sigma_3},\ \ z\in\mathbb{C}\backslash\Sigma_{\chi},\ \ c>0.
\end{equation*}
We summarize the relation between the nonlinear steepest descent analysis with fixed $\varkappa_{_\textnormal{Ai}}\in[\frac{2}{3}\sqrt{2},+\infty)$ and $\varkappa_{_\textnormal{Ai}}=+\infty$ in the following Corollary.
\begin{cor} There exists $t_0>0$ and $c>0$ such that
\begin{equation}\label{A:es4}
	R_0(z;t,\varkappa_{_\textnormal{Ai}})=|s|^{-\frac{1}{4}\sigma_3}\left\{I+\mathcal{O}\left(\frac{t^{-\frac{1}{2}}}{1+|z|}\right)\right\}|s|^{\frac{1}{4}\sigma_3}R_0(z;t,+\infty),\ \ \forall\, t\geq t_0,\ \ \varkappa_{_\textnormal{Ai}}\in\left[\frac{2}{3}\sqrt{2},+\infty\right)
\end{equation}
uniformly for $z\in\mathbb{C}\backslash\Sigma_{\chi}$.
\end{cor}
\subsection{Proof of expansion \ref{A:1}}\label{gderiv} We recall the differential identity \eqref{diff:1},
\begin{equation*}
	\frac{\partial}{\partial s}\ln D(J_{\textnormal{Ai}};\gamma)=\frac{\gamma}{2\pi\im}\big(X^{-1}(z)X'(z)\big)_{21}\Big|_{z\rightarrow s}
\end{equation*}
in which the derivative is taken for any fixed $\gamma\leq 1$. Repeating the steps carried out in Section \ref{asyderiv}, here however tailored to the transformation sequence
\begin{equation*}
	X(z)\mapsto T(z)\mapsto S_0(z)\mapsto R_0(z),
\end{equation*}
we derive immediately the following exact identity
\begin{equation*}
	F(J_{\textnormal{Ai}};\gamma)\equiv\frac{\partial}{\partial s}\ln D(J_{\textnormal{Ai}};\gamma)=\frac{1}{4}\gamma |s|^2-\frac{\gamma}{2}\Big\{R_0^{-1}(0)R_0'(0)\Big\}_{21}.
\end{equation*}
We can now either directly use the $R_0$-RHP to compute the remaining matrix entry or simply refer to our previous computations of Section \ref{asyderiv}. If we choose the latter we have to keep in mind the additional transformation $Q(z)\mapsto L(z)$ and tailor it to the current situation. In either case, we obtain after indefinite integration with respect to $s$,
\begin{equation}\label{A:es6}
	\ln D(J_{\textnormal{Ai}};\gamma)=\frac{s^3}{12}-\frac{1}{8}\ln|s|+\ln d_0+\mathcal{O}\left(t^{-\frac{1}{2}}\right),\ \ t\rightarrow+\infty,\ \gamma\uparrow 1:\ \ \varkappa_{_\textnormal{Ai}}\in\left[\frac{2}{3}\sqrt{2},+\infty\right)
\end{equation}
where $d_0>0$ may still depend on $\gamma$ (but not $s$) at this point. 
\begin{remark} Observe that for $\varkappa_{_\textnormal{Ai}}>\frac{2}{3}\sqrt{2}$ all power like error terms in \eqref{A:es6} are in fact $\gamma$-independent. This matches the result of Remark $7.4.$ in \cite{B}.
\end{remark}
It is clear that 
\begin{equation*}
	d_0(\gamma)=c_0\big(1+o(1)\big),\ \ \gamma\uparrow 1;\  \ \ c_0=\exp\left[-\frac{1}{24}+\z'(-1)\right]
\end{equation*}
but we would like to have a better error estimate. To this end, note that
\begin{equation*}
	\frac{\partial}{\partial\gamma}\ln D(J_{\textnormal{Ai}};\gamma)=-\textnormal{tr}\,\left((I-\gamma K_{\textnormal{Ai}})^{-1}K_{\textnormal{Ai}}\right)\Big|_{L^2(s,\infty)}=-\frac{1}{\gamma}\int_s^{\infty}R(\lambda,\lambda;t,\varkappa_{_\textnormal{Ai}})\,\d\lambda
\end{equation*}
with, (cf. \cite{CIK}, Section $2$),
\begin{equation*}
	R(\lambda,\lambda;t,\varkappa_{_\textnormal{Ai}})=\frac{\gamma}{2\pi\im}\big(X^{-1}(z;t,\varkappa_{_\textnormal{Ai}})X'(z;t,\varkappa_{_\textnormal{Ai}})\big)_{21}\Big|_{z\rightarrow\lambda},\ \ \ \ \lambda\in(s,\infty);\ \ \ \ (')=\frac{\d}{\d z}.
\end{equation*}
But from \eqref{A:es4}, we get, after tracing back the transformations,
\begin{equation*}
	\frac{\partial}{\partial\gamma}\ln D(J_{\textnormal{Ai}};\gamma)=\frac{\partial}{\partial\gamma}\ln D(J_{\textnormal{Ai}};1)+\mathcal{O}\left(t^{-\frac{1}{2}}\right)=\mathcal{O}\left(t^{-\frac{1}{2}}\right)
\end{equation*}
where we used the classical gap asymptotics from, say, \cite{CIK} in the last equality. Hence, back in \eqref{A:es6}, we deduce
\begin{equation*}
	\frac{\partial}{\partial\gamma}\ln d_0(\gamma)=\mathcal{O}\left(t^{-\frac{1}{2}}\right),
\end{equation*}
and thus 
\begin{equation*}
	\ln d_0(\gamma)\equiv\ln c_0+\mathcal{O}\left(t^{-\frac{1}{2}}\right),\ \ \ t\rightarrow+\infty,\ \gamma\uparrow 1:\ \ \varkappa_{_\textnormal{Ai}}\in\left[\frac{2}{3}\sqrt{2},+\infty\right),
\end{equation*}
which completes the proof of \eqref{A:1}.
\section{Refining Theorem $4$ in \cite{DKV}}\label{Bessapp}
This part of the appendix complements the previous one and we derive the analogue of Theorem \ref{bettererror} for the Bessel kernel. In more detail, we shall prove
\begin{theo} For any fixed $a>-1$, as $s\rightarrow+\infty,\gamma\uparrow 1$,
\begin{equation}\label{B:1}
	\det(I-\gamma K_{\textnormal{Bess}})\Big|_{L^2(0,s)}=\exp\left[-\frac{s}{4}+a\sqrt{s}\right]s^{-\frac{1}{4}a^2}\tau_a\left(1+\mathcal{O}\left(s^{-\frac{1}{2}}\right)\right),\ \ \ \tau_a=\frac{G(1+a)}{(2\pi)^{\frac{a}{2}}},
\end{equation}
uniformly for 
\begin{equation*}
	\varkappa_{_\textnormal{Bess}}=-\frac{\ln(1-\gamma)}{\sqrt{s}}\geq 2-a\frac{\ln t}{t}.
\end{equation*}
\end{theo}
Expansion \eqref{B:1} was proven in \cite{DKV} in case $\gamma=1$, i.e. $\varkappa_{_\textnormal{Bess}}=+\infty$, allowing also for $a\in\mathbb{C}:\,\Re a>-1$. Another proof of the gap expansion (i.e. $\varkappa_{_\textnormal{Bess}}=+\infty$) appeared in \cite{E} using operator theoretic methods valid for $|\Re a|<1$. Here we will use once more the method of integrable integral operators and RHP \ref{masterIIKS} subject to \eqref{BIIKS}.
\subsection{Comparison of nonlinear steepest descent analysis} Since the initial RHP \ref{masterIIKS} is valid for any $\gamma\leq 1$ we can again apply the first two transformations
\begin{equation*}
	Y(z)\mapsto X(z)\mapsto T(z)
\end{equation*}
as worked out in Sections \ref{Besspre:1} and \ref{Besspre:2}. After that we set $V\equiv 0$ in \eqref{g:2n} and define
\begin{equation*}
	S_0(z)=T(z)\e^{t(g(z)|_{V=0})\sigma_3},\ \ \ \ z\in\mathbb{C}\backslash\Sigma_T.
\end{equation*}
This leads us to the problem below
\begin{problem} The normalized function $S_0(z)=S_0(z;s,\gamma)\in\mathbb{C}^{2\times 2}$ is characterized by the following properties
\begin{enumerate}
	\item $S_0(z)$ is analytic for $z\in\mathbb{C}\backslash(\Sigma_T\cup\{0,1\})$.
	\item The jump conditions are as follows,
	\begin{equation*}
		S_{0+}(z)=S_{0-}(z)\begin{pmatrix} 1 & 0\\ \e^{-\im\pi a}\e^{2tg(z)|_{V=0}} & 1 \end{pmatrix},\ \ z\in\widehat{\Gamma}_{2,T};\ \ \ \ S_{0+}(z)=S_{0-}(z)\begin{pmatrix} 1 & 0\\ \e^{\im\pi a}\e^{2tg(z)|_{V=0}} & 1 \end{pmatrix},\ \ z\in\widehat{\Gamma}_{4,T}
	\end{equation*}
	and on the positive real axis,
	\begin{equation*}
		S_{0+}(z)=S_{0-}(z)\begin{pmatrix}\e^{-\im\pi a} & \e^{-t(\varkappa_{\textnormal{Bess}}+2g(z)|_{V=0})}\\ 0 & \e^{\im\pi a} \end{pmatrix},\ \ z\in\widehat{\Gamma}_{1,T};\ \ \ \ S_{0+}(z)=S_{0-}(z)\begin{pmatrix} 0 & 1\\ -1 & 0 \end{pmatrix},\ \ z\in\widehat{\Gamma}_{3,T}.
	\end{equation*}
	\item The singular behavior is identical to the one stated in RHP \ref{g2RHP}, subject to the replacement $V=0$.
	\item The normalization at $z=\infty$ is identical to the one stated in RHP \ref{g2RHP}, also here subject to the replacement $V=0$.
\end{enumerate}
\end{problem}
We have, as before, with $0<r<\frac{1}{4}$ fixed,
\begin{equation*}
	\Re\big(g(z)|_{V=0}\big)<0,\ \ \ \ z\in\big(\widehat{\Gamma}_{2,T}\cup\widehat{\Gamma}_{4,T}\big)\backslash D(1,r),
\end{equation*}
but since
\begin{equation*}
	g(z)|_{V=0}\geq -1,\ \ \ z\in(0,1);\ \ \ \ g\left(0\right)|_{V=0}=-1,
\end{equation*}
we now also have that for given $a>-1$ there exist $t_0>0,0<r<\frac{1}{4}$ such that for $z\in(0,1)\backslash D(0,r)$ and $t\geq t_0$,
\begin{equation*}
	\varkappa_{_\textnormal{Bess}}+2g(z)|_{V=0}\geq 2-2\sqrt{1-z}-a\frac{\ln t}{t}\geq\delta>0.
\end{equation*}
This means we only need a simplified local analysis near the origin $z=0$. In fact, take $k=0$ in \eqref{lagbare}
\begin{equation*}
	L(\z;a)\big|_{k=0}=\begin{pmatrix} 1 & \frac{1}{2\pi\im}\int_0^{\infty}t^a\e^{-t}\frac{\d t}{t-\z}\\ 0 & 1\end{pmatrix}(\e^{-\im\pi}\z)^{\frac{a}{2}\sigma_3}\e^{-\frac{1}{2}\z\sigma_3},\ \ \ \z\in\mathbb{C}\backslash[0,\infty),
\end{equation*}
and adjust \eqref{Lagpara} to $V=0=\alpha=k$. The same adjustments have to be made for the local model functions in \eqref{Besselout} and \eqref{Hankelpara}.
\subsection{Small norm theorem} We are left with the final transformation: put
\begin{equation}\label{bessg2:1}
	R_0(z)=\begin{pmatrix} 1 & 0\\ \omega & 1\end{pmatrix}S_0(z)\begin{cases} \big(P^{(0)}(z)\big)^{-1}\big|_{V=0=\alpha=k},&|z|<r\\\big(P^{(1)}(z)\big)^{-1}\big|_{V=0=\alpha=k},&|z-1|<r\\ \big(P^{(\infty)}(z)\big)^{-1}\big|_{V=0=\alpha=k},&|z|>r,\ \ |z-1|>r\end{cases}   ;\ \ \ \omega=-s^{\frac{1}{2}}\big(N(S_{\infty}-a\sigma_3)N^{-1}\big)_{21}
\end{equation}
with $0<r<\frac{1}{4}$ and
\begin{equation*}
	N=\frac{1}{\sqrt{2}}\begin{pmatrix} 1 & -1\\ 1 & 1\end{pmatrix}\e^{-\im\frac{\pi}{4}\sigma_3}.
\end{equation*}
The jump conditions of $R_0(z)$ are posed on the contour shown in Figure \ref{figure8} when $\varkappa_{_\textnormal{Bess}}+a\frac{\ln t}{t}\in[2,\infty)$. For $\varkappa_{_\textnormal{Bess}}=+\infty$ we have no jumps on $\partial D(0,r)\cup(r,1-r)$. In either case we easily derive the following estimate for
\begin{equation*}
	\widehat{G}_{R_0}(z;s,\gamma;a)=s^{\frac{1}{4}\sigma_3}G_{R_0}(z;s,\gamma;a)s^{-\frac{1}{4}\sigma_3},\ \ \ z\in\Sigma_{R_0}
\end{equation*}
with $G_{R_0}(z)$ denoting the jump matrix associated with \eqref{bessg2:1}.
\begin{prop} Given $a>-1$ there exist $t_0=t_0(a),v_0=v_0(a)$ and $c=c(a)$ such that
\begin{equation*}
	\|\widehat{G}_{R_0}(\cdot;s,\gamma;a)-I\|_{L^2\cap L^{\infty}(\Sigma_{R_0})}\leq\frac{c}{t},\ \ \ \forall\,t\geq t_0,\ v=-\ln(1-\gamma)\geq v_0:\ \ v\geq 2t-a\ln t.
\end{equation*}
\end{prop}
From this estimate we obtain unique solvability of the $R_0$-RHP in $L^2(\Sigma_{R_0})$ and the solution satisfies
\begin{equation}\label{bc:1}
	a>-1:\ \ \ \ s^{\frac{1}{4}\sigma_3}R_0(z)s^{-\frac{1}{4}\sigma_3}=I+\mathcal{O}\left(\frac{t^{-1}}{1+|z|}\right),\ \ t\rightarrow\infty,\ \gamma\uparrow 1:\ \ \varkappa_{_\textnormal{Bess}}\geq2-a\frac{\ln t}{t}
\end{equation}
uniformly for $z\in\mathbb{C}\backslash\Sigma_{R_0}$. The required comparison follows from writing $R_0=R_0(z;t,\varkappa_{_\textnormal{Bess}};a)$ and analyzing 
\begin{equation*}
	\chi(z;t,\varkappa_{_\textnormal{Bess}};a)=R_0(z;t,\varkappa_{_\textnormal{Bess}};a)\big(R_0(z;t,+\infty;a)\big)^{-1},\ \ z\in\mathbb{C}\backslash\Sigma_{R_0};\  t\geq t_0,\ \ \varkappa_{_\textnormal{Bess}}+a\frac{\ln t}{t}\in[2,+\infty)\ \ \textnormal{fixed}.
\end{equation*}
\begin{problem} Determine $\chi(z)=\chi(z;t,\varkappa_{_\textnormal{Bess}};a)\in\mathbb{C}^{2\times 2}$ such that
\begin{enumerate}
	\item $\chi(z)$ is analytic for $z\in\mathbb{C}\backslash\Sigma_{\chi}$ with square integrable boundary values on the jump contour
	\begin{equation*}
		\Sigma_{\chi}=\partial D(0,r)\cup\partial D(1,r)\cup(r,1-r)
	\end{equation*}
	shown in Figure \ref{figure9}.
	\begin{figure}[tbh]
\begin{center}
\resizebox{0.25\textwidth}{!}{\includegraphics{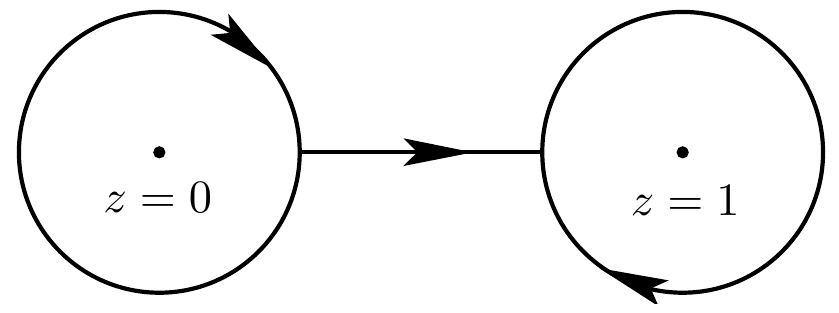}}
\caption{The oriented jump contours for the ratio function $R_0(z)$ in the complex $z$-plane.}
\label{figure9}
\end{center}
\end{figure}
	\item We have jump conditions $\chi_+(z)=\chi_-(z)G_{\chi}(z;t,\varkappa_{_\textnormal{Bess}};a)$ with 
	\begin{align*}
		G_{\chi}(z)=&\,R_{0}(z;t,+\infty;a)\Big[P^{(0)}(z)\big|_{V=0}\big(P^{(\infty)}(z)\big|_{V=0}\big)^{-1}\Big]R_{0}^{-1}(z;t,+\infty;a),\ \ z\in\partial D(0,r);\\
		G_{\chi}(z)=&\,R_{0-}(z;t,+\infty;a)\Big[P^{(1)}(z)\big|_{V=0}\big(P^{(1)}(z)\big|_{\substack{V=0\\ \varkappa_{_\textnormal{Bess}}=+\infty}}\big)^{-1}\Big]R_{0-}^{-1}(z;t,+\infty;a),\ \ z\in\partial D(1,r);
	\end{align*}
	and for $z\in(r,1-r)$,
	\begin{equation*}
		G_{\chi}(z)=R_0(z;t,+\infty;a)\Big[P^{(\infty)}(z)\big|_{V=0}\begin{pmatrix}1 & \e^{-\im\pi a}\e^{-t(\varkappa_{_\textnormal{Bess}}+2g(z)|_{V=0})}\\ 0 & 1\end{pmatrix}\big(P^{(\infty)}(z)\big|_{V=0}\big)^{-1}\Big]R_0^{-1}(z;t,+\infty;a).\end{equation*}
	\item Near $z=\infty$, the function behaves as $\chi(z)\rightarrow I$.
\end{enumerate}
\end{problem}
Now we use \eqref{bc:1} and derive for the conjugated jump matrix
\begin{equation*}
	\widehat{G}_{\chi}(z;t,\varkappa_{_\textnormal{Bess}};a)=t^{\frac{1}{2}\sigma_3}G_{\chi}(z;t,\varkappa_{_\textnormal{Bess}};a)t^{-\frac{1}{2}\sigma_3},\ \ z\in\Sigma_{\chi}
\end{equation*}
the estimate
\begin{prop} For any given $a>-1$ there exist positive $t_0=t_0(a)$ and $c=c(a)$ such that
\begin{equation*}
	\|\widehat{G}_{\chi}(\cdot;t,\varkappa_{_\textnormal{Bess}};a)-I\|_{L^2\cap L^{\infty}(\Sigma_{\chi})}\leq\frac{c}{t},\ \ \ \forall\, t\geq t_0,\ \ \varkappa_{_\textnormal{Bess}}\geq2-a\frac{\ln t}{t}.
\end{equation*}
\end{prop}
This provides in turn the analogue of \eqref{A:es4} for the Bessel kernel and the remaining argument is just as in Section \ref{gderiv}.
\end{appendix}


\begin{thebibliography}{100}

\bibitem{BBD} J. Baik, R. Buckingham, J. DiFranco, Asymptotics of Tracy-Widom distributions and the total integral of a Painlev\'e II function, {\it Commun. Math. Phys.} {\bf 280}, 463-497 (2008)
\bibitem{BW} E. Basor, H. Widom, Toeplitz and Wiener-Hopf determinants with piecewise continuous symbols, {\it Journal of Functional Analysis}, {\bf 50}, 387-413 (1983)
\bibitem{BL} M. Bertola, S. Lee, First colonization of a spectral outpost in random matrix theory, {\it Constr. Approx.} {\bf 30}, 225-263 (2009)
\bibitem{BT} M. Bertola, A. Tovbis, Universality in the profile of the semiclassical limit solutions to the focusing nonlinear Schr\"odinger equation at the first breaking curve, {\it Int. Math. Res. Not.} {\bf 2010}, 2119-2167 (2010)
\bibitem{BK} P. Bleher, A. Kuijlaars, Large $n$ limit of Gaussian random matrices with external source, part III: Double scaling limit, {\it Commun. Math. Phys.} {\bf 270}, 481-517 (2007)
\bibitem{BDIK} T. Bothner, P. Deift, A. Its, I. Krasovsky, On the asymptotic behavior of a log gas in the bulk scaling limit in the presence of a varying external potential I, {\it Commun. Math. Phys.} {\bf 337}, 1397-1463 (2015)
\bibitem{B} T. Bothner, Transition asymptotics for the Painlev\'e II transcendent, {\it Duke Math. J.} to appear, arXiv:1502.03402v2 (2015)
\bibitem{BI1} T. Bothner, A. Its, Asymptotics of a Fredholm determinant corresponding to the first bulk critical universality class in random matrix models, {\it Commun. Math. Phys.} {\bf 328}, 155-202 (2014)
\bibitem{BI2} T. Bothner, A. Its, Asymptotics of a cubic sine kernel determinant, {\it St. Petersburg Math. J.} {\bf 26}, 22-92 (2014)
\bibitem{BB} M. Bowick, E. Br\'ezin, Universal scaling of the tail of the density of eigenvalues in random matrix models, {\it Phys. Letts.} {\bf B268}, 21-28 (1991)
\bibitem{BM} R. Buckingham, P. Miller, Large-degree asymptotics of rational Painlev\'e-II functions. II. {\it Nonlinearity} {\bf 28}, 1539-1596 (2015)
\bibitem{BuBu} A. Budylin, V. Buslaev, Quasiclassical asymptotics of the resolvent of an integral convolution operator with a sine kernel on a finite interval, {\it Algebra i Analiz} {\bf 7}, no. 6, 79-103 (1995)
\bibitem{C1} T. Claeys, Pole-free solutions of the first Painlev\'e hierarchy and non-generic critical behavior for the KdV equation, {\it Physica D} {\bf 241} (2012) 2226-2236
\bibitem{C} T. Claeys, Birth of a cut in unitary random matrix ensembles, {\it Int. Math. Res. Not. IMRN} (2008) no.6, Art. ID rnm166, 40pp.
\bibitem{CIK} T. Claeys, A. Its, I. Krasovsky, Higher-Order analogues of the Tracy-Widom distribution and the Painlev\'e II hierarchy, {\it Commun. Pure. Appl. Math.} {\bf 63}(3), 362-412 (2010) 
\bibitem{CV} T. Claeys, M. Vanlessen, Universality of a double scaling limit near singular edge points in random matrix models, {\it Commun. Math. Phys.} {\bf 273}, 499-532 (2007)
\bibitem{DIZ} P. Deift, A. Its, X. Zhou, A Riemann-Hilbert approach to asymptotic problems arising in the theory of random matrix models, and also in the theory of integrable statistical mechanics, {\it Ann. Math.} {\bf 146}, 149-235 (1997)
\bibitem{DIKZ} P. Deift, A. Its, I. Krasovsky, X. Zhou, The Widom-Dyson constant for the gap probability in random matrix theory, {\it J. Comput. Appl. Math.} {\bf 202}, no. 1, 26-47 (2007)
\bibitem{DIK} P. Deift, A. Its, I. Krasovsky, Asymptotics of the Airy-kernel determinant, {\it Commun. Math. Phys.} {\bf 278}, 643-678 (2008)
\bibitem{DKV} P. Deift, I. Krasovsky, J. Vasilevska, Asymptotics for a determinant with a confluent hypergeometric kernel, {\it International Mathematics Research Notices}, Vol 2001, No. 9, 2117-2160 (2010)
\bibitem{DZ} P. Deift, X. Zhou, A steepest descent method for oscillatory Riemann-Hilbert problems. Asymptotics for the mKdV equation, {\it Ann. of Math.} {\bf 137} (1993) 295-368.
\bibitem{DKZ} S. Delvaux, A. Kuijlaars, L. Zhang, Critical behavior of non-intersecting Brownian motions at a tacnode, {\it Comm. Pure Appl. Math.} {\bf 64}, 1305-1383 (2011)
\bibitem{D} F. Dyson, Statistical theory of energy levels of complex systems, I, II, and III. {\it J. Math. Phys.} {\bf 3}, 140-156, 157-165, 166-175 (1962)
\bibitem{E0} T. Ehrhardt, Dyson's constant in the asymptotics of the Fredholm determinant of the sine kernel, {\it Commun. Math. Phys.} {\bf 262}, 317-341 (2006)
\bibitem{E} T. Ehrhardt, The asymptotics of a Bessel-kernel determinant which arises in Random Matrix Theory, {\it Advances in Mathematics} {\bf 225} (2010) 3088-3133
\bibitem{F} P. Forrester, The spectrum edge of random matrix ensembles, {\it Nucl. Phys.} {\bf B402} 709-728 (1993)
\bibitem{Fu} W. Fuchs, On the eigenvalues of an integral equation arising in the theory of band-limited signals, {\it J. Math. Anal. and Applic.} {\bf 9}, 317-330 (1964)
\bibitem{IIKS} A. Its, A. Izergin, V. Korepin, N. Slavnov, Differential equations for quantum correlation functions. {\it Int. J. Mod. Phys. B} {\bf 4}, 1003-1037 (1990)
\bibitem{KKMST} N. Kitanine, K. Kozlowski, J. Maillet, N. Slavnov, V. Terras, Riemann-Hilbert approach to a generalised sine kernel and applications, {\it Commun. Math. Phys.} {\bf 291}, 691-761 (2009)
\bibitem{K} I. Krasovsky, Gap probability in the spectrum of random matrices and asymptotics of polynomials orthogonal on an arc of the unit circle, {\it Int. Math. Res. Not.} {\bf 2004}, 1249-1272 (2004)
\bibitem{M} M. Mehta, Random matrices. Third edition. Elsevier/Academic Press, Amsterdam, 2004
\bibitem{Mo} G. Moore, Matrix models of 2D gravity and isomonodromic deformation, {\it Prog. Theor. Physics Suppl. No} {\bf 102}, 255-285 (1990)
\bibitem{NIST} NIST Digital Library of Mathematical Functions, http://dlmf.nist.gov
\bibitem{P} L. Pastur, On the universality of the level spacing distribution for some ensembles of random matrices, {\it Letts. Math. Phys.} {\bf 25}, 259-265 (1992)
\bibitem{PS} H. Pollak, D. Slepian, Prolate spheroidal wave functions, Fourier analysis and uncertainty-I. {\it Bell Systems Tech. J.} {\bf 40}, 43-64 (1961) 
\bibitem{TW} C. Tracy, H. Widom, Level-spacing distributions and the Airy kernel, {\it Commun. Math. Phys.} {\bf 159}, 151-174 (1994)
\bibitem{TW2} C. Tracy, H. Widom, Level-spacing distributions and the Bessel kernel, {\it Commun. Math. Phys.} {\bf 161}, 289-309 (1994)
\bibitem{TWP} C. Tracy, H. Widom, The Pearcey process, {\it Commun. Math. Phys.} {\bf 263}, 381-400 (2006)
\bibitem{WN} M. Wadati, T. Nagao, Correlation functions of random matrix ensembles related to classical orthogonal polynomials, {\it J. Phys. Soc. Japan} {\bf 60}, 3298-3322 (1991)
\bibitem{W2} H. Widom, Asymptotics for the Fredholm determinant of the sine kernel on a union of intervals, {\it Commun. Math. Phys.} {\bf 171}, 159-180 (1995)
\end{thebibliography}
\end{document}